\documentclass[11pt,letterpaper]{article}
\usepackage[T1]{fontenc}
\usepackage{lmodern}
\usepackage[margin=1in]{geometry}
\usepackage{amsmath,amssymb,amsthm,mathtools}
\usepackage{booktabs,array,tabularx,enumitem,float,microtype}
\usepackage[ruled,vlined,linesnumbered]{algorithm2e}
\usepackage[font=small,labelfont=bf,skip=8pt]{caption}
\usepackage{placeins}
\usepackage[colorlinks=true,linkcolor=blue,citecolor=blue,urlcolor=blue,
            linktoc=all,bookmarksnumbered]{hyperref}
\hypersetup{
 pdftitle={Provable Subexponential Algorithms for NIST Third-Round Lattice Families},
 pdfsubject={Gaussian sampling and secret recovery}}
\setlist{topsep=5pt,itemsep=3pt,parsep=0pt,leftmargin=*}
\newcolumntype{Y}{>{\raggedright\arraybackslash}X}
\newcolumntype{P}[1]{>{\raggedright\arraybackslash}p{#1}}
\newcommand{\papertablesetup}{\small\renewcommand{\arraystretch}{1.18}}
\newtheorem{theorem}{Theorem}[section]
\newtheorem{lemma}[theorem]{Lemma}
\newtheorem{proposition}[theorem]{Proposition}
\newtheorem{corollary}[theorem]{Corollary}
\theoremstyle{definition}
\newtheorem{definition}[theorem]{Definition}
\theoremstyle{remark}
\newtheorem{remark}[theorem]{Remark}
\numberwithin{equation}{section}
\allowdisplaybreaks[2]

\title{Provable Subexponential Algorithms for\\NIST Third-Round Lattice Families}

\author{
Yiming Gao\textsuperscript{1} \quad
Xuyuan Han\textsuperscript{1} \quad
Honggang Hu\textsuperscript{1,2}\\[0.8em]
\small \textsuperscript{1}School of Cyber Science and Technology,\\
University of Science and Technology of China\\
\small \texttt{qw1234567@mail.ustc.edu.cn, hanxuyuan@mail.ustc.edu.cn}\\[0.15em]
\small \textsuperscript{2}Hefei National Laboratory, Hefei, China\\
\small \texttt{hghu2005@ustc.edu.cn}
}

\date{}
\begin{document}
\maketitle
\begin{abstract}
We give provable classical subexponential algorithms for secret recovery
in growing parameter families associated with NIST third-round lattice
candidates. For the Kyber/ML-KEM, FrodoKEM, SABER, NTRU LPRime, and
Dilithium/ML-DSA families
studied here, polynomial moduli and polylogarithmic coefficient scales
yield recovery of the short secret component in expected time and space
\(2^{(1/2+o(1))n/\ln\ln n}\). 

For noisy or rounded linear relations, we exploit an exact gap in the squared 
Euclidean norm of a comparison vector defined by each coordinate guess. One Gaussian 
list suffices to identify every secret coordinate by binary search, without enumerating the others. We establish the
required sampling guarantees through new geometric bounds for structured
public operators over prime and power of two moduli. The construction
builds on the Wagner-style Gaussian sampling framework of Ducas, Engelberts, and Loyer (CRYPTO 2025)
and the low-error decision-LWE algorithm of Han, Gao, and Hu (2026).

For quotient relations of NTRU type, we develop an affine slice search: fixing
coordinates restricts candidate pairs to slices of the public lattice,
and their estimated Gaussian masses guide the choice of each next
coordinate. In the stated modulus window, Falcon's key generation
quality condition supplies the required mass bound. The algorithm then
recovers an equivalent signing key with high probability in time and space
\(2^{O(n/\ln\ln n)}\). The same search recovers the short key
core for cyclic NTRU-HPS/HRSS.

Together, these results give subexponential algorithms for problem families associated with all seven
NIST third-round lattice candidates. Despite the subexponential complexity, our
results do not establish a reduction in the concrete security of the currently specified parameter
sets.

\end{abstract}

\vfill
\noindent\begin{minipage}{\textwidth}
\footnotesize
\rule{0.4\linewidth}{0.4pt}\par
\textbf{AI disclosure.}
Building on our earlier work~\cite{HGH}, the authors used large language
model tools through Codex (GPT-5.6-sol and GPT-6-astra)
to explore extensions to growing parameter families associated with
Kyber. Further extensions to other NIST third-round lattice candidates
were explored with assistance from GPT-5.6-sol.
The research proceeded through iterative exchanges: the authors
contributed mathematical ideas and approaches, while the tools also
proposed ideas and generated candidate mathematical derivations.
The authors then organized these approaches into the two frameworks presented in this paper, with further assistance
from GPT-6-astra in organizing and drafting the manuscript.
The final manuscript was substantially verified and revised by the authors, who
remain responsible for all mathematical claims, proofs, references,
and conclusions.
\end{minipage}
\clearpage
\setcounter{tocdepth}{2}
\begingroup
\linespread{0.94}\selectfont
\tableofcontents
\endgroup
\clearpage
\section{Introduction}
\label{ui:introduction}
Small secrets and errors are central to efficient cryptography based on
lattices. They simplify arithmetic, reduce key and ciphertext sizes,
and make decryption reliable. These advantages explain their use in
many lattice candidates considered in NIST's standardization
process~\cite{KyberSpec21,FrodoSpec20,SaberSpec,NTRUPrimeSpec,FIPS204,FalconSpec,NTRUSpec}.
The usual estimates of concrete security account for coefficient size
through heuristic models of lattice reduction, sieving, and hybrid
attacks~\cite{APS15}. In these models, the cost of the lattice reduction
step is typically modeled as exponential in its block size.
This raises an asymptotic question: can narrow coefficient distributions
lead to provable recovery in time subexponential in the dimension?

For LWE, small errors have long been known to permit subexponential
algorithms. A central question is how many samples these algorithms
need. Kirchner and Fouque~\cite{KF15} obtained subexponential running
times with a refined BKW algorithm for sufficiently small secrets and
errors. Their claim that a linear number of original samples suffices
relied on sample amplification. Herold, Kirshanova, and May~\cite{HKM18}
identified a gap in that argument and proved the required sample
amplification guarantee from $\Theta(n\ln n)$ original samples.
Algebraic methods provide another route.
Arora and Ge~\cite{AG11} gave a provable subexponential algorithm for
sufficiently small Gaussian errors, using a subexponential number of
samples. Subsequent work reduced the sample requirements for binary
errors~\cite{ACFP14,STA20,Ste24}, including heuristic estimates with
linear samples and a provable result with linear samples for a biased
binary error distribution~\cite{STA20}.

Ducas, Engelberts, and Loyer~\cite{DEL25} avoided sample amplification
by developing a discrete Gaussian sampler based on Wagner's
algorithm~\cite{Wag02}. Their method gave a provable subexponential
algorithm for $\mathrm{SIS}^{\infty}$ with only slightly more variables
than equations. Applying this sampler to LWE required further control
of its accuracy and the dependencies among its outputs.
In our earlier work~\cite{HGH}, we used random auxiliary lattices and
a direct analysis of the dependent Gaussian samples to obtain a
subexponential algorithm for decision LWE. For polynomial prime
moduli and sufficiently small discrete Gaussian errors, the algorithm
allows an arbitrary secret distribution and uses
$n+\Theta(n/\ln\ln n)$ original samples, with expected time and space
$2^{O(n/\ln\ln n)}$. However, this result still uses more than $n$
samples and does not cover structured LWE, such as the module LWE
underlying ML-KEM.

We build on these sampling methods to develop a general framework
for recovering small secrets from noisy or rounded linear relations.
The framework uses the small size of both the secret and the error
to turn each coordinate guess into a measurable difference in squared
norm. A single Gaussian list then identifies every secret coordinate.
We separate this recovery procedure from the geometric properties of
the public operator and prove the bounds needed for structured
instances. This gives recovery from exactly $n$ coefficient
observations in the linear families studied here, including rounded
products and compressed public keys.

We also develop a framework for finding valid short pairs from public
quotients of NTRU type. Together, the two frameworks give provable
subexponential algorithms for growing parameter families associated
with all seven lattice candidates in the third round of NIST's
standardization process. Under the conditions stated below, the
algorithms are classical and randomized, and succeed with high
probability in expected time and space $2^{O(n/\ln\ln n)}$, where $n$
is the relevant coefficient dimension.
\subsection{Contributions}
\label{ui:contributions}

We develop two algorithmic frameworks and the geometric estimates
needed to apply them. The contributions have four parts.

\paragraph{Generic recovery from low energy linear relations.}
A guess for one secret coordinate changes the associated error vector's squared
norm by an exact integer gap. All other secret coordinates contribute
a common offset, so comparing the candidate scores identifies the
coordinate without enumerating the remaining secret.
A single Gaussian list estimates these scores simultaneously for
every coordinate and every admissible public label.
The resulting recovery theorem separates the coordinate tests and
list construction from the algebraic structure of the public operator.

\paragraph{Prefix geometry for structured public operators.}
We prove the smoothing estimates required by the Gaussian list for
scalar matrices over prime fields, finite field multipliers,
prime cyclotomic modules, and operators modulo powers of two.
These estimates connect the generic theorem to the five linear
families. Uniform images, ideal decompositions, and additive order
layers provide the appropriate mass bounds in each algebraic setting,
including coefficient prefixes that end within a polynomial block.

\paragraph{Short pair search from public quotients.}
For a public relation $hf=g\pmod q$, Gaussian averages estimate
the mass in successive affine slices of the public lattice.
The algorithm uses these masses to choose coordinates until it reaches
a short pair passing the public algebraic checks.
A bound on the total Gaussian mass controls the probability of
following a valid pair's path. This framework gives the Falcon and
cyclic NTRU results, with the required geometry supplied by key
quality or by spectral and saturation estimates.

\paragraph{Applications to seven growing families.}
We instantiate the two frameworks for seven families associated with
NIST's third-round lattice candidates.
Kyber and Dilithium were subsequently standardized as ML-KEM and
ML-DSA, respectively~\cite{FIPS203,FIPS204}.
For the five linear families, polynomial moduli and polylogarithmic
coefficient scales give subexponential recovery. Joint secret and error
energy controls the linear recovery cost, while the Gaussian mass of the
public lattice controls the short pair search probability.

Constructing the Gaussian lists is the expensive step in both
frameworks. Their construction combines the Gaussian sampling tools
of~\cite{DEL25,HGH,ADS15} with the geometric estimates proved here.
We first define the seven growing families and describe the related
SIS-infinity result of DEL, then state the general recovery results
and their seven corollaries.

\subsection{The NIST Third-Round Lattice Families}
\label{ui:families}

Table~\ref{ui:family-table} lists the seven NIST third-round lattice
candidates~\cite{NISTIR8413} and their public relations, alongside the
SIS-infinity problem studied by DEL.
We define the growing families through their public operators,
coefficient distributions, and parameter growth. Their public relations
and coefficient profiles follow the cited specifications.
Our NTRU Prime application concerns its NTRU LPRime branch.
We use boldface for matrices and module or coefficient vectors, and
ordinary italic symbols for individual ring elements and coefficients.
The map $\operatorname{coeff}$ converts ring objects to coefficient
vectors. We specify integer representatives where needed.
Throughout, lengths are Euclidean lengths of coefficient vectors.
A polylogarithmic quantity is bounded by a fixed power of $\ln n$.
Parameters are local to each family. For NTRU LPRime, $p$ denotes
the ring degree. For ML-DSA, $d$ denotes the number of dropped
public key bits, as in FIPS~204.

\begin{table}[!htbp]
\centering
\papertablesetup
\setlength{\tabcolsep}{5pt}
\caption{Results and prior work associated with the seven NIST
third-round lattice candidates. Unrounded linear and quotient relations
are interpreted modulo $q$. Rounded public outputs use the conventions
of the corresponding specifications. For SABER, the output is modulo
$p$. Module and quotient relations use the schemes' ring notation.}
\label{ui:family-table}
\begin{tabularx}{\linewidth}{@{}P{.24\linewidth}Y P{.23\linewidth}>{\centering\arraybackslash}p{.14\linewidth}@{}}
\toprule
Family & Relation & Structure & Modulus type\\
\midrule
\multicolumn{4}{@{}l}{\textit{This work}}\\
\addlinespace[3pt]
FrodoKEM
& $\mathbf B=\mathbf A\mathbf S+\mathbf E$
& Unstructured matrix
& Power of two\\
\addlinespace[3pt]
Kyber/ML-KEM
& $\mathbf t=\mathbf A\mathbf s+\mathbf e$
& Cyclotomic module
& Prime\\
\addlinespace[3pt]
SABER
& $\mathbf b=\operatorname{Round}_{q\to p}(\mathbf A^T\mathbf s)$
& Cyclotomic module
& Power of two\\
\addlinespace[3pt]
NTRU LPRime
& $A=\operatorname{Round}(aG)$
& Extension field
& Prime\\
\addlinespace[3pt]
Dilithium/ML-DSA
& $\mathbf A\mathbf s_1+\mathbf s_2=2^d\mathbf t_1+\mathbf t_0$
& Cyclotomic module
& Prime\\
\addlinespace[3pt]
Falcon
& $h=gf^{-1}$
& Cyclotomic ring
& Prime\\
\addlinespace[3pt]
NTRU-HPS/HRSS
& $hf=3g$
& Cyclic ring
& Power of two\\
\midrule
\multicolumn{4}{@{}l}{\textit{Prior work}~\cite[Theorem~4.1]{DEL25}\textsuperscript{1}}\\
\addlinespace[3pt]
Dilithium/ML-DSA
& Scalar $\mathrm{SIS}^{\infty}$
& Unstructured matrix
& Prime\\
\bottomrule
\end{tabularx}
\par\smallskip
\begin{minipage}{\linewidth}
\footnotesize
\textsuperscript{1}DEL25 gives a heuristic cost analysis for Dilithium
in Appendix~A.1.
\end{minipage}
\end{table}

For the five linear families, public matrices or multipliers are
sampled uniformly in the algebraic spaces specified below. Their
relations have the form
$\mathbf y=\mathbf M\mathbf s+\mathbf e\pmod q$, where $n$ counts
the coefficients of the integer secret. Module ranks are fixed.
For Kyber/ML-KEM and SABER, the degree $d$ grows through powers of
two and $n=kd$. For ML-DSA, the degree $m$ grows through powers of
two and $n=\ell m$ counts the coefficients of
$\mathbf s_1$. For quotient relations, $m$ is the number of coefficients
in each ring element: $m=n$ for Falcon and $m=N$ for cyclic NTRU.
We choose all public parameters and energy bounds
before drawing the public operator.
The coefficient support bounds below may depend on $n$, and constant
bounds are included.

\paragraph{FrodoKEM.}
\label{ui:frodo-family}
For the FrodoKEM relation~\cite{FrodoSpec20}, let $\mathbf A$ be uniform in
$(\mathbb Z/q\mathbb Z)^{n\times n}$, and take
\[
 q=2^{\lceil\kappa\log_2 n\rceil},\qquad \kappa>1/2.
\]
The public matrix $\mathbf B=\mathbf A\mathbf S+\mathbf E$ contains
polynomially many prescribed columns. Fix a polylogarithmic $V$ such
that, with probability $1-o(1)$, all prescribed columns simultaneously
have joint secret and error energy at most $nV$.

If the secret and error coefficients are bounded in absolute value by
public positive integers $H_s,H_e=(\ln n)^{O(1)}$, respectively, take
\[
 B=H_s,\qquad V=\max\{1,H_s^2+H_e^2\}.
\]
These bounds hold deterministically for every prescribed column.
The finite sampling tables in FrodoKEM provide such support bounds.
All prescribed columns share the same public matrix.

\paragraph{Kyber/ML-KEM.}
\label{ui:kyber-family}
For Kyber/ML-KEM~\cite{KyberSpec21,FIPS203}, let
\[
 R=\mathbb Z[X]/(X^d+1),\qquad R_q=R/qR,\qquad
 \mathbf A\leftarrow\operatorname{Unif}(R_q^{k\times k}),\qquad
 \mathbf t=\mathbf A\mathbf s_R+\mathbf e_R\pmod q.
\]
Here $\mathbf s_R,\mathbf e_R\in R^k$ are integer module vectors and
$\mathbf t\in R_q^k$. With $\mathbf M$ the coefficient matrix of
$\mathbf A$, set $\mathbf s=\operatorname{coeff}(\mathbf s_R)$,
$\mathbf e=\operatorname{coeff}(\mathbf e_R)$, and
$\mathbf y=\operatorname{lift}_q^+(\operatorname{coeff}(\mathbf t))$.
Take prime $q=n^{\kappa+o(1)}$ for fixed $\kappa\ge1$, with
constant splitting degree $\operatorname{ord}_{2d}(q)$.
Assume that a public polylogarithmic $V\ge1$ gives a joint secret and
error energy bound $nV$ with probability $1-o(1)$.
For centered binomial coefficients with public positive integer parameters
$\eta_s,\eta_e=(\ln n)^{O(1)}$, the respective supports are
$[-\eta_s,\eta_s]$ and $[-\eta_e,\eta_e]$. One may take
\[
 B=\eta_s,\qquad V=\max\{1,\eta_s^2+\eta_e^2\},
\]
giving a deterministic joint energy bound $nV$.
ML-KEM key generation uses $\eta_s=\eta_e=\eta_1$.
The parameter $\eta_2$ is used for encryption errors.
The $k$ ring equations give exactly $n$ coefficient observations.

\paragraph{SABER.}
\label{ui:saber-family}
For SABER~\cite{SaberSpec}, let $\mathbf A$ be uniform in $R_q^{k\times k}$, where
$R_q=(\mathbb Z/q\mathbb Z)[X]/(X^d+1)$.
Choose moduli $p\mid q$ that are powers of two, with
\[
 q=n^{\kappa+o(1)},\qquad \kappa>1/2,\qquad \Delta=q/p\ge2.
\]
The secret coefficient bound $\eta$ and the rounding step $\Delta$
are polylogarithmic. Let $\mathbf s_R\in(\mathbb Z[X]/(X^d+1))^k$,
$\mathbf s=\operatorname{coeff}(\mathbf s_R)$, and let $\mathbf M$
be the coefficient matrix of $\mathbf A^T$.
The public vector $\mathbf b\in R_p^k$ uses the biased rounding map
defined coefficientwise in~\eqref{ud:rnd:map}. Let
$\mathbf b^+=\operatorname{lift}_p^+(\operatorname{coeff}(\mathbf b))$
be the nonnegative integer coefficient lift. The public lift
$\mathbf y=\Delta\mathbf b^+$ satisfies
\[
 \mathbf y=\mathbf M\mathbf s+\mathbf e_{\rm rnd}\pmod q,\qquad
 \|\mathbf e_{\rm rnd}\|^2\le n\Delta^2/4.
\]
Thus $B=\eta$ and $V=\eta^2+\Delta^2/4$ are deterministic bounds.

\paragraph{NTRU LPRime.}
\label{ui:lprime-family}
Take supplied parameters $(p,q,w)$ satisfying the arithmetic conditions
of NTRU LPRime~\cite{NTRUPrimeSpec}, with primes $p,q$ and
$X^p-X-1$ irreducible over $\mathbb F_q$. Our growing family assumes
$q=\Theta(p)$ and uses recovery dimension $n=p$. Define
\[
 O=\mathbb Z[X]/(X^p-X-1),\qquad
 R_q=O/qO=\mathbb F_q[X]/(X^p-X-1).
\]
The public multiplier $G$ is uniform in $R_q$, a degree-$p$ extension
of the prime field $\mathbb F_q$. Let $a\in O$ be a ternary secret
of weight $w$, with $0\le w\le p$. The public product uses centered representatives and is
rounded coefficientwise to the nearest multiples of three.
Writing $A:=\operatorname{Round}(aG)\in O$ for the rounded ring
element and $\mathbf y:=\operatorname{coeff}(A)\in\mathbb Z^n$ for its integer
coefficient vector gives
\[
 \mathbf y=\mathbf M_G\mathbf a+\mathbf e_{\rm rnd}\pmod q,\qquad
 \|\mathbf a\|^2+\|\mathbf e_{\rm rnd}\|^2\le w+n.
\]
Here $\mathbf a=\operatorname{coeff}(a)$, $B=1$, and $V=1+w/n$.

\paragraph{Dilithium/ML-DSA.}
\label{ui:mldsa-family}
For ML-DSA~\cite[Algorithms~6 and~35]{FIPS204}, fix integers
$k\ge\ell\ge1$, let $m\to\infty$ through powers of two, and put
$n=\ell m$, $R=\mathbb Z[X]/(X^m+1)$, and $R_q=R/qR$.
Take prime $q\equiv1\pmod{2m}$ with $q\le n^{C_q}$ for fixed
$C_q$. Choose public integers $\eta\ge1$, $d\ge1$, with
$\eta,\Delta=2^d=(\ln n)^{O(1)}$, and draw
$\mathbf A\leftarrow\operatorname{Unif}(R_q^{k\times\ell})$.
Let $\mathbf s_1\in R^\ell$, $\mathbf s_2\in R^k$ have coefficient
absolute values at most $\eta$.
Set $\mathbf t=\mathbf A\mathbf s_1+\mathbf s_2\pmod q$ and apply
$\operatorname{Power2Round}_d$ coefficientwise to obtain
$\mathbf t_1,\mathbf t_0\in R^k$.
The public key retains $\mathbf t_1$ and discards $\mathbf t_0$, giving
\[
 \Delta\mathbf t_1=\mathbf A\mathbf s_1+\mathbf e_R\pmod q,
 \qquad \mathbf e_R=\mathbf s_2-\mathbf t_0,
 \qquad \|\mathbf e_R\|_\infty\le\eta+\Delta/2.
\]
The coefficient bounds give the parameters
\[
 B=\eta,\qquad V=\eta^2+(\eta+\Delta/2)^2.
\]

For both quotient families, the ring degree $m$ tends to infinity.
Their common modulus window is
\begin{equation}\label{ui:quotient-window}
 2m+1\le q\le m^{A_q},\qquad
 \ln(q/m)\le(\ln m)^{1-c},
 \qquad A_q>0,\quad 0<c<1,
\end{equation}
with fixed constants.

\paragraph{Falcon.}
\label{ui:falcon-family}
For Falcon~\cite{FalconSpec}, let $n\to\infty$ through powers of two,
put $R=\mathbb Z[x]/(x^n+1)$ and $R_q=R/qR$, and take prime
$q\equiv1\pmod{2n}$ in the window~\eqref{ui:quotient-window}
with $m=n$.
Fix a rational constant $C\ge1$ independent of $n$.
Falcon uses $C=117/100$.
The family consists of public keys $h\in R_q$ admitting
$f,g,F_0,G_0\in R$ with
\begin{equation}\label{ui:falcon-quality-family}
 \begin{gathered}
 f\bmod q\in R_q^\times,\qquad h=gf^{-1}\pmod q,\\
 fG_0-gF_0=q,\qquad \gamma_q(f,g)^2\le C^2q.
 \end{gathered}
\end{equation}
Here $\gamma_q$ is the key generation quality measure defined
in~\eqref{uf:falcon-quality}.
We call these public keys \emph{quality-qualified}.

\paragraph{NTRU-HPS/HRSS.}
\label{ui:ntru-family}
For cyclic NTRU-HPS/HRSS~\cite[Sections~1.3.2--1.3.3]{NTRUSpec},
let $N\to\infty$ along a supplied sequence of primes satisfying
\[
 \operatorname{ord}_N(2)=\operatorname{ord}_N(3)=N-1.
\]
Work in $R_N=\mathbb Z[x]/(x^N-1)$ and
$S_N=\mathbb Z[x]/(\Phi_N)$.
Both profiles use a power of two modulus $q=\Theta(N)$ with
$q\ge2N+1$. Their input blocks are independent.
The coefficient laws and optional HPS sorting variant are defined
in Definition~\ref{uf:ntru-proposal-laws}.

For HPS, put $w=q/8-2\le2N/3$.
Sample $f$ with independent ternary coefficients.
Sample $g$ uniformly from polynomials with exactly $w/2$ coefficients
of each sign and all remaining coefficients zero.
Both polynomials have degree at most $N-2$.
Thus $w=\Theta(N)$ while the coefficient sizes remain bounded.

For HRSS, take $q>8\sqrt2(N-1)$, for example the specification's
choice $q=2^{\lceil7/2+\log_2N\rceil}$.
Sample $f,g_0$ independently with the TernaryPlus sampler, using
blocks of length $N-1$, and set $g=(x-1)g_0$.

Condition each raw distribution on the image of $f$ in $S_N$ being
invertible modulo 2 and 3, and that of $g$ being invertible modulo 2.
The public representative satisfies $h(1)=0\pmod q$ and
$hf=3g\pmod q$ in $R_N$.
Setting $\widehat h=3^{-1}h\bmod q$ gives
$\widehat h f=g\pmod q$.

\paragraph{Prior work: scalar SIS-infinity.}
\label{ui:del-sis}
Ducas, Engelberts, and Loyer~\cite[Theorem~4.1]{DEL25} study
SIS-infinity for a uniformly random matrix over $\mathbb Z_q$:
\[
 \mathbf x\in\mathbb Z^m\setminus\{\mathbf0\},\qquad
 \mathbf A\mathbf x=\mathbf0\pmod q,\qquad
 \|\mathbf x\|_\infty\le\beta,
 \qquad \mathbf A\leftarrow\operatorname{Unif}(\mathbb Z_q^{n\times m}).
\]
It uses a polynomial prime modulus, $m=n+\omega(n/\ln\ln n)$,
and $q^{1-n/m}\ge6$, with
\[
 \beta=\frac qf\sqrt{\ln m},\qquad f>1,\qquad
 0<\varepsilon\le\frac1{mq^4},\qquad
 \frac qf\ge\sqrt{\ln(1/\varepsilon)}.
\]
For $m\le\operatorname{poly}(n)$, $q=n^{\Theta(1)}$, and admissible
polylogarithmic $f$ with $\varepsilon$ inverse polynomial in $n$, the expected
running time is $2^{O(n/\ln\ln n)}$ and success probability is
$1-n^{-\Omega(1)}$. They also develop the heuristic cost model
applied to Dilithium in Appendix~A.1.

\subsection{Low Energy Linear Recovery}

The linear recovery framework takes as input
\begin{equation}\label{ui:model}
 \mathbf y=\mathbf M\mathbf s+\mathbf e\pmod q,\qquad
 \mathbf M\in(\mathbb Z/q\mathbb Z)^{h\times n},\qquad 1\le h\le n,
\end{equation}
for an integer $q\ge2$, with public bounds
\[
 \|\mathbf s\|^2+\|\mathbf e\|^2\le nV,\qquad
 \|\mathbf s\|_\infty\le B.
\]
Let $p_E$ bound the probability that these inequalities fail.
The rational $V\ge1$ and integer
$1\le B\le\lfloor\sqrt{nV}\rfloor$ have polynomial bit length and are
fixed before the public draw. If an energy bound alone is given, take
$B=\lfloor\sqrt{nV}\rfloor$.

The candidate radius and effective depth are
\begin{equation}\label{ui:depth}
 R^2=nV+3B^2,\qquad U=R/\sqrt n,\qquad
 Q=\ln(q/\sqrt n)>0,\qquad D=\ln\frac{Q}{1+\ln U}.
\end{equation}
The radius $R$ bounds the error vector for every coordinate candidate.
The normalized radius $U$ measures the secret and error scale, while
$Q$ measures the modulus relative to $\sqrt n$.
At fixed modulus and coordinate bound, smaller joint energy gives a
smaller $U$, a larger depth $D$, and a lower recovery exponent.

Let $\mathbf M_m$ consist of the first $m$ columns of $\mathbf M$.
The public operator enters through the prefix lattices
\[
 \Lambda_m(\mathbf M)=
 \{(\mathbf u,\mathbf v)\in\mathbb Z^h\times\mathbb Z^m:
       \mathbf M_m^T\mathbf u+\mathbf v=0\pmod q\}.
\]
The sampling schedule fixes prefixes $m_i$, widths $t_i$, and a
tolerance $\varepsilon_0$ from the public parameters.
The required geometric event is
\[
 \eta_{\varepsilon_0}(\Lambda_{m_i}(\mathbf M))\le t_i/2
 \quad\text{for every scheduled prefix }i.
\]
The prefix and auxiliary smoothing bounds keep each Gaussian lifting
and pairing step close to its target distribution with controlled error.
The following is an informal form of
Theorem~\ref{ue:main}.

\begin{theorem}[Low energy linear recovery]
\label{ui:informal-recovery}
Assume $D\to\infty$ and $\ln q\le(\ln n)^{c_q}$ for a fixed $c_q$.
If the prescribed geometric event fails with probability at most
$\omega_n=o(1)$, a randomized algorithm recovers the original
integer secret from the given $h$ observations with probability at least
\[
 1-p_E-\omega_n-o(1),
\]
in expected time and space $2^{(1/2+o(1))n/D}$.
One Gaussian list serves all coordinate tests and all admissible
labels for the same public operator.
\end{theorem}

On the geometric event, recovery holds for all inputs within the energy
bound, including secrets and errors that depend on $\mathbf M$.
Thus the theorem covers both sampled errors and deterministic rounding.
For polynomial $q=n^{\kappa+o(1)}$ with fixed $\kappa>1/2$ and
polylogarithmic $U$, one has $D\sim\ln\ln n$. Write
\begin{equation}\label{ui:headline-complexity}
 T_{\rm lin}(n)=2^{(1/2+o(1))n/\ln\ln n}.
\end{equation}
The five linear families of Section~\ref{ui:families} satisfy these
conditions, and their geometric estimates give the following corollaries.

\begin{corollary}\label{ui:frodo-cor}
For the FrodoKEM family, one Gaussian list recovers all prescribed
secret columns from $\mathbf B=\mathbf A\mathbf S+\mathbf E$ with
probability $1-o(1)$, in expected time and space
$T_{\rm lin}(n)$. Here $n$ is the dimension of one column.
\end{corollary}

\begin{corollary}\label{ui:kyber-cor}
For the Kyber/ML-KEM family, the original coefficient secret is recovered
from the $k$ ring equations $\mathbf t=\mathbf A\mathbf s_R+\mathbf e_R$
with probability $1-o(1)$, in expected time and space
$T_{\rm lin}(n)$.
\end{corollary}

\begin{corollary}\label{ui:saber-cor}
For the SABER family, the original secret is recovered from the $k$
rounded ring equations with probability $1-o(1)$, in expected time
and space $T_{\rm lin}(n)$.
\end{corollary}

\begin{corollary}\label{ui:lprime-cor}
For the NTRU LPRime family, the original ternary secret $a$ is recovered
from $A=\operatorname{Round}(aG)$ with probability $1-o(1)$, in
expected time and space $T_{\rm lin}(n)$, where $n=p$.
\end{corollary}

\begin{corollary}\label{ui:mldsa-cor}
For the Dilithium/ML-DSA family, the original vector $\mathbf s_1$
is recovered from the compressed public key with probability $1-o(1)$,
in expected time and space $T_{\rm lin}(n)$, where $n=\ell m$.
\end{corollary}

The corollaries use the following geometric estimates.
FrodoKEM requires control over additive orders modulo a power of two.
Kyber/ML-KEM requires estimates
for proper ideals and prefixes that end inside a polynomial block.
SABER combines coefficient prefix analysis with the ideal structure
modulo a power of two. NTRU LPRime uses the uniform image of a
nonzero element under a random field multiplier.
For ML-DSA, fix $\ell$ rows before drawing $\mathbf A$ and restrict
the public relation to this square block. The secret and restricted
error each have $n=\ell m$ coefficients, giving joint energy at most
$nV$. With the discarded bits included in the error, the same module
estimate used for Kyber/ML-KEM applies.
Sections~\ref{up:section}--\ref{umldsa:section} prove these estimates and
give the complete application statements.

The method also applies to square ternary LWE over a uniform matrix
modulo a prime. This gives a direct comparison with May's
meet-in-the-middle attacks based on representations~\cite{May21}.
For ternary secret distributions with a fixed positive entropy rate,
May's heuristic cost estimates are exponential in $n$, whereas our recovery bound is
$T_{\rm lin}(n)=2^{o(n)}$. Both use the original $n$ equations in the
common square model. Section~\ref{ut:section} gives the precise
comparison and its parameter conditions.

\subsection{Short Pair Search from Public Quotients}

A quotient relation supplies a public lattice containing a short pair.
For a public $m\times m$ matrix $\mathbf H$, define
\[
 \Lambda_{\mathbf H}=
 \{(\mathbf u,\mathbf v)\in\mathbb Z^{2m}:
          \mathbf H\mathbf u=\mathbf v\pmod q\}.
\]
For a ring pair $(u,v)$ in either quotient family, the corresponding
coefficient vector is written
$\operatorname{coeff}(u,v)=(\mathbf u,\mathbf v)\in\Lambda_{\mathbf H_h}$,
where $\mathbf H_h$ is the coefficient matrix of multiplication by $h$.
The cyclic NTRU application instead uses $\mathbf H_{\widehat h}$.
The target is a short nonzero vector that
passes a public algebraic test.
Such search problems, specified by a ring, a public distribution, and a
length bound, have been studied as NTRU problems~\cite{PS21} and through
related module problems of rank two~\cite{FPS22}.

The algorithm chooses coordinates using Gaussian masses.
Fixing some coordinates of $\mathbf u$ leaves an affine slice of the
public lattice. Each value for the next coordinate defines a smaller
slice. The algorithm samples that value with weight proportional to
the slice's Gaussian mass, which gives greater weight to shorter
completions.

Poisson summation expresses each slice mass through a Gaussian average
on the dual of its direction lattice. Reusable dual Gaussian lists
estimate these averages throughout the search. Once $\mathbf u$ is
fixed, the centered representative of $\mathbf H\mathbf u\bmod q$
completes $\mathbf v$. Exact norm and algebraic checks select valid
short pairs.

A public bound on the lattice's total Gaussian mass gives a lower
bound on the probability of reaching each fixed valid short pair.
Repeating the search then gives high success probability.

Use the modulus window~\eqref{ui:quotient-window}, and fix the finite
schedule with
\[
 \begin{gathered}
 D_{\rm Q}\sim\ln\ln m,\qquad s_{\rm mass}^2\sim q/m^\beta,\\
 s_{\rm search}^2=c_{\rm search}qD_{\rm Q}/m,
 \qquad 1/2<\beta<1,\quad c_{\rm search}>0.
 \end{gathered}
\]
Here $s_{\rm mass}$ is the Gaussian mass width and $s_{\rm search}$
is the search width, both in the Gaussian width convention used below.
The public input supplies a rational mass bound $A_{\rm mass}$ satisfying
\[
 \begin{gathered}
 \rho_{s_{\rm mass}}(\Lambda_{\mathbf H})\le A_{\rm mass},\qquad
 \ln A_{\rm mass}=O(m/D_{\rm Q}),\\
 s_{\rm mass}^2\ge8s_{\rm search}^2.
 \end{gathered}
\]
Here $\rho_s(L)=\sum_{\mathbf z\in L}e^{-\pi\|\mathbf z\|^2/s^2}$.
A public radius $1\le B_{\rm Q}=O(\sqrt q)<q/2$ specifies the target
length. The two families take $B_{\rm Q}^2=C_{\rm len}^2q$ for a
fixed rational constant $C_{\rm len}$ depending on the family.
Numerical bounds have encodings of polynomial length.
The next statement is an informal form of
Theorem~\ref{uq:mass-search-framework}.

\begin{theorem}[Verifiable short pair search]
\label{ui:informal-quotients}
Suppose a polynomial-time verifier accepts some nonzero vector in
$\Lambda_{\mathbf H}$ of norm at most $B_{\rm Q}$.
Under the public mass bound above, a randomized algorithm finds
such a vector with probability at least $1-m^{-K}$ for every fixed
$K>0$, in time and space $2^{O(m/\ln\ln m)}$.
\end{theorem}

The algorithm uses the public matrix, length and mass bounds, and verifier.
For Falcon, the key generation quality condition supplies the Gaussian
mass bound for every quality-qualified public key.
For cyclic NTRU, spectral estimates and a bound on the logarithmic
saturation index give the required mass bound with probability $1-o(1)$
over the accepted HPS and HRSS distributions.
These bounds and the algebraic checks give the following results.

\begin{corollary}\label{ui:falcon-cor}
For the Falcon family, the algorithm outputs
$u,v,F,G\in R$ such that
\[
 \begin{gathered}
 u\bmod q\in R_q^\times,\qquad vu^{-1}=h\bmod q,\\
 uG-vF=q,\qquad \gamma_q(u,v)^2\le C^2q,
 \end{gathered}
\]
with probability at least $1-n^{-K}$ for every fixed $K>0$ and
every quality-qualified public key, over its internal randomness.
Its time and space are $2^{O(n/\ln\ln n)}$.
\end{corollary}

The output has $\|(u,v)\|_2\le C\sqrt q$ and gives an integer NTRU
basis with maximum Gram--Schmidt norm at most $C\sqrt q$.
Theorem~\ref{uf:falcon-quality-recovery} derives this guarantee directly
from the key generation quality condition.

\begin{corollary}\label{ui:ntru-cor}
For the NTRU-HPS/HRSS family, a valid short pair and its
decryption inverses are recovered with probability $1-o(1)$ over
the accepted input law and independent algorithmic randomness,
in time and space
$2^{O(N/\ln\ln N)}$.
\end{corollary}

This output is the short key core defined in
Lemma~\ref{uf:ntru-short-check}.
Theorems~\ref{uf:ntru-hrss-distributional} and~\ref{uf:ntru-hps-distributional},
and Corollary~\ref{uf:ntru-hps-tags-distributional}, give the
explicit failure bounds and their proofs.

\subsection{Consequences and Organization}

Small coefficients lower the recovery cost. For linear relations,
reducing the joint energy increases
$D$ and lowers the leading recovery exponent. Constant and
polylogarithmic coefficient scales both lead to $D\sim\ln\ln n$
in the families above with polynomial moduli.
For quotient relations, a short valid pair and a bound on the total
Gaussian mass make the pair accessible to the slice search.

The recovery exponent also measures the cost of composing a reduction
with our algorithm. If a problem with $N$ input variables is reduced
to a matching recovery law of low energy and dimension $n(N)$, the relevant
scale is $n(N)/D_{n(N)}$.
Appendix~\ref{ux:section} gives a short distributional composition
argument and, under randomized ETH, excludes matching reductions with
$n(N)/D_{n(N)}=o(N)$. It explains how lowering the target energy
restricts the allowable dimension growth.

The finite schedules give the route to concrete cost estimates for
the parameter sets considered here.
Appendix~\ref{uk:section} explains how geometric guarantees transfer
to the actual public distributions, including those produced by
seed expansion.

Section~\ref{uo:overview} develops the recovery ideas before the
formal proofs. Sections~\ref{ur:section} and~\ref{ue:section} prove
the coordinate test and construct its Gaussian list.
Sections~\ref{ut:section}--\ref{umldsa:section} verify the linear applications.
Section~\ref{uq:overview} gives the quotient search theorem and its
applications, with proofs in Appendix~\ref{uf:section}.
Appendix~\ref{ue:finite} supplies finite arithmetic and sampling
procedures.
Appendix~\ref{uc:parameters} tabulates parameter values, joint energy bounds,
and formal leading term values for five linear public families.
Section~\ref{uz:fhe} discusses the common RLWE public key relation in
homomorphic encryption.

\section{Technique Overview}
\label{uo:overview}

Our algorithm for noisy or rounded linear relations involves three tasks.
First, a coordinate guess must create a detectable difference in the
public data. Second, one Gaussian list must estimate all the resulting
scores accurately.
Third, the public lattice must support construction of that list at
subexponential cost. We explain these tasks in order, then describe
how Gaussian mass estimation leads to the quotient search.

\subsection{A coordinate guess reveals an energy gap}

Consider $\mathbf y=\mathbf M\mathbf s+\mathbf e\pmod q$.
A direct guess for part of $\mathbf s$ leaves the contribution of the
unguessed coordinates in the equation. We instead place the entire
secret into a short error vector and arrange for a guess to change just
one of its coordinates. This is possible because both $\mathbf s$
and $\mathbf e$ have small joint energy.

Append the identity matrix and define
\[
 \mathbf F=[\mathbf M^T\mid\mathbf I_n],\qquad
 \mathbf v_{i,c}=(\mathbf y,c\mathbf e_i),\qquad
 \mathbf E_{i,c}=(\mathbf e,c\mathbf e_i-\mathbf s).
\]
Here $\mathbf e_i$ is the $i$th standard basis vector.
The candidate label $\mathbf v_{i,c}$ is public, and
\[
 \mathbf v_{i,c}=\mathbf F^T\mathbf s+\mathbf E_{i,c}\pmod q.
\]
The hidden vector has the particularly simple energy
\[
 \|\mathbf E_{i,c}\|^2
 =\|\mathbf e\|^2+\sum_{j\ne i}s_j^2+(c-s_i)^2.
\]
Everything except the last term is independent of the candidate $c$.
Consequently,
\begin{equation}
 \|\mathbf E_{i,c}\|^2-\|\mathbf E_{i,s_i}\|^2=(c-s_i)^2.
 \label{uo:dual-energy-gap}
\end{equation}
The correct integer candidate uniquely minimizes the energy, and
every incorrect candidate increases it by at least one.

For example, if $s_i\in\{-1,0,1\}$, the three candidate energies are
a common unknown offset plus $(-1-s_i)^2$, $s_i^2$, and $(1-s_i)^2$.
Recovering $s_i$ therefore requires comparing three values of one
public statistic. The offset contains every other secret coordinate,
but is the same in all three comparisons.
The appended identity coordinates are chosen by the algorithm, so
this construction uses the original observations.

The energy bound controls all these candidates at once. If
$\|\mathbf s\|^2+\|\mathbf e\|^2\le nV$ and $|s_i|,|c|\le B$, then
$\|\mathbf E_{i,c}\|^2\le nV+3B^2=R^2$.
The next step is a public estimate of energy that is accurate
throughout this ball.

\subsection{From a Gaussian score to a shared list}

Set
\[
 \Lambda=\{\mathbf x\in\mathbb Z^{h+n}:
                   \mathbf F\mathbf x=0\pmod q\}.
\]
For $\mathbf x\in\Lambda$, the term containing the unknown secret
vanishes from the phase:
\[
 \cos\!\left(\frac{2\pi}{q}\langle\mathbf x,\mathbf v_{i,c}\rangle\right)
 =\cos\!\left(\frac{2\pi}{q}\langle\mathbf x,\mathbf E_{i,c}\rangle\right).
\]
Thus the public label and the hidden error vector give exactly the
same cosine. Averaging over a discrete Gaussian on $\Lambda$ turns
this equality into information about the error norm.

For $\mathbf X\sim D_{\Lambda,q/f}$, Poisson summation gives
\[
 \mathbb E\!\left[
   \cos\!\left(\frac{2\pi}{q}\langle\mathbf X,\mathbf v_{i,c}\rangle\right)
 \right]
 =\frac{\rho_f(\mathcal D-\mathbf E_{i,c})}{\rho_f(\mathcal D)},
 \qquad \mathcal D=q\Lambda^*.
\]
The zero vector in $\mathcal D$ contributes
$e^{-\pi\|\mathbf E_{i,c}\|^2/f^2}$ to the numerator.
A bound on the Gaussian mass of the nonzero vectors controls all
remaining terms. Theorem~\ref{ur:score} makes this approximation
uniform over $\|\mathbf E_{i,c}\|\le R$.
Combining it with~\eqref{uo:dual-energy-gap} gives a public score
whose maximum occurs at the correct coordinate value.

We estimate the score by averaging over a list of $N$ lattice vectors.
The list is built by lifting and combining earlier samples, so its
entries are dependent. The useful invariant is \emph{conditional
similarity}: even after fixing the other entries, each entry has a
distribution close to the target Gaussian.
A martingale concentration bound then controls the empirical cosine
average. This is the dependence analysis inherited from
DEL and HGH~\cite{DEL25,HGH}.

A single comparison would require concentration for one bounded
function. We instead take a union bound over all public labels and
coordinate candidates. Theorem~\ref{ur:shared-list} shows that one
list then answers every admissible coordinate test for the same
operator. This simultaneous event has two useful consequences.
First, it covers a rounding error that is determined by the
operator and the secret. Second, it allows the same list to serve
all secret columns of a FrodoKEM-type matrix relation.
The costly sampling is performed once, while each additional column
adds polynomial scoring work.

For a larger candidate interval, adjacent scores also give a binary
search. The identity
\[
 \|\mathbf E_{i,c+1}\|^2-\|\mathbf E_{i,c}\|^2=2(c-s_i)+1
\]
determines whether $s_i$ lies to the left or right of $c$.
The difference is a nonzero integer, so the same score accuracy
preserves its sign. Each coordinate requires $O(\log(2B+1))$
comparisons. The simultaneous event already includes these
adaptively selected candidates.

\paragraph{Comparison with provable dual attack.}
Pouly and Shen~\cite{PS24} generate a Gaussian list for a chosen block of secret coordinates and reuse it to score all candidates for that block. Their lattice vectors are orthogonal modulo \(q\) to the unguessed columns of the public matrix. This removes the contribution of the unguessed secret from the score without requiring it to be small. Moving to another block changes this lattice and requires a new list in their construction. In contrast, we exploit the small size of the entire secret by including it in the error vector. For each coordinate test, the other secret coordinates contribute a common energy term that is independent of the candidate. Our embedding uses one fixed lattice for all coordinates.

The same energy structure also makes the coordinate search efficient. The uniform score estimates preserve the ordering of adjacent candidates, so comparing \(c\) and \(c+1\) reveals whether \(s_i\le c\) or \(s_i\ge c+1\), even when both candidates are incorrect. This gives a binary search using \(O(\log(2B+1))\) comparisons per coordinate.

\subsection{Building the list via Wagner-Style Gaussian Sampling}

The list construction starts in $\mathbb Z^h$, where Gaussian sampling
is easy, and adds the $n$ secret coordinate constraints in blocks.
At each level, a temporary auxiliary lattice makes Gaussian lifting
possible. Pairing samples with the same quotient label then enforces
the new constraints. This is the Gaussian list construction of
DEL~\cite{DEL25}, with the random auxiliary lattices introduced by
HGH~\cite{HGH}.

Pairing multiplies the Gaussian width by $\sqrt2$.
Early levels therefore operate at small widths and handle relatively
few new coordinates. Later levels have larger widths and can handle
larger blocks. The central allocation question is whether these
capacities sum to $n$ before the width reaches the target value $q/f$.

At a threshold $t_i$, a list of length $N$ can handle approximately
\[
 b_i\approx\frac{\ln N}{\ln(q/t_i)}
\]
new coordinates. The thresholds grow geometrically, so the denominators
decrease in equal steps. Summing their reciprocals produces a logarithm.
More precisely, the allocation argument gives
\[
 \sum_i b_i\sim2\log_2 N\,W,\qquad
 W\sim D=\ln\frac{\ln(q/\sqrt n)}{1+\ln U},
 \qquad R=\sqrt n\,U.
\]
Taking $\log_2N\sim n/(2D)$ therefore provides enough capacity.
This accumulated capacity is the source of the depth factor in the
exponent. The number of public observations stays fixed throughout.

The choice of $f$ balances score accuracy and sampling cost.
The smallest leading score on the candidate ball is
$\exp(-\pi R^2/f^2)$, while the list has Gaussian width $q/f$.
The schedule takes $f^2$ proportional to $U^2D$ and chooses a list
long enough to resolve the coordinate gap.
The random auxiliary lattices permit the required accuracy at each
lifting step. The shifted sampler of Aggarwal, Dadush, and
Stephens-Davidowitz~\cite{ADS15} implements those steps, and the block
dimensions keep its extra cost below the leading exponent.
Section~\ref{ue:section} verifies the width and error conditions.

For $q=n^{\kappa+o(1)}$ with fixed $\kappa>1/2$,
$\ln(q/\sqrt n)=\Theta(\ln n)$.
If $U$ is polylogarithmic, then
\[
 D=\ln\ln n-O(\ln\ln\ln n)\sim\ln\ln n.
\]
The resulting time and space are
$2^{(1/2+o(1))n/\ln\ln n}$.
A larger error, including one produced by coarser rounding,
increases $U$ and reduces the available depth.
Appendix~\ref{ue:finite} gives finite arithmetic and a direct
auxiliary coset implementation as well.

\subsection{The geometry of the public operator}

The allocation argument determines which prefixes the sampler uses.
To carry out each lifting and pairing step, the corresponding prefix
lattice must have a sufficiently small smoothing parameter.
Proving this condition is where the public distribution matters.

For a uniform scalar matrix over a prime field, a fixed nonzero
label has a uniform image. Averaging the Gaussian mass over that
image gives a first moment bound for the dual lattice.
Multiplication by a uniform field element has the same property,
even though all entries of its coefficient matrix are related.
This observation gives the NTRU LPRime geometry directly.

For a prime cyclotomic module, a nonzero label may generate a proper
ideal. Its image is uniform within that ideal.
Large ideals have many elements, which lowers the probability
assigned to each possible image.
Small ideals have fewer elements, but their nonzero integer
coefficient representatives are long.
The proof combines these two estimates to bound their Gaussian
contributions. This is the main geometric argument for Kyber/ML-KEM.
It also applies to the square row block selected from the ML-DSA
public operator. For that family, restricting
$\mathbf y_R=\Delta\mathbf t_1$ and $\mathbf e_R=\mathbf s_2-\mathbf t_0$
to the predetermined rows and taking coefficients makes the discarded
public-key bits enter the same joint energy certificate.

Modulo a power of two, additive order plays the corresponding role
for a scalar matrix. Grouping labels by their orders gives an exact
theta series identity and the FrodoKEM smoothing bound.
For SABER's cyclotomic module, these groups are replaced by the
powers of the ramified ideal generated by $X+1$.
The order-two part requires a separate distance estimate, because
a nonzero systematic coordinate is at least $1/2$ from an integer.

The common task is to bound the expected nonzero dual mass, which
gives the required prefix smoothing estimates.
Uniform images provide this bound over fields, while ideals and
additive orders describe how labels spread in the structured rings.
All of these estimates are made in the original coefficient norm.
They also cover prefixes that stop partway through a polynomial.
For example, a sampling level may add several complete ring elements
and then a proper prefix of the next one.
Treating that partial block is essential to using the common
allocation without changing the recovery algorithm.
The family sections organize their proofs around these specific
geometric tasks and then substitute the energy bounds into
Theorem~\ref{ue:main}.

\subsection{From Gaussian mass to a short pair search}
The two remaining NIST third round families considered here, Falcon and NTRU HPS/HRSS, are based on public quotient relations. A quotient $h=gf^{-1}\pmod q$ determines a lattice containing the short pair $(f,g)$. Our aim is to find a short pair satisfying this relation and the corresponding public checks.

The linear method above relies on geometric conditions that make the solution unique within the prescribed energy bound. However, a quotient relation has a different homogeneous structure and admits both the zero pair and a valid nonzero pair. A norm ball centered at the origin and containing the valid pair also contains its negative, and may contain further short solutions. We therefore retain Gaussian mass estimation, but use it to guide a search through affine slices for a pair satisfying the public checks. A bound on the total Gaussian mass accounts for these additional lattice points and controls the probability of reaching each fixed valid pair.

Here each ring element has $m$ coefficients, so the full pair lattice
has dimension $2m$.
For a norm bound $B$, a width $s_{\rm search}$, and a total mass bound
$\rho_{s_{\rm search}}(\Lambda_h)\le A$, every fixed pair $\mathbf z$ of
norm at most $B$ has Gaussian probability at least
\[
 p_B=A^{-1}\exp(-\pi B^2/s_{\rm search}^2).
\]
In the parameter range of the quotient theorem,
$\ln A=O(m/D_{\rm Q})$, $B=O(\sqrt q)$, and
$s_{\rm search}^2=\Theta(qD_{\rm Q}/m)$.
Hence $p_B=\exp(-O(m/D_{\rm Q}))$.
This gives the scale of the search.

To reach a pair, the algorithm fixes coordinates successively.
At each stage, the current candidates form an affine slice of the
lattice. Poisson summation expresses its Gaussian mass as a cosine
average on the dual of its direction lattice.
A list for that direction therefore estimates the weights of the
next possible coordinate values.
Choosing according to these weights follows the Gaussian mass
toward a complete short pair.
A positive coordinate anchor excludes zero, and retaining both signs
covers either sign of a valid pair.

The sampling difficulty is that the search uses many direction
lattices. One persistent random extension serves all of them.
We first build lists in this enlarged lattice and then take
differences of predetermined disjoint pairs.
A difference is retained when it belongs to the original lattice.
Theorem~\ref{uf:persistent-theorem} proves that a public source mass
bound suffices for these lists.
Theorem~\ref{uf:candidate-table} then shows that the resulting search
contains each fixed short witness with high probability.
The proof combines simultaneous accuracy for all queried phases with
a lower bound on the probability of the witness's path.
Scanning the candidate table with a public verifier gives
Theorem~\ref{uq:mass-search-framework}.

For Falcon, the key generation quality condition gives a full
coefficient basis with Gram--Schmidt lengths at most $C\sqrt q$.
The identity $q\Lambda_h^*=\mathbf J\Lambda_h$, with $\mathbf J$
orthogonal, converts its reversed dual basis into a basis of
$\Lambda_h$ with Gram--Schmidt lengths at least $\sqrt q/C$.
Bounding each one-dimensional Gaussian sum in these coordinates gives
$\rho_s(\Lambda_h)\le\theta(Cs/\sqrt q)^{2n}$.
At the search width this is $1+n^{-\omega(1)}$, so a public budget
$A_{\rm mass}=2$ works uniformly for every quality-qualified key.
The verifier imposes the same exact quality threshold on the output.

Spectral bounds on a short pair provide another way to bound the
Gaussian mass of a negacyclic public lattice. They control lengths
in the pair's polynomial span and in the orthogonal projection.
Appendix~\ref{uf:spectral-proof} derives the corresponding recovery
guarantee.

For cyclic NTRU of degree $N$, the span and projection estimate also
includes a saturation index. We bound this index by
$N|f(1)|\operatorname{Norm}(f(\zeta_N),g(\zeta_N))$.
A bound on the expected logarithm of this ideal norm gives
$\ln\Delta_{\rm sat}\le N/D_N$ with probability $1-o(1)$,
where $D_N\sim\ln\ln N$.
The spectral estimates and the key generation acceptance
probabilities then supply a public mass budget for HPS and HRSS.
The slice search and exact verification give their recovery results.

\section{Preliminaries}
\label{ui:prelim}

\subsection{Notation and coefficient embeddings}
\label{up:notation}

Bold symbols denote vectors (including module vectors) or matrices;
Greek vector and matrix symbols are bold as well. Ordinary italic
symbols denote scalars, coefficients, and individual ring elements.
Thus $s_i,e_i$ are scalar coordinates, whereas $\mathbf e_i$ is a
standard basis vector when so declared.

Let $R$ be an integer ring with a specified coefficient basis of
length $d$, and let $R_q=R/qR$ for an integer $q\ge2$.
The map $\operatorname{coeff}$ lists coefficients in this basis and
concatenates the blocks of a ring vector in the stated order.
For example, for $\mathbf s_R\in R^k$ and
$\mathbf A\in R_q^{r\times k}$, let
\[
 \mathbf s=\operatorname{coeff}(\mathbf s_R)\in\mathbb Z^{kd},\qquad
 \mathbf M=\operatorname{Mat}_{\rm coeff}(\mathbf A)
       \in\mathbb Z_q^{rd\times kd},
\]
so that
\[
 \operatorname{coeff}(\mathbf A\mathbf s_R)
       =\mathbf M\mathbf s\pmod q.
\]
Here $r$ is the number of ring equations and $h=rd$ is the number of
coefficient observations. For a residue-ring element,
$\operatorname{coeff}$ first denotes its residue-class coefficient
vector. We write $\operatorname{lift}_q^+$ for the coordinatewise lift
to $\{0,\ldots,q-1\}$, and $\operatorname{lift}_q^{\rm cen}$ for the
lift to $[-\lfloor q/2\rfloor,\lceil q/2\rceil-1]$.
An integer lift is specified before using a norm or an integer
observation; scheme-specific rounding and tie rules are stated locally.
All norms and inner products use the Euclidean coefficient embedding.
In particular, norms of integer ring elements and ring vectors refer
to their coefficient vectors, not to residue classes without a lift.
Polynomial coefficients are indexed by $0,\ldots,d-1$, and abstract
vector coordinates by $[n]=\{1,\ldots,n\}$, unless otherwise stated;
thus $\mathbf e_0=\operatorname{coeff}(1)$ is compatible with this rule.
For a matrix $\mathbf M$ with $n$ columns, $\mathbf M_m$ denotes its first
$m$ columns, with $\mathbf M_0$ the empty column matrix.
The transpose of a coefficient matrix is its ordinary matrix transpose.
We write $\ln$ for the natural logarithm and $\operatorname{Unif}(S)$
for the uniform distribution on a finite set $S$, also denoted by
$U(S)$ in some cited specifications. The scalar $U$ below is instead
a normalized radius.

Rings and local parameters are redeclared in each scheme or theorem.
In particular, $R$ may denote a ring or a candidate radius, $h$ a
public ring element or the number of coefficient observations, and
$f$ a ring element or a scoring width, as explicitly defined locally.
The smoothing parameter $\eta_\varepsilon(L)$ is distinct from a
scheme's coefficient bound $\eta$.

\subsection{Lattices and discrete Gaussians}
\label{up:lattices}

A rank-$d$ lattice in $\mathbb R^k$ is a set
\[
 L=\mathcal L(\mathbf B)=\{\mathbf B\mathbf z:\mathbf z\in\mathbb Z^d\},
 \qquad \mathbf B=(\mathbf b_1,\ldots,\mathbf b_d)\in\mathbb R^{k\times d},
\]
where the columns of $\mathbf B$ are linearly independent. Its span is
$V=\operatorname{span}_{\mathbb R}(L)$, and its determinant is the
$d$-dimensional volume
$\det L=\sqrt{\det(\mathbf B^T\mathbf B)}$ of a fundamental parallelepiped.
For a full-rank lattice, this is $|\det \mathbf B|$.
The dual lattice and first minimum are
\[
 L^*=\{\mathbf u\in V:\langle\mathbf u,\mathbf x\rangle\in\mathbb Z
                       \text{ for every }\mathbf x\in L\},\qquad
 \lambda_1(L)=\min_{0\ne\mathbf x\in L}\|\mathbf x\|.
\]
Thus $(L^*)^*=L$ and $\det L^*=(\det L)^{-1}$, with the dual
always taken inside $V$. If $L\subseteq L'$ have the same rank, their
finite quotient has size $[L':L]=\det L/\det L'$.
A sublattice $S\subseteq L'$ is \emph{primitive} when
$S=L'\cap\operatorname{span}_{\mathbb R}(S)$. Its orthogonal
projection $\pi_{S^\perp}(L')$ is a lattice. Each fiber is a
translate of $S$. These fibers are used in Gaussian lifting below.

For $t>0$ and a discrete subset $T$ of Euclidean space, put
\[
 \rho_t(\mathbf x)=e^{-\pi\|\mathbf x\|^2/t^2},\qquad
 \rho_t(T)=\sum_{\mathbf x\in T}\rho_t(\mathbf x),\qquad
 \theta(t)=\rho_t(\mathbb Z).
\]
The discrete Gaussian $D_{L,t}$ assigns probability
$\rho_t(\mathbf x)/\rho_t(L)$ to $\mathbf x\in L$. More generally,
\[
 \Pr_{\mathbf X\sim D_{L,t,\mathbf c}}[\mathbf X=\mathbf x]
   =\frac{\rho_t(\mathbf x-\mathbf c)}{\rho_t(L-\mathbf c)}
   \quad(\mathbf x\in L).
\]
The parameter $t$ is a Gaussian width. For the corresponding
one-dimensional continuous Gaussian the standard deviation is
$t/\sqrt{2\pi}$. We also use the centered law $D_{L+\mathbf c,t}$ supported
on a coset. These conventions give
$\mathbf X-\mathbf c\sim D_{L-\mathbf c,t}$ when
$\mathbf X\sim D_{L,t,\mathbf c}$.

\subsection{Smoothing and Poisson summation}
\label{up:gaussian-smoothing}

For $0<\varepsilon<1$, the smoothing parameter is
\begin{equation}\label{up:smoothing-definition}
 \eta_\varepsilon(L)=\inf\{t>0:
       \rho_{1/t}(L^*\setminus\{0\})\le\varepsilon\}.
\end{equation}
It measures the width at which the nonzero dual frequencies have small
total Gaussian mass. Scaling gives
$\eta_\varepsilon(aL)=a\eta_\varepsilon(L)$ for $a>0$.
For lattices of the same rank and span, $L\subseteq L'$ implies
$\eta_\varepsilon(L')\le\eta_\varepsilon(L)$, because
$(L')^*\subseteq L^*$.

We use Poisson summation in the conventions of~\cite{MR07}.
For a rank-$d$ lattice $L$ and $\mathbf c\in V$, it gives
\begin{equation}\label{up:poisson-mass}
 \rho_t(L+\mathbf c)
 =\frac{t^d}{\det L}
   \sum_{\mathbf u\in L^*}e^{-\pi t^2\|\mathbf u\|^2}
                       e^{2\pi\mathrm i\langle\mathbf u,\mathbf c\rangle}.
\end{equation}
Taking $L=\mathbb Z$ and $\mathbf c=0$ yields
$\theta(t)=t\theta(1/t)$. The triangle inequality gives
$\rho_t(L+\mathbf c)\le\rho_t(L)$, and, if
$t\ge\eta_\varepsilon(L)$, separating the zero frequency gives
\begin{equation}\label{up:shift-mass}
 (1-\varepsilon)\frac{t^d}{\det L}
 \le\rho_t(L+\mathbf c)\le
 (1+\varepsilon)\frac{t^d}{\det L}.
\end{equation}
This uniformity over the center is the reason smoothing controls
Gaussian lifting and convolution.

The same formula, applied to a Gaussian times a character, gives
the exact Fourier identity
\begin{equation}\label{up:fourier-identity}
 \mathbb E_{\mathbf X\sim D_{L,t}}
    \left[e^{2\pi\mathrm i\langle\mathbf X,\mathbf w\rangle}\right]
 =\frac{\rho_{1/t}(L^*-\mathbf w)}{\rho_{1/t}(L^*)}
 \qquad(\mathbf w\in V).
\end{equation}
The right side is real, so it is also the expected cosine.
For a modular kernel, substituting $t=q/f$ and $\mathbf w=\mathbf v/q$ produces
the periodic Gaussian score used in Section~\ref{ur:section}.

\subsection{Modular kernels and the LWE relation}
\label{up:qary}

For $q\ge2$, write $\mathbb Z_q=\mathbb Z/q\mathbb Z$. When $q$ is
prime this is the field $\mathbb F_q$. For a modular matrix
$\mathbf F\in\mathbb Z_q^{a\times b}$, define
\[
 \ker_q(\mathbf F)=\Lambda_q^\perp(\mathbf F)
   =\{\mathbf x\in\mathbb Z^b:\mathbf F\mathbf x=0\pmod q\}.
\]
This lattice has rank $b$, since it contains $q\mathbb Z^b$,
regardless of the number of modular constraints. If
$\mathbf F=[\mathbf C\mid\mathbf I_a]$, then the map
$\mathbf x\mapsto\mathbf F\mathbf x$ is surjective, even
for composite $q$, and
\begin{equation}\label{up:systematic-identities}
 \det\ker_q(\mathbf F)=q^a,\qquad
 \ker_q(\mathbf F)^*=\mathbb Z^b+q^{-1}\mathbf F^T\mathbb Z^a.
\end{equation}
The determinant is the index of the kernel in $\mathbb Z^b$.
The displayed dual contains $\mathbb Z^b$, and its $q^a$ distinct
cosets are distinguished by the identity coordinates. This proves
equality, since the dual has determinant $q^{-a}$.

In the conventional LWE notation~\cite{Reg05}, the public input is
$\mathbf A\in\mathbb Z_q^{n\times m}$ together with
\[
 \mathbf b=\mathbf A^T\mathbf s+\mathbf e\pmod q.
\]
Search-LWE asks for the secret $\mathbf s$. Decision-LWE distinguishes this
input from an independent uniform $\mathbf b\in\mathbb Z_q^m$.
For every $\mathbf x\in\Lambda_q^\perp(\mathbf A)$,
$\langle\mathbf x,\mathbf b\rangle=\langle\mathbf x,\mathbf e\rangle\pmod q$,
so a short
kernel vector supplies a public statistic of the error. The DEL25
sampler constructs lists of such vectors, and HGH26 uses their
Fourier scores for decision-LWE. Our recovery problem instead writes
the public relation as $\mathbf y=\mathbf M\mathbf s+\mathbf e$ with
$\mathbf M\in\mathbb Z_q^{h\times n}$
and allows $h\le n$.

\subsection{Conditional similarity and concentration}
\label{up:similarity}

The analysis of dual-sieve attacks has highlighted the need to justify
score predictions for the actual sampling process~\cite{DP23,DP26}.
In our setting, correct one-sample marginals alone would not imply
concentration: identical copies of one Gaussian sample have the same
marginals but provide no averaging. Conditional similarity supplies
the additional control used below.

For discrete laws $P,Q$ on a countable set, their statistical distance
is $\operatorname{TV}(P,Q)=\tfrac12\sum_x|P(x)-Q(x)|$.
A random variable is $\gamma$-similar to a law $D$ if its point
probabilities lie between $e^{-\gamma}D(x)$ and $e^\gamma D(x)$.
A list $(\mathbf X_1,\ldots,\mathbf X_N)$ is \emph{strongly conditionally
$\gamma$-similar} to $D$, or to $D^N$, when
\begin{equation}\label{up:conditional-similarity}
 e^{-\gamma}D(\mathbf x)\le
 \Pr[\mathbf X_j=\mathbf x\mid\mathbf X_{-j}=\mathbf x_{-j}]
 \le e^\gamma D(\mathbf x)
\end{equation}
for every index $j$ and every value of the other entries having
positive probability. The notation $D^N$ specifies the reference
marginals, not independence of the actual list.
This is the conditional similarity notion of DEL25~\cite{DEL25},
following the Gaussian list framework of~\cite{ALS21}, and used
in HGH26~\cite{HGH}. The word ``strongly'' emphasizes conditioning
on all other entries. Averaging proves the same bounds when
conditioning only on a prefix.

The following concentration bound will be used for every score.
\begin{lemma}\label{up:concentration}
If $(\mathbf X_1,\ldots,\mathbf X_N)$ is strongly conditionally $\gamma$-similar
to $D$, where $0\le\gamma\le1$, then for every function
$H$ with values in $[-1,1]$ and every $u>0$,
\begin{equation}\label{up:concentration-bound}
 \Pr\left[\left|\frac1N\sum_{j=1}^N H(\mathbf X_j)
                  -\mathbb E_D[H]\right|>2\gamma+u\right]
 \le2e^{-Nu^2/2}.
\end{equation}
\end{lemma}
\begin{proof}
Let $a_j=\mathbb E[H(\mathbf X_j)\mid\mathbf X_1,\ldots,\mathbf X_{j-1}]$.
The prefix version of~\eqref{up:conditional-similarity} gives
$|a_j-\mathbb E_D[H]|\le e^\gamma-1\le2\gamma$.
Conditional on the preceding entries, $H(\mathbf X_j)-a_j$ has mean zero
and lies in an interval of length two. Hoeffding's lemma~\cite{Hoe63}
therefore gives
\[
 \mathbb E[e^{\lambda(H(\mathbf X_j)-a_j)}\mid\mathbf X_1,\ldots,\mathbf X_{j-1}]
 \le e^{\lambda^2/2}.
\]
Iterating this bound and optimizing at $\lambda=u$ bounds either
tail of the sum by $e^{-Nu^2/2}$. Add the conditional mean
discrepancy to obtain~\eqref{up:concentration-bound}.
\end{proof}

Conditional similarity and implementation error play different roles.
The former permits concentration for a dependent reference list.
If an implemented list is within statistical distance $\xi$ of that
reference law, the probability of any event concerning the whole
list changes by at most $\xi$. In particular, after a union bound
over many scores under the reference law, the implementation cost
in probability is still a single additive $\xi$.

\subsection{Gaussian lifting and pairing}
\label{up:del}

To construct a Gaussian list on the final prefix lattice, we enlarge
each prefix by a small auxiliary quotient. Gaussian lifting extends
samples from the previous prefix to this enlarged lattice, and pairing
samples with the same quotient label places their differences in the
desired prefix lattice.
The smoothing conditions control the distortion introduced by both
operations.

For the systematic matrix $[\mathbf M^T\mid\mathbf I_n]$, where
$\mathbf M\in\mathbb Z_q^{h\times n}$, and a block
partition $n=b_1+\cdots+b_r$, put $m_i=\sum_{j\le i}b_j$ and
\[
 L_i=\{(\mathbf u,\mathbf v)\in\mathbb Z^{h+m_i}:
       \mathbf M_{m_i}^T\mathbf u+\mathbf v=0\pmod q\},
 \qquad L_0=\mathbb Z^h.
\]
An auxiliary lattice $\Gamma_i\supseteq q\mathbb Z^{b_i}$ gives
$L_i'=L_i+(0\oplus\Gamma_i)$ and
\begin{equation}\label{up:del-extension}
 0\longrightarrow\Gamma_i\longrightarrow L_i'
 \longrightarrow L_{i-1}\longrightarrow0,\qquad
 L_i'/L_i\simeq\Gamma_i/q\mathbb Z^{b_i}.
\end{equation}

Gaussian lifting uses DEL~\cite[Lemma~3.1]{DEL25},
with $0<\varepsilon\le1/2$.
For a primitive fiber $S\subseteq L'$ and
$P_0=\pi_{S^\perp}(L')$, shifted Gaussian lifting with fresh
randomness and $t\ge\eta_\varepsilon(S)$ preserves list length and gives
\begin{equation}\label{up:del-lift}
 (D_{P_0,t},\gamma)\longmapsto(D_{L',t},\gamma+3\varepsilon).
\end{equation}
Here each pair records the reference Gaussian and its strong
conditional similarity parameter.

Write \textsc{BucketAndCombine} for the ordered pairing operation in
Algorithm~\ref{alg:bucket-and-combine}. We use the cyclic quotient
version of HGH~\cite[Algorithm~2]{HGH}, which specializes
DEL~\cite[Algorithm~3]{DEL25}.
The two passes preserve the input order, and every input entry is used
at most once. In particular, the pairing indices and their output order
depend only on the original coset labels.

\begin{algorithm}[H]
\caption{\textsc{BucketAndCombine}: cyclic quotient pairing~\cite[Algorithm~2]{HGH}}
\label{alg:bucket-and-combine}
\KwIn{Full-rank lattices $L\subseteq L'\subset\mathbb R^d$ with
$L'/L\simeq\mathbb Z/P\mathbb Z$, and an ordered list
$\mathcal X=(\mathbf x_1,\ldots,\mathbf x_{M_0})$ in $L'$, with $M_0\ge3P$.}
\KwOut{An ordered list of $\lfloor M_0/3\rfloor$ vectors in $L$.}
Initialize an empty ordered bucket $B(c)$ for each coset $c\in L'/L$\;
\For{$a\gets1$ \KwTo $M_0$}{
  Record $c_a\gets\mathbf x_a\bmod L$\;
  Append $\mathbf x_a$ to $B(c_a)$\;
}
Initialize $\mathcal X_{\rm out}\gets()$\;
\For{$a\gets1$ \KwTo $M_0$}{
  \If{$|\mathcal X_{\rm out}|<\lfloor M_0/3\rfloor$ and $|B(c_a)|\ge2$}{
    Remove the first two entries $\mathbf x,\mathbf x'$ from $B(c_a)$,
    in that order\;
    Append $\mathbf x-\mathbf x'$ to $\mathcal X_{\rm out}$\;
  }
}
\Return{$\mathcal X_{\rm out}$}\;
\end{algorithm}

We record the correctness and distribution guarantee of
HGH~\cite[Lemma~3.5]{HGH} in our notation. Its smoothing condition
is on the target sublattice $L$.

\begin{lemma}[Correctness and distribution of cyclic pairing]
\label{up:ordered-pairing}
Let $L\subseteq L'$ be full-rank lattices in $\mathbb R^d$ with
$L'/L\simeq\mathbb Z/P\mathbb Z$, and let $M_0\ge3P$.
\begin{enumerate}
\item For every ordered input list of $M_0$ vectors in $L'$,
Algorithm~\ref{alg:bucket-and-combine} returns exactly
$\lfloor M_0/3\rfloor$ vectors in $L$, with no input entry reused.
\item If $\gamma\ge0$, $0<\varepsilon\le1/2$,
$t\ge\sqrt2\eta_\varepsilon(L)$, and the input list is strongly
conditionally $\gamma$-similar to $D_{L',t}$, then the output list is
strongly conditionally $(4\gamma+3\varepsilon)$-similar to
$D_{L,\sqrt2t}$. Thus ordered pairing gives
\begin{equation}\label{up:del-pair}
 (D_{L',t},\gamma)\longmapsto
 (D_{L,\sqrt2t},4\gamma+3\varepsilon),\qquad
 M_0\longmapsto\lfloor M_0/3\rfloor.
\end{equation}
\end{enumerate}
\end{lemma}
\begin{proof}
Two vectors in the same coset modulo $L$ have difference in $L$.
If bucket $c$ initially contains $M_c$ entries, scanning all its original
labels makes $\lfloor M_c/2\rfloor$ disjoint pairs available. Since
\[
 \sum_c\lfloor M_c/2\rfloor\ge(M_0-P)/2\ge M_0/3,
\]
the algorithm reaches the prescribed output length. This is the counting
argument of DEL~\cite[Lemma~3.2]{DEL25}.

The distributional claim is exactly HGH~\cite[Lemma~3.5]{HGH}, with
the notation $(\Lambda_i,\Lambda_i',M,s,\delta)$ replaced by
$(L,L',M_0,t,\gamma)$. Its cyclic quotient label is the coset label
used by Algorithm~\ref{alg:bucket-and-combine}; hence its ordered,
disjoint pairing rule is the same operation.
\end{proof}

Starting from $\gamma_0=0$, lifting followed by
Lemma~\ref{up:ordered-pairing} gives
$\gamma_i=4(\gamma_{i-1}+3\varepsilon)+3\varepsilon$.
Consequently, combining the two operations at each level gives
\begin{equation}\label{up:del-recurrence}
 \gamma_r=5(4^r-1)\varepsilon,\qquad
 3^rN\longmapsto N,\qquad t_0\longmapsto2^{r/2}t_0,
\end{equation}
provided the two smoothing conditions and the quotient-size bound
hold at every level.

\subsection{Auxiliary lattices and shifted sampling}
\label{up:hgh}

\paragraph{Random auxiliary lattices.}
Choose a prime $P$ and a uniform nonzero
$\mathbf z\in\mathbb F_P^b$. Define
\begin{equation}\label{up:hgh-auxiliary}
 \Gamma=q\mathbb Z^b+(q/P)\mathbf z\mathbb Z,
 \qquad \Gamma/q\mathbb Z^b\simeq\mathbb Z/P\mathbb Z,
 \qquad \det\Gamma=q^b/P.
\end{equation}
The quotient has only $P$ classes, independent of the block dimension.
Its dual is
$\Gamma^*=q^{-1}\{\mathbf u\in\mathbb Z^b:
\langle\mathbf u,\mathbf z\rangle=0\pmod P\}$.
If $P^{1/b}\ge16$ and $\sqrt bP^{1/b}\le P$, then
\cite[Lemma~3.2]{HGH} gives, with probability at least $1-2^{1-b}$,
simultaneously for $0<\varepsilon<1$,
\begin{equation}\label{up:hgh-smoothing}
 \eta_\varepsilon(\Gamma)
 \le16qP^{-1/b}
       \sqrt{\max\{1,\ln(1/\varepsilon)/b\}}.
\end{equation}
\paragraph{Implementing the shifted sampler.}
To sample from the random auxiliary lattices, we use the
shifted Gaussian algorithm of Aggarwal, Dadush, and
Stephens-Davidowitz~\cite[Theorem~22]{ADS15}, which supports the
required widths. HGH26 uses the same algorithm in its lifting
steps. We record the specialization needed here.

\begin{lemma}[Shifted Gaussian sampling]\label{up:shifted-sampler}
Let $L\subset\mathbb R^b$ be a full-rank lattice given by a rational
basis, and let $\mathbf c\in\mathbb Q^b$. There is an absolute constant
$c_{\rm DGS}>0$ such that, whenever
\begin{equation}\label{up:hgh-shifted}
 t>\operatorname{dist}(\mathbf c,L)e^{-\sqrt b},
\end{equation}
one can sample within statistical distance $2^{-c_{\rm DGS}b^2}$
of $D_{L,t,\mathbf c}$ in time and space
$2^{b+O(\sqrt b\log b)}=2^{b+o(b)}$, up to polynomial input-bit
factors.
\end{lemma}
Apply~\cite[Theorem~22]{ADS15} with $F(b)=e^{\sqrt b}$.
It supplies at least one sample from $D_{L-\mathbf c,t}$. Retain the
first and translate by $\mathbf c$. The time bound becomes
$2^{b+O(\sqrt b\log b)}$, with error $2^{-\Omega(b^2)}$.

Since $q\mathbb Z^b\subseteq\Gamma$, every center is within
$q\sqrt b/2$ of $\Gamma$, so the width condition can be checked
before the adaptive centers are known.

\begin{lemma}\label{up:adaptive-error}
Suppose call $j$ uses fresh randomness and, for every possible common
history, its implemented law is within total variation distance
$\xi_j$ of the reference call. The joint output laws of at most $J$
adaptive calls are then within distance $\sum_{j=1}^J\xi_j$.
\end{lemma}
This is the sequential composition bound of~\cite[Lemma~3.3]{HGH}.
Section~\ref{ue:section} verifies the shifted sampling conditions and
error budget. Section~\ref{ue:common-sampling} gives finite sampling
tools for both recovery frameworks. Section~\ref{uimpl:linear}
specifies the finite parameters for linear recovery.

\subsection{Gaussian mass from uniform images}
\label{ug:adapters}

The prefix smoothing proof reduces to bounding the expected nonzero
dual mass. The relevant property of a public operator is how widely
it distributes the image of each nonzero label. Over a field, uniform
images give a direct mass calculation even when the matrix entries
are dependent. The structured rings considered later require a finer
decomposition by ideals or additive orders.

Let $q$ be prime and let
$\mathbf A_m\in\mathbb F_q^{n\times m}$ be random. Define the systematic prefix
\begin{equation}\label{ug:prefix}
 \begin{split}
 L(\mathbf A_m)&=\{(\mathbf u,\mathbf v)\in\mathbb Z^n\times\mathbb Z^m:
                    \mathbf A_m^T\mathbf u+\mathbf v=0\pmod q\},\\
 L(\mathbf A_m)^*&=\mathbb Z^{n+m}
       +q^{-1}\begin{pmatrix}\mathbf A_m\\\mathbf I_m\end{pmatrix}\mathbb Z^m.
 \end{split}
\end{equation}
The identity block proves both surjectivity and the displayed dual
decomposition. Call the family \emph{uniform-image} if, for every fixed
$0\ne\mathbf z\in\mathbb F_q^m$, the vector $\mathbf A_m\mathbf z$ is uniform in
$\mathbb F_q^n$. No joint independence between the images of different
labels is required.

\begin{lemma}\label{ug:uniform-mass}
If $\mathbf A_m$ has the uniform-image property, then for every $T>0$,
\begin{equation}\label{ug:uniform-mass-identity}
 \mathbb E\!\left[\rho_{1/T}
   \bigl(L(\mathbf A_m)^*\setminus\mathbb Z^{n+m}\bigr)\right]
 =q^{-n}\theta(q/T)^n
   \bigl(\theta(q/T)^m-\theta(1/T)^m\bigr).
\end{equation}
Let $I_m(T)=\theta(1/T)^{n+m}-1$. Whenever
$I_m(T)<\varepsilon$, one also has
\begin{equation}\label{ug:uniform-markov}
 \Pr\bigl[\eta_\varepsilon(L(\mathbf A_m))>T\bigr]
 \le
 \frac{q^{-n}\theta(q/T)^n
   (\theta(q/T)^m-\theta(1/T)^m)}
      {\varepsilon-I_m(T)}.
\end{equation}
\end{lemma}
\begin{proof}
The identity coordinates make the dual cosets indexed by
$\mathbf z\in\mathbb F_q^m$ disjoint.
For $\mathbf z\ne0$, their Gaussian masses factor as
\[
 \rho_{1/T}(\mathbb Z^n+\mathbf A_m\mathbf z/q)
 \rho_{1/T}(\mathbb Z^m+\mathbf z/q).
\]
The first factor has expectation
\[
 q^{-n}\sum_{\mathbf a\in\mathbb F_q^n}
 \rho_{1/T}(\mathbb Z^n+\mathbf a/q)
 =q^{-n}\theta(q/T)^n.
\]
The residue cosets partition $q^{-1}\mathbb Z^n$, which gives
the last equality.
For the systematic factor, the same partition in dimension $m$ gives
\[
 \sum_{\mathbf z\ne0}
 \rho_{1/T}(\mathbb Z^m+\mathbf z/q)
 =\theta(q/T)^m-\theta(1/T)^m.
\]
Multiplying these expressions and using linearity of expectation
proves~\eqref{ug:uniform-mass-identity}.
The argument uses the marginal uniformity of each fixed image.

The nonzero integer part of the dual mass is
$I_m(T)=\theta(1/T)^{n+m}-1$.
If $\eta_\varepsilon(L(\mathbf A_m))>T$, the remaining mass
exceeds $\varepsilon-I_m(T)$.
Markov's inequality gives~\eqref{ug:uniform-markov}.
\end{proof}
\begin{lemma}\label{up:mass-to-smoothing}
Let $\ln n=o(L_0)$, $\varepsilon_0=4ne^{-L_0}$, and let
$\Lambda_i\subseteq\mathbb Z^{a_i}$ be full-rank random lattices with
$a_i\le2n$. If the prescribed widths satisfy $T_i\ge8\sqrt{L_0}$ and
\[
 \mathbb E[\rho_{1/T_i}(\Lambda_i^*\setminus\mathbb Z^{a_i})]
 \le e^{-3L_0}\qquad(1\le i\le r),
\]
then, for all sufficiently large $n$,
\[
 \Pr[\exists i:\eta_{\varepsilon_0}(\Lambda_i)>T_i]
 \le \frac r{2n}e^{-2L_0}.
\]
\end{lemma}
\begin{proof}
Write
\[
 I_i=\rho_{1/T_i}(\mathbb Z^{a_i}\setminus\{0\}),
 \qquad
 S_i=\rho_{1/T_i}(\Lambda_i^*\setminus\mathbb Z^{a_i}).
\]
Since $\Lambda_i\subseteq\mathbb Z^{a_i}$, the total nonzero
dual mass is $I_i+S_i$.
The bound $\theta(1/T)-1\le4e^{-\pi T^2}$ for $T\ge1$ gives
\[
 I_i\le\exp(8ne^{-\pi T_i^2})-1
 \le2ne^{-L_0}=\varepsilon_0/2
\]
for all sufficiently large $n$.
Thus $\eta_{\varepsilon_0}(\Lambda_i)>T_i$ implies
$S_i>\varepsilon_0/2$.
By Markov's inequality,
\[
 \Pr[\eta_{\varepsilon_0}(\Lambda_i)>T_i]
 \le\frac{\mathbb E[S_i]}{\varepsilon_0/2}
 \le\frac{e^{-3L_0}}{2ne^{-L_0}}
 =\frac1{2n}e^{-2L_0}.
\]
A union bound over the prescribed indices proves the claim,
including when the lattices depend on the same public draw.
\end{proof}

\begin{corollary}
\label{ug:uniform-geometry}
Let $L_0\to\infty$, $L_0=o(n)$, and $\ln n=o(L_0)$. Prescribe
$(m_i,T_i)_{i=1}^r$ before the public draw, with
\begin{equation}\label{ug:uniform-schedule}
 1\le m_i\le n,\qquad 8\sqrt{L_0}\le T_i\le q,\qquad
 \max_i\frac{m_i\ln(q/T_i)}{n\ln T_i}=o(1),
\end{equation}
and assume every $\mathbf A_{m_i}$ has the uniform-image property under the
same distribution of the public matrix. For $\varepsilon_0=4ne^{-L_0}$ and all sufficiently
large $n$,
\begin{equation}\label{ug:uniform-failure}
 \Pr\bigl[\exists i:\eta_{\varepsilon_0}(L(\mathbf A_{m_i}))>T_i\bigr]
 \le \frac r{2n}e^{-2L_0}.
\end{equation}
\end{corollary}
\begin{proof}
Put $c_\theta=\ln\theta(1)$.
For $T_i\le q$, Poisson summation gives
$\theta(q/T_i)=(q/T_i)\theta(T_i/q)\le(q/T_i)\theta(1)$.
Lemma~\ref{ug:uniform-mass} therefore bounds the expected
noninteger dual mass by $e^{-H_i}$, where
\[
 H_i=n\ln T_i-m_i\ln(q/T_i)-(n+m_i)c_\theta.
\]
The prefix ratio in~\eqref{ug:uniform-schedule} makes the second
term $o(n\ln T_i)$ uniformly in $i$.
Also $m_i\le n$ and $\ln T_i\to\infty$ uniformly, so the last term is
$o(n\ln T_i)$.
Consequently $H_i=(1-o(1))n\ln T_i$ uniformly.
Since $T_i\ge8\sqrt{L_0}$ and $L_0=o(n)$,
\[
 \frac{n\ln T_i}{L_0}
 \ge\frac n{L_0}\left(\ln8+\tfrac12\ln L_0\right)
 \longrightarrow\infty.
\]
Hence $H_i\ge3L_0$ for all prescribed indices and all sufficiently
large $n$.
Apply Lemma~\ref{up:mass-to-smoothing} to obtain~\eqref{ug:uniform-failure}.
\end{proof}

We will apply this calculation to independent prime-modulus matrices
and to multiplication by a random field element.

\section{Secret Recovery by Coordinate Scoring}
\label{ur:section}

Let $q\ge2$, let $1\le h\le n$, and let
$\mathbf M\in(\mathbb Z/q\mathbb Z)^{h\times n}$ be a fixed public matrix. The
observed relation is
\begin{equation}\label{ur:model}
 \mathbf y=\mathbf M\mathbf s+\mathbf e\pmod q,\qquad
 \mathbf s\in\mathbb Z^n,\quad \mathbf e\in\mathbb Z^h.
\end{equation}
A rounded observation enters this model through an integer lift, with
$\mathbf e$ equal to the resulting modular residue. The argument allows
$\mathbf s$ and $\mathbf e$ to depend on $\mathbf M$.

\subsection{Systematic coordinates and the joint energy bound}

To test a coordinate, append its proposed value to the public label.
The corresponding lattice therefore has one identity coordinate for
each secret coordinate. Let $\mathbf M_m\in\mathbb Z_q^{h\times m}$
consist of the first $m$ columns of $\mathbf M$, and define
\begin{equation}\label{ur:lattices}
 \begin{gathered}
 \mathbf F=[\mathbf M^T\mid\mathbf I_n],\qquad \Lambda=\Lambda_n,\\
 \Lambda_m=\{(\mathbf u,\mathbf v)\in\mathbb Z^h\times\mathbb Z^m:
                  \mathbf M_m^T\mathbf u+\mathbf v=0\pmod q\},\quad 0\le m\le n,\\
 \mathcal D=q\Lambda^*
       =q\mathbb Z^{h+n}+\mathbf F^T\mathbb Z^n.
 \end{gathered}
\end{equation}
The systematic lattice identities from Section~\ref{ui:prelim} give
$\Lambda_0=\mathbb Z^h$ and $\det\Lambda_m=q^m$, as well as the displayed
formula for $\mathcal D$. The appended coordinates are chosen by the
algorithm, so the observation count remains $h$.

Assume a known joint energy bound and coordinate bound
\begin{equation}\label{ur:certificate}
 \begin{gathered}
 \|\mathbf s\|^2+\|\mathbf e\|^2\le nV,\qquad |s_i|\le B,\qquad V\ge1,\\
 1\le B\le\lfloor\sqrt{nV}\rfloor,\qquad R^2=nV+3B^2.
 \end{gathered}
\end{equation}
Here $V$ is rational and $B$ is an integer. The energy bound itself
provides $B=\lfloor\sqrt{nV}\rfloor$. Any smaller valid bound reduces
the number of candidates. The case $B=0$ is immediate. If the energy
bound holds with probability $1-p_E$, we condition on it for now and
include $p_E$ in the final failure bound. Choose any integer lift of $\mathbf y$.
The scores depend only on its residue modulo $q$.

\begin{lemma}[Coordinate energy gap]\label{ur:energy-gap}
For $i\in[n]$ and a candidate $c\in[-B,B]\cap\mathbb Z$, define
\begin{equation}\label{ur:candidates}
 \mathbf v_{i,c}=(\mathbf y,c\mathbf e_i),\qquad
 \mathbf E_{i,c}=(\mathbf e,c\mathbf e_i-\mathbf s).
\end{equation}
Here $\mathbf e_i$ is the $i$th standard basis vector of $\mathbb Z^n$.
Then, simultaneously for all candidates,
\begin{align}
 \mathbf v_{i,c}&=\mathbf F^T\mathbf s+\mathbf E_{i,c}\pmod q,\nonumber\\
 \|\mathbf E_{i,c}\|^2&\le R^2,\label{ur:radius}\\
 \|\mathbf E_{i,c}\|^2-\|\mathbf E_{i,s_i}\|^2&=(c-s_i)^2.
 \label{ur:gap}
\end{align}
\end{lemma}
\begin{proof}
The first statement is the public congruence and the identity block.
Expansion gives
\[
 \|\mathbf E_{i,c}\|^2=\|\mathbf e\|^2+\|\mathbf s\|^2+c^2-2cs_i.
\]
The last two terms sum to at most $3B^2$. Substituting $c=s_i$ and
subtracting proves the exact gap.
\end{proof}

\subsection{A score estimate on the entire energy ball}

The distance gap in Lemma~\ref{ur:energy-gap} is hidden because $\mathbf s$ and
$\mathbf e$ are unknown. We reveal it by averaging a periodic score over short
lattice vectors. The periodic Gaussian estimator of Aharonov and
Regev~\cite{AR05}, also used in decoding with preprocessing~\cite{DRS14},
connects this average to Gaussian mass around the error vector.
Provable dual attacks analyze such scores through lattice geometry
\cite{PS24,QX26}. Here we need an estimate uniform over the candidate
ball, including its boundary, so that one list can test every label
consistent with the energy bound. For $f>0$, define the ideal score
\[
 \mu(\mathbf v)=\mathbb E_{\mathbf X\sim D_{\Lambda,q/f}}
 \left[\cos\!\left(\frac{2\pi}{q}\langle \mathbf X,\mathbf v\rangle\right)\right].
\]
The phase of $\mathbf F^T\mathbf a$ is an integer multiple of $2\pi$ on
$\Lambda$.
Thus Poisson summation in~\eqref{up:fourier-identity} gives, for every
$\mathbf a\in\mathbb Z^n$ and $\mathbf E\in\mathbb Z^{h+n}$,
\begin{equation}\label{ur:fourier}
 \mu(\mathbf F^T\mathbf a+\mathbf E)=
 \frac{\rho_f(\mathcal D-\mathbf E)}{\rho_f(\mathcal D)}.
\end{equation}
To measure the score and its coordinate gap, put
\begin{equation}\label{ur:signal}
 \tau=e^{-\pi R^2/f^2},\qquad
 g=\tau(1-e^{-\pi/f^2}).
\end{equation}
Here $\tau$ bounds the leading Gaussian score from below on the
candidate ball, while $g$ is a lower bound for the gap between the
leading scores of the correct and any incorrect candidate.

\begin{theorem}[Uniform periodic Gaussian score]\label{ur:score}
Suppose the fixed operator satisfies
\begin{equation}\label{ur:mass-contract}
 \rho_{\sqrt2 f}(\mathcal D\setminus\{0\})\le\zeta,
 \qquad 0<\zeta\le\tau g/64.
\end{equation}
For every integer vector $\mathbf a\in\mathbb Z^n$ and every integer vector
$\mathbf E\in\mathbb Z^{h+n}$ with $\|\mathbf E\|\le R$, simultaneously,
\begin{equation}
 \big|\mu(\mathbf F^T\mathbf a+\mathbf E)-
 e^{-\pi\|\mathbf E\|^2/f^2}\big|
 \le\zeta/\tau\le g/64.
 \label{ur:score-error}
\end{equation}
Each correct coordinate candidate has Gaussian score exceeding every
incorrect candidate's score by at least $31g/32$.
\end{theorem}
\begin{proof}
Put $A_{\mathbf E}=\pi\|\mathbf E\|^2/f^2$ and
write numerator and denominator as $e^{-A_{\mathbf E}}+R_{\mathbf E}$ and
$1+r$.
Gaussian monotonicity gives $0\le r\le\zeta$. The inequality
\[
 \|\mathbf z-\mathbf E\|^2\ge\tfrac12\|\mathbf z\|^2-\|\mathbf E\|^2
\]
gives $0\le R_{\mathbf E}\le\zeta e^{A_{\mathbf E}}$, whence
\[
 \left|\frac{e^{-A_{\mathbf E}}+R_{\mathbf E}}{1+r}-e^{-A_{\mathbf E}}\right|
 =\frac{|R_{\mathbf E}-e^{-A_{\mathbf E}}r|}{1+r}
 \le\max\{R_{\mathbf E},e^{-A_{\mathbf E}}r\}
 \le\zeta e^{A_{\mathbf E}}\le\zeta/\tau.
\]
By~\eqref{ur:gap}, the difference between the two leading Gaussian terms is at least
$\tau(1-e^{-\pi/f^2})=g$. Subtracting their approximation
errors proves the score gap.
\end{proof}

The sampling construction will establish the sufficient condition
\begin{equation}\label{ur:smoothing-suffices}
 \eta_\zeta(\Lambda)\le q/(4f).
\end{equation}
Indeed, the definition of smoothing and $\mathcal D=q\Lambda^*$ give
$\rho_{4f}(\mathcal D\setminus\{0\})
=\rho_{4f/q}(\Lambda^*\setminus\{0\})\le\zeta$.
Gaussian monotonicity then implies~\eqref{ur:mass-contract}.

\begin{lemma}\label{ur:unique}
Under~\eqref{ur:mass-contract}, $\lambda_1(\mathcal D)>2R$, and each
coset of $\mathcal D$ contains at most one vector of norm at most $R$.
In particular, a fixed public label $\mathbf y$ determines at most one integer
secret satisfying~\eqref{ur:model} and~\eqref{ur:certificate}, and $q>2R$.
\end{lemma}
\begin{proof}
A nonzero $\mathbf z\in\mathcal D$ with $\|\mathbf z\|\le2R$ would contribute
at least $e^{-2\pi R^2/f^2}=\tau^2>\zeta$ to the mass in
\eqref{ur:mass-contract}. Two representatives in the same coset and ball
would have such a difference. If $(\mathbf s_1,\mathbf e_1)$ and
$(\mathbf s_2,\mathbf e_2)$ yield the
same label, then
\[
 (\mathbf e_1-\mathbf e_2,\mathbf s_2-\mathbf s_1)\in\mathcal D,
 \qquad \|(\mathbf e_1-\mathbf e_2,\mathbf s_2-\mathbf s_1)\|\le2\sqrt{nV}\le2R.
\]
Distinct secrets give a nonzero vector, which is impossible.
Since $q\mathbb Z^{h+n}\subseteq\mathcal D$, its shortest vector has
length at most $q$.
\end{proof}

\subsection{One list for all public labels}

We estimate all coordinate scores using the same list. Conditional
similarity, as defined in Section~\ref{ui:prelim}, supplies concentration
even when the list entries are dependent. Taking a union bound over
public labels also makes the recovery event valid for secrets and
residues chosen as functions of $\mathbf M$.

\begin{theorem}[One list for all bounded-energy inputs]
\label{ur:shared-list}
Fix an operator satisfying~\eqref{ur:mass-contract}, and let
$0<\delta<1$. Assume that:
\begin{enumerate}
\item a reference list of length $N$ is strongly conditionally
$\alpha$-similar to $D_{\Lambda,q/f}^N$, where
$\alpha\le g/64$.
\item the implemented list is within total variation distance $\nu$
of the reference law.
\item the list length satisfies
\begin{equation}\label{ur:list-demand}
 N\ge\frac{512}{g^2}
       \ln\frac{2n(2B+1)q^h}{\delta}.
\end{equation}
\end{enumerate}
Then, with probability at least $1-\delta-\nu$, that same list recovers
$\mathbf s$ from~\eqref{ur:model} for every pair $(\mathbf s,\mathbf e)$ satisfying
\eqref{ur:certificate}. The procedure \textnormal{\textsc{LinearScoreRecover}}
uses scores of the form
\begin{equation}\label{ur:empirical}
 \widehat\mu_{i,c}=\frac{1}{N}\sum_{j=1}^N
 \cos\!\left(\frac{2\pi}{q}\langle \mathbf X_j,(\mathbf y,c\mathbf e_i)\rangle\right),
\end{equation}
and locates each coordinate's unique maximizing candidate by
adjacent-score binary search.
Given the list and a public label, the procedure uses
\[
 O\!\left(N\bigl[h+n\log(2B+1)\bigr]\right)
\]
scalar operations and cosine evaluations.
Certified error at most $g/64$ in each evaluated score
preserves this conclusion.
\end{theorem}
\begin{proof}
For the proof, consider all labels $\mathbf y\in(\mathbb Z/q\mathbb Z)^h$
and all $n(2B+1)$ candidates. There are
$H=n(2B+1)q^h$ associated functions, all bounded in $[-1,1]$.
Apply Lemma~\ref{up:concentration} to each function and take a union
bound:
\[
 \Pr\!\left[\max_{a\le H}|\widehat\mu_a-\mu_a|
                  >2\alpha+t\right]\le2H e^{-Nt^2/2}.
\]
Use $t=g/16$ and $2\alpha\le g/32$.
Equation~\eqref{ur:list-demand} bounds the failure by $\delta$.
The stochastic score error is at most $6g/64$. Adding the
Gaussian error and the permitted evaluation error leaves at most
$8g/64$ error from each leading Gaussian term. Their separation is
at least $g$, so the correct candidate is the unique maximum.
For adjacent candidates,
\[
 \|\mathbf E_{i,c+1}\|^2-\|\mathbf E_{i,c}\|^2=2(c-s_i)+1.
\]
Its absolute value is at least one. The leading Gaussian scores
therefore differ by at least $g$, with the larger score at $c+1$
exactly when $s_i\ge c+1$. The same error bound preserves this sign,
so binary search retains the half-interval containing $s_i$.
The joint event is simultaneous over every admissible label. Transfer
this one event to the implemented list at cost $\nu$.
\end{proof}

\begin{algorithm}[H]
\caption{\textsc{LinearScoreRecover}: coordinatewise secret recovery}
\label{alg:linear-score-recover}
\KwIn{Modulus $q$, public label $\mathbf y\in(\mathbb Z/q\mathbb Z)^h$, candidate bound $B$, and a shared list $\mathbf X_1,\ldots,\mathbf X_N\in\mathbb Z^{h+n}$.}
\KwOut{The integer vector $\widehat{\mathbf s}\in[-B,B]^n$.}
Write $\mathbf X_j=(\mathbf u_j,\mathbf t_j)$ with $\mathbf u_j\in\mathbb Z^h$ and $\mathbf t_j\in\mathbb Z^n$\;
Compute $a_j\gets\langle \mathbf u_j,\mathbf y\rangle\bmod q$ for every $j$\;
\For{$i\gets1$ \KwTo $n$}{
  Set $\ell\gets-B$ and $r\gets B$\;
  \While{$\ell<r$}{
    Set $c\gets\lfloor(\ell+r)/2\rfloor$\;
    Evaluate $\widehat\mu_{i,c}$ and $\widehat\mu_{i,c+1}$ by~\eqref{ur:empirical}\;
    \eIf{$\widehat\mu_{i,c+1}>\widehat\mu_{i,c}$}{
      $\ell\gets c+1$\;
    }{
      $r\gets c$\;
    }
  }
  $\widehat s_i\gets\ell$\;
}
\Return{$\widehat{\mathbf s}$}\;
\end{algorithm}

For the given label, write $\mathbf X_j=(\mathbf u_j,\mathbf t_j)$ and compute
$\langle \mathbf u_j,\mathbf y\rangle\bmod q$ once per sample. The candidate
phases are then $\langle \mathbf u_j,\mathbf y\rangle+c(\mathbf t_j)_i\bmod q$,
giving the cost in Theorem~\ref{ur:shared-list}. If the observed alphabet
has size at most
$A_y$ per coordinate, $q^h$ in~\eqref{ur:list-demand} can instead be
replaced by $A_y^h$.

The concentration event still covers all $2B+1$ candidates, so it
also covers the adaptively chosen binary search queries.

\begin{corollary}\label{ur:dependence}
Suppose the public geometry and energy bound fail with probabilities
at most $\omega_n$ and $p_E$, respectively. If the list guarantee of
Theorem~\ref{ur:shared-list} holds for each good fixed operator and good
auxiliary choice, then recovery with fresh randomness fails with
probability at most
\[
 \omega_n+p_E+\delta+\nu,
\]
plus the probabilities of auxiliary or finite-execution failures.
This holds for any conditional law of $(\mathbf s,\mathbf e)$ given
$\mathbf M$ that satisfies
the energy bound.
\end{corollary}
\begin{proof}
For each good fixed operator, Theorem~\ref{ur:shared-list} is uniform over
all admissible pairs. Intersect that event with the energy event and then
average over the joint law of the input, adding the failure probabilities
by a union bound.
\end{proof}

\subsection{Energy bounds for the secret and error}
\label{up:dist:energy-distributions}

The recovery theorem uses the sum $\|\mathbf s\|^2+\|\mathbf e\|^2$. The following
estimates let us supply this one quantity throughout the applications.
For example, separate bounds
$\Pr[\|\mathbf s\|^2>n\alpha_n^2]\le p_{s,n}$ and
$\Pr[\|\mathbf e\|^2>n\beta_n^2]\le p_{e,n}$ give
\begin{equation}
 \begin{gathered}
 V_n=\max\{1,\alpha_n^2+\beta_n^2\},\qquad p_n=p_{s,n}+p_{e,n},\\
 \Pr[\|\mathbf s\|^2+\|\mathbf e\|^2>nV_n]\le p_n.
 \end{gathered}
 \label{up:dist:energy-sum}
\end{equation}
This union bound allows dependence between the two vectors. Bounded
coefficients give a deterministic instance: if $h\le n$,
$\|\mathbf s\|_\infty\le H_s$, and $\|\mathbf e\|_\infty\le H_e$, take
$V_n=\max\{1,H_s^2+H_e^2\}$ and $p_n=0$.

\begin{lemma}
\label{up:dist:subgaussian}
Let $\mathbf Y=\boldsymbol\mu+\mathbf X\in\mathbb R^N$, with deterministic $\boldsymbol\mu$, and suppose
a positive semidefinite matrix $\boldsymbol\Sigma$ satisfies
\[
 \mathbb E[e^{\mathbf u^T\mathbf X}]\le e^{\mathbf u^T\boldsymbol\Sigma\mathbf u/2}
 \qquad(\mathbf u\in\mathbb R^N).
\]
For $t>0$, put
\[
 R_t=\operatorname{tr}\boldsymbol\Sigma+
       2\sqrt{t\operatorname{tr}(\boldsymbol\Sigma^2)}+
       2t\|\boldsymbol\Sigma\|_{\rm op}.
\]
Then
\[
 \Pr[\|\mathbf Y\|^2>(\|\boldsymbol\mu\|+\sqrt{R_t})^2]\le e^{-t}.
\]
In particular, for $\mathbf Y=(\mathbf s,\mathbf e)$, any rational $V_n\ge1$ with
$nV_n\ge(\|\boldsymbol\mu\|+\sqrt{R_t})^2$ supplies an energy bound with
failure probability $e^{-t}$.
\end{lemma}
\begin{proof}
The quadratic form inequality of Hsu, Kakade, and
Zhang~\cite[Theorem~2.1]{HKZ12} gives
$\Pr[\|\mathbf X\|^2>R_t]\le e^{-t}$. For a singular $\boldsymbol\Sigma$, the assumed
moment-generating function bound first implies that $\mathbf X$ lies in
$\operatorname{range}(\boldsymbol\Sigma)$ almost surely. Whitening on that
subspace reduces to their isotropic hypothesis. Finally apply
$\|\boldsymbol\mu+\mathbf X\|\le\|\boldsymbol\mu\|+\|\mathbf X\|$.
\end{proof}

For Gaussian and centered binomial coefficients, concentration gives
an explicit energy threshold. Write
$Z_\eta=\sum_{a=1}^{\eta}(B_a-B'_a)$ for the centered binomial
distribution, with independent fair Bernoulli variables.

\begin{corollary}
\label{up:dist:gaussian-cbd}
\label{um:cbd}
Let $Y_1,\ldots,Y_N$ be independent centered coordinates. For a
Gaussian coordinate $Y_i\sim D_{\mathbb Z,w_i}$ put
$v_i^2=w_i^2/(2\pi)$. For a centered binomial coordinate
$Y_i\sim Z_{\eta_i}$ put $v_i^2=\eta_i/2$. For every $t>0$,
\[
 \Pr\!\left[\sum_iY_i^2>
   \sum_i v_i^2+2\sqrt{t\sum_i v_i^4}+2t\max_i v_i^2\right]
 \le e^{-t}.
\]
In particular, for $n$ secret
coordinates of parameter $\eta_s$ and $n$ error coordinates of parameter
$\eta_e$, one may take
\begin{equation}
 V_n=\max\left\{1,\frac{\eta_s+\eta_e}{2}
       +\sqrt{\eta_s^2+\eta_e^2}+\max(\eta_s,\eta_e)\right\},
 \qquad p_n=e^{-n},
 \label{up:dist:cbd-energy}
\end{equation}
rounding $V_n$ upward to a rational number when needed. The deterministic
alternative is $V_n=\max\{1,\eta_s^2+\eta_e^2\}$ and $p_n=0$.
\end{corollary}
\begin{proof}
Completing the square and using the maximum of a shifted theta series
at an integer shift gives, for $Z\sim D_{\mathbb Z,w}$,
\[
 \mathbb E[e^{uZ}]
 =e^{w^2u^2/(4\pi)}
   \frac{\rho_w(\mathbb Z-w^2u/(2\pi))}{\rho_w(\mathbb Z)}
 \le e^{w^2u^2/(4\pi)}.
\]
For the centered binomial distribution,
$\mathbb E[e^{uZ_\eta}]=\cosh^{2\eta}(u/2)\le e^{\eta u^2/4}$.
Independence combines these bounds into the hypothesis of
Lemma~\ref{up:dist:subgaussian} with $\boldsymbol\mu=0$ and
$\boldsymbol\Sigma=\operatorname{diag}(v_i^2)$. This proves the stated tail
bound. Setting $t=n$ for the two centered binomial vectors
proves~\eqref{up:dist:cbd-energy}. The deterministic alternative uses
$|Z_\eta|\le\eta$.
\end{proof}
In every case, the energy threshold is fixed before the public matrix
is sampled and must satisfy the recovery theorem's parameter window.

\section{Constructing the Gaussian List}
\label{ue:section}

The recovery theorem reduces the algorithmic problem to producing a
Gaussian list at width $q/f$. We start from the easy lattice
$\mathbb Z^h$ and add the secret coordinates in blocks. At each step,
Gaussian lifting adds a block and pairing removes an auxiliary modulus.
The lifting and pairing method is due to DEL, with the quantitative
LWE construction developed by HGH~\cite{DEL25,HGH}. We use the
width and block-capacity organization of HGH~\cite[Sec.~3 and
Appendix~B]{HGH}, while checking the prefix conditions for the present
recovery setting.
Sections~\ref{up:del} and~\ref{up:hgh} recall these transformations
and their sampling guarantees in the notation used below.

The main choice is how many coordinates to add at each width. Early
steps have little capacity, while later steps can handle larger blocks. We
show that their combined capacity reaches $n$ while keeping the list
length $2^{(1/2+o(1))n/D}$. The public distribution enters through one
condition: smoothing at the selected coordinate prefixes.

\subsection{The unified recovery statement}

A sampling schedule specifies a tolerance $\varepsilon_0$, widths
$s_0,\ldots,s_r=q/f$, and consecutive prefixes
$0=m_0<m_1<\cdots<m_r=n$. At level $i$, the construction adds
$b_i=m_i-m_{i-1}$ coordinates and uses the geometric threshold
$t_i=s_{i-1}/\sqrt2$. These quantities depend only on $n,q,V,B$ and
are fixed before the public matrix is drawn. The next subsection
constructs the schedule. The theorem needs a smoothing bound at
these prescribed prefixes.

\begin{theorem}[Low-energy linear relation recovery]\label{ue:main}
Let $n\to\infty$, $1\le h\le n$, and $q\ge2$. Fix rational $V$ and
integer $B$ as in~\eqref{ur:certificate}, with polynomial-length binary
representations, before drawing $\mathbf M$. Put
\begin{equation}\label{ue:window}
 \begin{gathered}
 U=\sqrt{R^2/n}\ge1,\qquad Q=\ln(q/\sqrt n)>0,\\
 D=\ln\frac{Q}{1+\ln U}\longrightarrow\infty,
 \qquad \ln q\le(\ln n)^{c_q},
 \end{gathered}
\end{equation}
where $c_q$ is fixed, and use the prescribed schedule below. Assume that:
\begin{enumerate}
\item the input satisfies $\mathbf y=\mathbf M\mathbf s+\mathbf e\pmod q$, and the energy and candidate
bounds~\eqref{ur:certificate} fail with probability at most $p_E$.
\item the scheduled prefixes~\eqref{ur:lattices} satisfy
\begin{equation}\label{ue:geometry}
 \Pr_{\mathbf M}\left[\exists i\in[r]:
       \eta_{\varepsilon_0}(\Lambda_{m_i})>t_i/2\right]
 \le\omega_n=o(1).
\end{equation}
\end{enumerate}
Then there is a randomized algorithm that recovers the integer vector
$\mathbf s$ from the $h$ labels with probability at least
$1-p_E-\omega_n-o(1)$, in expected time and space
\begin{equation}\label{ue:refined-cost}
 2^{(1/2+o(1))n/D}.
\end{equation}
The algorithm constructs one Gaussian list and applies
\textnormal{\textsc{LinearScoreRecover}} to all coordinates using that
list. The joint law of $(\mathbf s,\mathbf e)$ may depend on
$\mathbf M$. The algorithm uses fresh internal randomness independent
of the input.
\end{theorem}

All references to the prescribed schedule use the finite values obtained
by the rounding rules in Section~\ref{uimpl:linear-schedule}. Those rules preserve
the bounds proved below, and their checks pass for every sufficiently
large member of the family.
The implementation attaining~\eqref{ue:refined-cost} uses
Lemma~\ref{up:shifted-sampler} for each Gaussian lift. The proof below
verifies its width condition and accumulated sampling error.
Direct enumeration of the auxiliary cosets gives a second finite
implementation with time and space
\begin{equation}\label{ue:finite-cost}
 2^{O(n/D)}.
\end{equation}
Section~\ref{uimpl:linear} specifies finite arithmetic and stopping
rules for both implementations. When $h=n-o(n)$, either complexity exponent
can also be expressed in terms of the observation dimension up to a
$1+o(1)$ factor.

\begin{remark}\label{ue:substitution}
If only $\|\mathbf s\|^2+\|\mathbf e\|^2\le nV$ is known, take
$B=\lfloor\sqrt{nV}\rfloor$. Then $nV\le R^2\le4nV$, so replacing
$U=R/\sqrt n$ by $\sqrt V$ changes $D$ by at most $\ln2$.
For $q=n^{\kappa+o(1)}$, fixed $\kappa>1/2$, and polylogarithmic
$V$, this gives $D\sim\ln\ln n$.
\end{remark}

\subsection{Widths and block sizes}

The construction starts at a small width on $\mathbb Z^h$, where
sampling is easy. At these early widths, a list can enforce only
small blocks of new coordinate constraints. Lifting and pairing increase
the width while preserving the constraints already imposed, so
later levels can add larger blocks.
The factor $D$ in the exponent comes from adding these capacities
across successive widths. At threshold $t_i$, a list of length $N$ can add
about $\ln N/\ln(q/t_i)$ coordinates. The thresholds grow
geometrically, so the denominators decrease in equal steps. Summing
their reciprocals produces a logarithmic factor, asymptotic to $D$.
Thus a list with $\log_2N\sim n/(2D)$ has enough total capacity.
The parameters below also ensure that its scores can be estimated
accurately.

Fix a rational constant
$K>4\pi/\ln2$ and first define the ideal parameters
\begin{equation}\label{ue:accuracy}
 \begin{gathered}
 f^2=K(R^2/n)D,\qquad \tau=e^{-\pi R^2/f^2},\\
 g=\tau(1-e^{-\pi/f^2}),\qquad
 \varepsilon=\frac{\tau g}{256q^4},\\
 \varepsilon_0=\varepsilon/5,\qquad
 L=\ln(4n/\varepsilon),\qquad L_0=L+\ln5.
 \end{gathered}
\end{equation}
For a sufficiently large absolute constant $C_0$, put
\begin{equation}\label{ue:list-size}
 W=\ln\ln(q/\sqrt L)-\ln(1+\ln f),\qquad
 N=2^{\lceil n/[2(W-C_0-1)]\rceil}.
\end{equation}
The denominator is eventually positive. Section~\ref{uimpl:linear-schedule}
rounds these parameters in directions that preserve the inequalities
below. All geometric hypotheses and the list construction use that
actual finite schedule. The ideal parameters give
\begin{align}
 L_0&=(2\pi/K+o(1))n/D,& W&=(1+o(1))D,
 \label{ue:asymptotics}\\
 \log_2N&=(1/2+o(1))n/D,&
 \gamma_{\rm cap}:=\frac{\ln(1/\varepsilon_0)}{\ln(N/2)}
 &=\frac{4\pi}{K\ln2}+o(1)<1.
 \label{ue:capacity-ratio}
\end{align}
Here $D=O(\ln\ln n)$ follows from $U\ge1$ and the modulus bound.
Moreover $\ln U=o(Q)$ and $Q\le(\ln n)^{c_q}$,
so $\ln f=O((\ln n)^{c_q}+\ln\ln n)=o(n/D)$. To check the
potentially slow divergence in~\eqref{ue:asymptotics}, write
$z=1+\ln U$ and $Q=ze^D$. Then
\[
 \ln f=\ln U+\tfrac12\ln D+O(1),\qquad
 \ln(q/\sqrt L)=Q+\tfrac12\ln D+O(1).
\]
Taking logarithms gives
$W=D+O(1+\ln(1+\ln D))=(1+o(1))D$ even if $D$ grows slowly.
The same calculations give
\begin{equation}\label{ue:margins}
 (\ln q)^3=o(\ln N),\qquad \ln N=o(n),\qquad
 q/(f\sqrt{L_0})\longrightarrow\infty.
\end{equation}

Use the binary schedule
\begin{equation}\label{ue:schedule}
 \begin{gathered}
 r=\left\lfloor2\log_2(q/f)-\log_2(512L_0)\right\rfloor,\\
 s_i=q/(f2^{(r-i)/2})\quad(0\le i\le r),\\
 t_i=s_{i-1}/\sqrt2,\qquad x_i=\ln(q/t_i)\quad(1\le i\le r).
 \end{gathered}
\end{equation}
Then $s_r=q/f$, $s_0^2\ge512L_0$, $t_1\ge16\sqrt{L_0}$,
$t_r=q/(2f)$, and the $x_i$ decrease with spacing $\ln2/2$.
Also $s_0^2<1024L_0$ and $r=O(\ln q)$.

\begin{lemma}[Block allocation and prefix bounds]\label{ue:allocation}
We follow the block-allocation argument of HGH~\cite[Lemma~B.1]{HGH},
verifying below the parameter conditions needed here.
For all sufficiently large $n$, there are a prime $P\in(N/2,N)$ and positive integers
$b_1,\ldots,b_r$ summing to $n$ such that
\begin{equation}\label{ue:capacity}
 \begin{gathered}
 P^{1/b_i}>16,\qquad \sqrt{b_i}P^{1/b_i}\le P,\\
 16qP^{-1/b_i}
 \sqrt{\max\{1,\ln(1/\varepsilon_0)/b_i\}}\le t_i.
 \end{gathered}
\end{equation}
The cumulative prefixes $m_i=\sum_{j\le i}b_j$ obey
\begin{equation}\label{ue:block-bounds}
 b_{\min}=\Omega(\ln N/\ln q),\qquad
 b_{\max}=O(\ln N/(1+\ln f)),\qquad
 m_ix_i=O(\ln N\ln t_i).
\end{equation}
The allocation can be fixed independently of $\mathbf M$.
\end{lemma}
\begin{proof}
We first bound the capacity of each level, then sum these bounds and
choose the blocks. Let $U_0=N/2$, and let $y_i>1/2$ be the unique solution of
\[
 y_i=x_i+\tfrac12\ln(1+\gamma_{\rm cap}y_i)+\ln16,
 \qquad B_i=\lfloor\ln U_0/y_i\rfloor.
\]
The function $y-\tfrac12\ln(1+\gamma_{\rm cap}y)$ is strictly
increasing on this interval. If $1\le b\le B_i$, then
$z=\ln U_0/b\ge y_i$ and, for $P\ge U_0$,
\[
 16qP^{-1/b}\sqrt{\max\{1,\ln(1/\varepsilon_0)/b\}}
 \le16qe^{-z}\sqrt{1+\gamma_{\rm cap}z}\le t_i.
\]
For every $b\le B_i$ and all sufficiently large $n$,
\[
 b\le B_i\le\frac{\ln U_0}{y_i}
 <\frac{\ln U_0}{\ln16}<U_0<P.
\]
Thus, for $b\ge2$, $P^{1/b}\le\sqrt P$ and
$\sqrt b<\sqrt P$, proving the other conditions. For $b=1$ the second
condition is an equality.

To sum the capacities, use $B_i\ge\ln U_0/y_i-1$:
\[
 \sum_iB_i\ge \ln U_0\sum_i\frac1{y_i}-r.
\]
Uniformly $y_i=x_i+O(\ln(1+x_i))$, and therefore
\[
 \frac1{y_i}=\frac1{x_i}
   -O\!\left(\frac{\ln(1+x_i)}{x_i^2}\right).
\]
The $x_i$ have spacing $\ln2/2$, so the integrability of
$\ln(1+x)/x^2$ on $[\ln2,\infty)$ makes the accumulated error $O(1)$.
The same spacing and an integral comparison give
\[
 \sum_i\frac1{x_i}
  =\frac2{\ln2}\ln\frac{x_1}{x_r}+O(1),
 \qquad
 \ln\frac{x_1}{x_r}=W+O(1).
\]
Consequently there is an absolute constant $C_{\rm cap}$ such that,
for all sufficiently large $n$, the harmonic capacity satisfies
\[
 \sum_iB_i
 \ge2\log_2N\,(W-C_{\rm cap})-C_{\rm cap}r\ge n.
\]
Finally fix $C_0\ge C_{\rm cap}+10$ to absorb the summation and rounding
errors. First
assign $b_*:=\lfloor c_*\ln N/\ln q\rfloor$ coordinates to each
level, with sufficiently small fixed $c_*>0$. Each $B_i$ is at least
$2b_*$, and $rb_*=O(\ln N)=o(n)$. Assign the remaining coordinates
in increasing level order without exceeding $B_i$. This fixes positive
blocks with the first two bounds in~\eqref{ue:block-bounds}. Moreover,
\[
 \begin{aligned}
 m_ix_i&\le\ln U_0\left(1+\frac2{\ln2}x_i\ln(x_1/x_i)\right)\\
 &\le\ln U_0\left(1+\frac2{\ln2}(x_1-x_i)\right)
   =O(\ln N\ln t_i),
 \end{aligned}
\]
using $\ln u\le u-1$ and $x_1-x_i=\ln(t_i/t_1)$.
A prime in $(U_0,2U_0)$ can be found by deterministic primality tests
in $N\operatorname{poly}(\ln N)$ work. All choices depend only on the
specified numerical parameters.
\end{proof}

The subsequent geometric proofs need only the following consequences
of the schedule:
\begin{equation}\label{ue:geometry-scale}
 \begin{gathered}
 \varepsilon_0=4ne^{-L_0},\qquad
 L_0=\Theta(n/D),\qquad \ln n=o(L_0),\qquad L_0=o(n),\\
 T_i:=t_i/2\ge8\sqrt{L_0},\qquad T_i\le q,\qquad
 \max_i\frac{m_i\ln(q/T_i)}{n\ln T_i}=o(1).
 \end{gathered}
\end{equation}
For the last estimate, $\ln(q/T_i)=x_i+\ln2\le2x_i$ because
$x_i\ge\ln2$. Since $\ln T_i=\ln t_i-\ln2$,
\eqref{ue:block-bounds} and $\ln N=o(n)$ give
\[
 \frac{m_i\ln(q/T_i)}{n\ln T_i}
 \le\frac{2m_ix_i}{n\ln T_i}
 =O\!\left(\frac{\ln N}{n}\frac{\ln t_i}{\ln T_i}\right)=o(1)
\]
uniformly in the prescribed levels.
In particular
$\ln(n\ln q/L_0)=o(\ln L_0)$, when this additional entropy comparison
is convenient in a structured geometric argument. The remaining task for
each public family is to establish~\eqref{ue:geometry} at these scales.

\subsection{Lifting and pairing}

Fix an integer representative $\widehat{\mathbf M}\in\mathbb Z^{h\times n}$
of $\mathbf M$. For each level define the column index block
\[
 I_i=\{m_{i-1}+1,\ldots,m_i\},\qquad b_i=m_i-m_{i-1}.
\]
Let $\mathbf M^{(i)}\in\mathbb Z^{h\times b_i}$ consist of the columns
of $\widehat{\mathbf M}$ indexed by $I_i$, in increasing order.
This fixed lift supplies the integer values used in extension blocks
and shifted sampling centers.

To add a block, we first enlarge its lattice by a cyclic group of order
$P$. This makes Gaussian lifting possible. Pairing vectors with the
same quotient label then returns their differences to the desired
prefix lattice. Independently of $\mathbf M$, sample each $\mathbf z_i$ uniformly
from $\mathbb F_P^{b_i}\setminus\{0\}$, mutually independently over
$i\in[r]$, and define
\begin{equation}\label{ue:auxiliary}
 \Gamma_i=q\mathbb Z^{b_i}+(q/P)\mathbf z_i\mathbb Z,\qquad
 \Lambda_i'=\Lambda_{m_i}+(0\oplus\Gamma_i).
\end{equation}
Deletion of the newest coordinates gives
\begin{equation}\label{ue:extensions}
 0\longrightarrow\Gamma_i\longrightarrow\Lambda_i'
 \longrightarrow\Lambda_{m_{i-1}}\longrightarrow0,
 \qquad \Lambda_i'/\Lambda_{m_i}\simeq\mathbb Z/P\mathbb Z.
\end{equation}
Indeed a previous vector $(\mathbf u,\mathbf v)$ can be extended by
$-(\mathbf M^{(i)})^T\mathbf u$.
The kernel after deletion is $\Gamma_i$. The second quotient is the
class of $\Gamma_i$ modulo $q\mathbb Z^{b_i}$, of order exactly $P$
because $\mathbf z_i\ne0$. The construction works over every integer modulus.

For this fixed schedule, write
\begin{equation}\label{ue:good-operator}
 \mathcal G_n=\left\{\mathbf M:
   \eta_{\varepsilon_0}(\Lambda_{m_i}(\mathbf M))\le t_i/2
   \text{ for every }i\in[r]\right\}.
\end{equation}
Equation~\eqref{ue:geometry} bounds
$\Pr_{\mathbf M}[\mathbf M\notin\mathcal G_n]$ by $\omega_n$.

\begin{algorithm}[H]
\caption{\textsc{GaussianListEngine}: constructing the shared list}
\label{alg:gaussian-list-engine}
\KwIn{A public operator $\mathbf M$, the fixed prefixes and widths
$(m_i,s_i)_{i=0}^r$, the auxiliary prime $P$, and target length $N$.}
\KwOut{A list of $N$ vectors in the final lattice $\Lambda_n$.}
Initialize a list $L\gets$ $3^rN$ independent samples from
$D_{\mathbb Z^h,s_0}$\;
\For{$i\gets1$ \KwTo $r$}{
  Draw $\mathbf z_i$ uniformly from $\mathbb F_P^{b_i}\setminus\{0\}$,
  independently of the input and previous choices, and form
  $\Gamma_i,\Lambda_i'$ as in~\eqref{ue:auxiliary}\;
  Independently lift every vector in $L$ to $\Lambda_i'$ using the
  shifted Gaussian on $\Gamma_i$ of width $s_{i-1}$\;
  $L\gets\textsc{BucketAndCombine}(L,\Lambda_i',\Lambda_{m_i})$,
  using Algorithm~\ref{alg:bucket-and-combine}\;
}
\Return{$L$}\;
\end{algorithm}

\begin{theorem}[Gaussian list construction]\label{ue:list-engine}
Fix a public operator $\mathbf M\in\mathcal G_n$. With probability at least
$1-2r2^{-b_{\min}}$ over the independent uniform choices
$(\mathbf z_i)_{i=1}^r$, the following hold:
\begin{enumerate}
\item Exact shifted sampling produces exactly $N$ vectors in
$\Lambda$, strongly conditionally $\alpha$-similar to
$D_{\Lambda,q/f}$, with $\alpha\le g/256$.
\item Using the finite initialization in Section~\ref{uimpl:linear-initialization}
and Lemma~\ref{up:shifted-sampler} for the shifted draws produces a
list within total variation distance
$\nu_{\rm DGS}\le\exp[-\Omega(b_{\min}^2)]=o(1)$ of this reference
list. Its expected time and space satisfy~\eqref{ue:refined-cost}.
\end{enumerate}
The expected resource bound is uniform over all public inputs.
\end{theorem}
\begin{proof}
The random auxiliary lattice bound~\eqref{up:hgh-smoothing}, with
\eqref{ue:capacity}, gives
$\eta_{\varepsilon_0}(\Gamma_i)\le t_i$ for all $i$ except with
probability at most $2r2^{-b_{\min}}$. On the resulting
auxiliary-good event,
\[
 s_{i-1}>\eta_{\varepsilon_0}(\Gamma_i),\qquad
 s_{i-1}>\sqrt2\eta_{\varepsilon_0}(\Lambda_{m_i}).
\]
These are the smoothing conditions on the auxiliary and target lattices
required by the lifting transformation~\eqref{up:del-lift}
and Lemma~\ref{up:ordered-pairing}.

Concretely, for $(\mathbf u,\mathbf v)\in\Lambda_{m_{i-1}}$ put
$\mathbf c_i=(\mathbf M^{(i)})^T\mathbf u$, sample
$\mathbf z\leftarrow D_{\Gamma_i,s_{i-1},\mathbf c_i}$, and append
$\mathbf w=\mathbf z-\mathbf c_i$. Thus
$\mathbf w+(\mathbf M^{(i)})^T\mathbf u=\mathbf z\in\Gamma_i$.
For a coordinate $j$ with $z_{i,j}\ne0\pmod P$,
the quotient label of the appended block is
\[
 z_{i,j}^{-1}\frac{P(\mathbf w+(\mathbf M^{(i)})^T\mathbf u)_j}{q}\pmod P.
\]
The displayed numerator divided by $q$ is an integer, and equal labels
give differences in $\Lambda_{m_i}$. The input length
$M_i=3^{r-i+1}N\ge3N>3P$ meets the quotient-size condition of
Lemma~\ref{up:ordered-pairing}, which supplies the exact output length
and conditional similarity in~\eqref{up:del-pair}.
Initialization uses the product Gaussian on
$\mathbb Z^h$, whose width condition follows from
$s_0^2\ge512L_0\gg\ln(h+2)$. For the finite implementation, the
scalar procedure in Section~\ref{uimpl:linear-initialization} generates the initial
list with total statistical error at most $\nu_{\rm init}=2^{-n^2}$
and cost $N^{1+o(1)}$.

The recurrence~\eqref{up:del-recurrence} gives length $3^{r-i}N$
and width $s_i=2^{i/2}s_0$ after level $i$. At the final level,
\begin{equation}\label{ue:similarity}
 \alpha=5(4^r-1)\varepsilon_0
     <q^4\varepsilon\le\tau g/256\le g/256,
\end{equation}
since $4^r=(q/(fs_0))^4<q^4$.

For the implementation, reduce $\mathbf c_i$ modulo $q\mathbb Z^{b_i}$ first.
Translation back preserves the lift, since this product lattice is
contained in $\Gamma_i$. The reduced center is within distance
$q\sqrt{b_i}/2$ of $\Gamma_i$ for every history. The bounds
\eqref{ue:margins} and~\eqref{ue:block-bounds} imply
$\sqrt{b_{\min}}\gg\ln q$, so
\[
 s_{i-1}>q\sqrt{b_i}/(2e^{\sqrt{b_i}}).
\]
Thus Lemma~\ref{up:shifted-sampler} applies to every center encountered by the
algorithm. Summing the input lengths over all levels gives
$J=(3^{r+1}-3)N/2$ shifted calls, each using fresh random bits.
Lemma~\ref{up:adaptive-error} gives
\begin{equation}\label{ue:tv}
 \nu_{\rm DGS}\le\nu_{\rm init}+J2^{-c_{\rm DGS}b_{\min}^2}.
\end{equation}
In particular,
\[
 \nu_{\rm DGS}
 \le\exp\!\left[-\Omega\!\left(
             \frac{n^2}{D^2(\ln q)^2}\right)\right]=o(1),
\]
because $\ln J=O(n/D+\ln q)=o(b_{\min}^2)$ and
$\nu_{\rm init}=2^{-n^2}$ is smaller than the same exponential scale.

Each shifted call costs $2^{b_i+o(b_i)}$, by
Lemma~\ref{up:shifted-sampler}.
As $b_{\max}=O(\ln N/(1+\ln f))=o(\ln N)$, all lifting, pairing,
initialization, quotient arithmetic, and primality search have expected
cost at most
\[
 N^{1+O(1/(1+\ln f))+o(1)}
 \operatorname{poly}(n,\ln q,\ln(1/\varepsilon_0))
 =2^{(1/2+o(1))n/D}.
\]
The bound holds uniformly because the schedule and center-distance
estimate hold for every public input.
\end{proof}

Section~\ref{uimpl:linear-samplers} converts this uniform bound on expected cost
into a deterministic time limit by stopping after
$n^{K_t}\bar T_n$ steps, with at most $n^{-K_t}$ additional
failure probability. The truncated algorithm retains the exponent
in~\eqref{ue:refined-cost} and supplies the time bound used in
Section~\ref{ux:section}.

\subsection{From the list to recovery}

\begin{proof}[Proof of Theorem~\ref{ue:main}]
On the public good event, the final prefix of~\eqref{ue:geometry} gives
$\eta_{\varepsilon_0}(\Lambda)\le t_r/2=q/(4f)$ and hence
$\rho_{4f}(\mathcal D\setminus\{0\})\le\varepsilon_0$.
This implies the mass hypothesis of Theorem~\ref{ur:score}, since
$\varepsilon_0\le\tau g/(1280q^4)$. The list from
Theorem~\ref{ue:list-engine} has conditional similarity error at most
$g/256$.

The list is also long enough to estimate every score. Indeed,
\begin{equation}\label{ue:signal-margin}
 \ln(Ng^2)
 =\left(\frac{\ln2}{2}-\frac{2\pi}{K}+o(1)\right)n/D
 =\Omega(n/D),
\end{equation}
and, writing $H=n(2B+1)q^h$, we have
\[
 \ln H=O(n\ln q),\qquad
 \ln\ln H=O(\ln n+\ln\ln q)=o(n/D).
\]
Thus $Ng^2/\ln H\to\infty$, as required by
Theorem~\ref{ur:shared-list}. The estimate for $H$ uses $B\le R$ and
$\ln R=\tfrac12\ln n+\ln U=O(\ln q+\ln n)$, as follows from
\eqref{ue:window}. A conservative concentration failure bound is
$2n(2B+1)q^h\exp(-Ng^2/2048)=o(1)$.
Theorems~\ref{ur:score} and~\ref{ur:shared-list} therefore identify every
coordinate on the energy event. Lemma~\ref{ur:unique} identifies the
original integer representative.

Combining the public, auxiliary, implementation, and scoring events gives
the failure bound
\begin{equation}\label{ue:failure}
 \begin{split}
 p_E+\omega_n+2r2^{-b_{\min}}+\delta_{\rm impl}
   +2n(2B+1)q^h e^{-Ng^2/2048}.
 \end{split}
\end{equation}
For the ideal reference process $\delta_{\rm impl}=0$.
Section~\ref{uimpl:linear} proves $\delta_{\rm impl}=o(1)$ for both finite
implementations, including sampling, truncation, and stopping errors.
The score computation costs $O(N[h+n\log(2B+1)])$ operations. Since
$\ln B=O(\ln q+\ln n)=o(n/D)$, it preserves the exponent
in~\eqref{ue:refined-cost}.
The implementations in Section~\ref{uimpl:linear} give the claimed
resource bounds. Corollary~\ref{ur:dependence} allows the secret and
residue to depend on the public operator.
\end{proof}

\section{Ternary LWE: Provable Subexponential Recovery}
\label{ut:section}

We first demonstrate the power of our linear recovery framework on
unstructured LWE. In particular, we consider ternary LWE as studied by
May~\cite{May21}, who gave meet-in-the-middle attacks based on
representations with heuristic exponential complexity.
In independent and concurrent work, Chen, Jiang, Wang, and
Zhang~\cite{CJWZ26} give a provable subexponential-time decision
algorithm for narrow LWE, including ternary secrets and errors,
over odd prime polynomial moduli using sublinear samples.

We obtain provable subexponential
recovery in the common unstructured model
\begin{equation}\label{ut:square-model}
 \begin{gathered}
 \mathbf M\leftarrow\operatorname{Unif}(\mathbb F_q^{n\times n}),\qquad
 \mathbf y=\mathbf M\mathbf s+\mathbf e\pmod q,\qquad
 \mathbf s,\mathbf e\in\{-1,0,1\}^n,\\
 q=n^{\kappa+o(1)}\text{ prime},\qquad \kappa>1/2\text{ fixed}.
 \end{gathered}
\end{equation}
Here the entries of $\mathbf M$ are independent and uniform. For comparison
with May, we take $(\mathbf s,\mathbf e)$ independent of $\mathbf M$.
Coordinates within $(\mathbf s,\mathbf e)$ may be dependent, as for a
prescribed Hamming weight.
The recovery result below also allows dependence on $\mathbf M$.
May writes $\mathbf M\mathbf s=\mathbf y+\mathbf e$. Replacing his error by
its negative gives our convention. Both attacks use the original
$n$ equations~\cite[Sections~2.1--2.3]{May21}. 

Let $\mathbf E_m:\mathbb F_q^m\hookrightarrow\mathbb F_q^n$ append zero
coordinates, so $\mathbf M_m=\mathbf M\mathbf E_m$ consists of the first
$m$ columns of $\mathbf M$. The prefixes used by the algorithm are
\begin{equation}\label{ut:prefix}
 \Lambda_m(\mathbf M)=\{(\mathbf u,\mathbf v)\in\mathbb Z^{n+m}:
                    \mathbf M_m^T\mathbf u+\mathbf v=0\pmod q\}.
\end{equation}
For every fixed nonzero $\mathbf z\in\mathbb F_q^m$, the vector
$\mathbf M_m\mathbf z=\mathbf M\mathbf E_m\mathbf z$ is uniform in $\mathbb F_q^n$: each independent
uniform row of $\mathbf M$ takes
a uniform value on $\mathbf E_m\mathbf z$.
Lemma~\ref{ug:uniform-mass} and Corollary~\ref{ug:uniform-geometry}
therefore control the dual mass and smoothing of the prescribed prefixes.

Dual Gaussian mass is also central to smoothing analyses of random
orthogonal lattices~\cite{KNSW20} and of codes and
lattices~\cite{DADRT23}. Here the modular residues index dual cosets,
and their total nonzero Gaussian mass controls the prescribed prefixes.

\begin{corollary}
\label{ut:direct}
Let $q>2$ be prime, $\mathbf M\leftarrow\operatorname{Unif}(\mathbb F_q^{n\times n})$, and
\[
 \mathbf y=\mathbf M\mathbf s+\mathbf e\pmod q,\qquad
 \mathbf s\in\{-1,0,1\}^n,\quad \mathbf e\in\mathbb Z^n.
\]
Suppose a deterministic rational $V_n\ge1$ of polynomial bit length
satisfies
\begin{equation}\label{ut:energy}
 \Pr[\|\mathbf s\|^2+\|\mathbf e\|^2>nV_n]\le p_E.
\end{equation}
With $U=\sqrt{V_n+3/n}$, assume
\begin{equation}\label{ut:window}
 D=\ln\frac{\ln(q/\sqrt n)}{1+\ln U}\longrightarrow\infty,
 \qquad \ln q\le(\ln n)^{c_q}
\end{equation}
for a fixed constant $c_q$. Then the algorithm of
Theorem~\ref{ue:main} recovers $\mathbf s$ from the $n$ labels with probability
$1-p_E-o(1)$, in expected time and space
\begin{equation}\label{ut:cost}
 2^{(1/2+o(1))n/D}.
\end{equation}
The secret and error may depend on $\mathbf M$.
\end{corollary}
\begin{proof}
Take $h=n$, $V=V_n$, and $B=1$, so that
$R^2=nV_n+3$ and $U=\sqrt{V_n+3/n}$.
For each prescribed prefix and each fixed nonzero
$\mathbf z\in\mathbb F_q^{m_i}$, the independent uniform rows of
$\mathbf M$ make $\mathbf M_{m_i}\mathbf z$ uniform in $\mathbb F_q^n$.
Lemma~\ref{ug:uniform-mass} therefore gives the expected dual mass
in~\eqref{ug:uniform-mass-identity}.
The actual finite schedule satisfies~\eqref{ue:geometry-scale},
so Corollary~\ref{ug:uniform-geometry} applies with
$T_i=t_i/2$ and $\varepsilon_0=4ne^{-L_0}$, giving
\[
 \Pr_{\mathbf M}\!\left[\exists i\in[r]:
   \eta_{\varepsilon_0}(\Lambda_{m_i})>t_i/2\right]
 \le \frac r{2n}e^{-2L_0}=o(1).
\]
This is the geometry condition~\eqref{ue:geometry}.
The energy bound~\eqref{ut:energy} and $|s_j|\le1$ give the
certificate~\eqref{ur:certificate} with failure probability at most $p_E$.
Combine these failure events by a union bound under the joint input law.
Together with~\eqref{ut:window}, Theorem~\ref{ue:main} gives the stated
success probability and resource bounds.
\end{proof}

In the square ternary model~\eqref{ut:square-model},
$\|\mathbf s\|^2+\|\mathbf e\|^2\le2n$ deterministically. Thus
Corollary~\ref{ut:direct} applies with
\begin{equation}\label{ut:ternary-parameters}
 \begin{gathered}
 V_n=2,\qquad B=1,\qquad p_E=0,\qquad U=\sqrt{2+3/n},\\
 D=\ln\frac{\ln(q/\sqrt n)}{1+\ln U}
   =\ln\ln n+O(1).
 \end{gathered}
\end{equation}
Indeed, $\ln(q/\sqrt n)=(\kappa-1/2+o(1))\ln n$, and the
growth condition on $\ln q$ in~\eqref{ut:window} also holds.
The original secret is therefore recovered from the given $n$ equations
with probability $1-o(1)$, in expected time and space
\begin{equation}\label{ut:ternary-cost}
 2^{(1/2+o(1))n/\ln\ln n}=2^{o(n)}.
\end{equation}
This energy bound holds for every ternary pair $(\mathbf s,\mathbf e)$.

For a prescribed ternary type of fixed positive entropy rate
$\mathcal H>0$, the secret space has size
\begin{equation}\label{ut:search-space}
 S=2^{\mathcal H n+o(n)}.
\end{equation}
May's representation-based attacks have heuristic cost
$S^{c+o(1)}$ for fixed $c>0$ (roughly $S^{0.25}$ in the summary
of~\cite{May21}). In the common square model with
$(\mathbf s,\mathbf e)$ independent of $\mathbf M$, these are
exponential estimates, whereas~\eqref{ut:ternary-cost} is a proved
subexponential upper bound.

The energy corollary also allows larger errors. If
$q=n^{\kappa+o(1)}$ with $\kappa>1/2$ and
$\|\mathbf e\|_\infty\le B_n=(\ln n)^{O(1)}$, choose $V_n=1+B_n^2$ and
$p_E=0$. Then $D=(1+o(1))\ln\ln n$. A squared-norm bound
$\Pr[\|\mathbf e\|^2>R_n]\le\nu_n$ works just as directly, with
$V_n=1+R_n/n$ and $p_E=\nu_n$.

\section{FrodoKEM: Recovery over Power-of-Two Moduli}
\label{up:section}

Let $q=2^{b_q}$ and write $R_q=\mathbb Z/q\mathbb Z$. A nonzero vector
over this ring may have additive order smaller than $q$, so its image
under a uniform matrix can be confined to a proper subgroup.
Grouping dual vectors by their additive order restores an exact
Gaussian mass calculation. For $q=n^{\kappa+o(1)}$, $\kappa>1/2$,
and polylogarithmically bounded coefficients, this yields recovery
in expected time and space $2^{(1/2+o(1))n/\ln\ln n}$.
We prove the smoothing bound, then apply it to the matrix relation
used by FrodoKEM.

\subsection{An exact additive order identity}

For a uniform scalar matrix over a field, every nonzero label has a
uniform image. For a power of two modulus, the image instead depends
on the additive order of the label. Grouping labels by this order
separates the output mass from the systematic mass and gives an exact
theta series identity for the nonzero dual mass.

As before, $\theta(t)=\rho_t(\mathbb Z)$ and lengths are measured in the
ordinary coefficient norm.

\begin{lemma}\label{up:theta}
Let $h,d\ge1$, let $a=h+d$, and sample
$\mathbf U\leftarrow(\mathbb Z/q\mathbb Z)^{d\times h}$ uniformly.
For $\Lambda_{\mathbf U}=\Lambda_q^\perp([\mathbf U\mid\mathbf I_d])$ and $T>0$,
\begin{equation}\label{up:theta-identity}
 \mathbb E_{\mathbf U}\!\left[\rho_{1/T}(\Lambda_{\mathbf U}^*\setminus\mathbb Z^a)\right]
 =\sum_{\substack{Q=2^j\\1\le j\le b_q}}
 Q^{-h}\theta(Q/T)^h
 \bigl(\theta(Q/T)^d-\theta(Q/(2T))^d\bigr).
\end{equation}
\end{lemma}
\begin{proof}
The systematic coordinates identify distinct dual cosets with
$\mathbf z\in(\mathbb Z/q\mathbb Z)^d$. A nonzero $\mathbf z$ has additive order
$Q=2^j$ precisely when $\mathbf z=(q/Q)\mathbf u$, with
$\mathbf u\in(\mathbb Z/Q\mathbb Z)^d$ having an odd coordinate.
A uniform column of $\mathbf U$ sends $\mathbf z/q$ uniformly to
$Q^{-1}\mathbb Z/\mathbb Z$, and the $h$ columns are independent.
For each such output coordinate, averaging the Gaussian coset mass gives
$Q^{-1}\theta(Q/T)$. Summing the systematic coordinate mass over $\mathbf u$
with an odd coordinate gives the difference
$\theta(Q/T)^d-\theta(Q/(2T))^d$. The subtracted vectors are exactly
the all-even vectors. Multiplying these contributions and summing over
all additive orders proves the formula.
\end{proof}

\begin{lemma}[Rectangular dyadic prefix smoothing]\label{up:smoothing}
Suppose $1\le d\le3h/2$, $a=h+d$, $0<\zeta\le1$, and
$\ln(1/\zeta)\le h/8$. If
\begin{equation}\label{up:smoothing-width}
 T\ge8\max\left\{\sqrt{\ln(2a/\zeta)},\,
                         q^{d/a}\zeta^{-1/a}\right\},
\end{equation}
then $\Pr_{\mathbf U}[\eta_\zeta(\Lambda_{\mathbf U})>T]\le8e^{-h/4}$.
\end{lemma}
\begin{proof}
For $0<u\le1$, direct Gaussian series bounds give
$\theta(u)\le1+3e^{-\pi/u^2}$ and $\theta(u)\le e^{1/8}$.
Poisson summation gives $\theta(u)=u\theta(1/u)$.
Let $F_Q$ denote a summand in~\eqref{up:theta-identity}. For $Q\le T$,
\[
 \sum_{Q\le T}F_Q
 \le\frac{e^{a/8}2^{-h}}{1-2^{-h}}
 \le\frac{e^{-(\ln2-5/16)h}}{1-2^{-h}}.
\]
For $Q>T$, transform all theta factors and use the determinant threshold:
\[
 \sum_{Q>T}F_Q
 \le e^{a/8}T^{-a}\frac{q^d}{1-2^{-d}}
 \le\frac{\zeta e^{-(\ln8-1/8)a}}{1-2^{-d}}.
\]
The integer part is deterministic and satisfies
\[
 \rho_{1/T}(\mathbb Z^a\setminus\{0\})
 =\theta(1/T)^a-1
 \le e^{3ae^{-\pi T^2}}-1\le\zeta/2.
\]
Markov's inequality bounds the probability that the remaining mass is
larger than $\zeta/2$ by
\[
 \frac{2e^{-(\ln2-7/16)h}}{1-2^{-h}}
 +\frac{2e^{-(\ln8-1/8)a}}{1-2^{-d}}
 \le8e^{-h/4}.
\]
Here $\ln2-7/16>1/4$, $a\ge h+1$, and both denominators are at least
$1/2$. The prefix $d=0$ is the integer lattice used to initialize the
algorithm.
\end{proof}

\subsection{Recovery from a square or nearly square system}

For $U\ge1$, write
\begin{equation}\label{up:window}
 Q_n=\ln(q/\sqrt n),\qquad D_n(U)=\ln\frac{Q_n}{1+\ln U}.
\end{equation}
Throughout this section $Q_n>0$ and $\ln q\le(\ln n)^{c_q}$ for a
fixed $c_q$. The recovery schedule is that of
Theorem~\ref{ue:main}, fixed before drawing the public matrix.

\begin{proposition}
\label{up:prefix-interface}
Let $h=n-o(n)$ and $\mathbf M\leftarrow\operatorname{Unif}(R_q^{h\times n})$.
Let $\mathbf M_m$ consist of the first $m$ columns of $\mathbf M$.
For the scheduled prefixes
$m_i=\sum_{j\le i}b_j$ and
\[
 \Lambda_{m_i}=\Lambda_q^\perp([\mathbf M_{m_i}^T\mid \mathbf I_{m_i}]),
\]
the geometric condition~\eqref{ue:geometry} fails with probability at
most $8r e^{-h/4}=o(1)$.
\end{proposition}
\begin{proof}
Put $T_i=t_i/2$, $\zeta=\varepsilon_0=4ne^{-L_0}$, and $a_i=h+m_i$.
The schedule gives $t_i\ge16\sqrt{L_0}$ and
$\ln(2a_i/\zeta)\le L_0$, establishing the first threshold
in~\eqref{up:smoothing-width}. The second is equivalent to
\[
 a_i\ln t_i-m_i\ln q\ge\ln(1/\zeta)+a_i\ln16.
\]
Write $x_i=\ln(q/t_i)$. After subtracting $a_i\ln16$, the left side is
\[
 h(\ln t_i-\ln16)-m_i(x_i+\ln16).
\]
The allocation bound~\eqref{ue:block-bounds} gives
$m_ix_i=O(\ln N\ln t_i)$, with $x_i\ge\ln2$ and $\ln N=o(h)$.
Also $\ln t_i\ge(1/2+o(1))\ln n$ and
$\ln(1/\zeta)=O(n/D)=o(n)$. Hence this expression is eventually at
least $(h/4)\ln t_i\gg\ln(1/\zeta)$. Finally, $m_i\le n\le3h/2$
and $\ln(1/\zeta)\le h/8$ eventually. Apply
Lemma~\ref{up:smoothing} to each prefix and sum its failure bound over
the $r=O(\ln q)$ prescribed prefixes.
\end{proof}

\begin{theorem}[Dyadic recovery]
\label{up:energy}
Let $h=n-o(n)$, $\mathbf M\leftarrow\operatorname{Unif}(R_q^{h\times n})$, and
$\mathbf y=\mathbf M\mathbf s+\mathbf e\pmod q$. Suppose a deterministic rational $V_n\ge1$ of
polynomial bit length satisfies
\[
 \Pr[\|\mathbf s\|^2+\|\mathbf e\|^2>nV_n]\le p_n,
 \qquad D_n(\sqrt{V_n})\longrightarrow\infty.
\]
Then Theorem~\ref{ue:main} recovers the original integer vector $\mathbf s$
from the $h$ labels with probability $1-p_n-o(1)$, in expected time
and space
\begin{equation}\label{up:energy-cost}
 2^{(1/2+o(1))n/D_n(\sqrt{V_n})}.
\end{equation}
The joint law of $\mathbf s,\mathbf e$ may depend on $\mathbf M$.
\end{theorem}
\begin{proof}
Set $V=V_n$ and $B=\lfloor\sqrt{nV_n}\rfloor$.
On the stated energy event, every secret coordinate satisfies
$|s_j|\le B$, so the certificate~\eqref{ur:certificate} holds.
The candidate radius satisfies
\[
 nV_n\le R^2=nV_n+3B^2\le4nV_n,
 \qquad \sqrt{V_n}\le U=R/\sqrt n\le2\sqrt{V_n}.
\]
As in Remark~\ref{ue:substitution}, this gives
$D_n(U)=D_n(\sqrt{V_n})+O(1)\sim D_n(\sqrt{V_n})$.
Proposition~\ref{up:prefix-interface} verifies the geometry
condition~\eqref{ue:geometry} for the prescribed finite schedule,
with failure probability at most $8r e^{-h/4}=o(1)$.
A union bound with the energy failure event, followed by
Theorem~\ref{ue:main}, gives the stated recovery probability and
resource bounds for the joint law of $\mathbf M,\mathbf s,\mathbf e$.
\end{proof}

\subsection{Several columns with the same public matrix}
\label{ua:frodo}

FrodoKEM uses the relation $\mathbf B=\mathbf A\mathbf S+\mathbf E\pmod q$, with a scalar matrix $\mathbf A$
and several short columns~\cite[Sections~2.2.2, 2.2.4, 2.2.5, and~2.2.9,
Algorithm~12]{FrodoSpec20}. In our asymptotic family,
$\mathbf A\leftarrow\operatorname{Unif}(R_q^{n\times n})$ and $\mathbf S,\mathbf E\in\mathbb Z^{n\times t_n}$,
where $1\le t_n\le n^c$ for a fixed constant $c$. Thus
\begin{equation}\label{ua:frodo-relation}
 \mathbf B=\mathbf A\mathbf S+\mathbf E\pmod q
\end{equation}
is a collection of relations with one shared operator. The list can
therefore be reused.

\begin{corollary}
\label{ua:frodo-corollary}
Suppose a deterministic rational $V_n\ge1$ of polynomial bit length
satisfies
\begin{equation}\label{ua:frodo-energy}
 \begin{gathered}
 \Pr\!\left[\max_{j\le t_n}
   \bigl(\|\mathbf s^{(j)}\|^2+\|\mathbf e^{(j)}\|^2\bigr)>nV_n\right]\le p_n,\\
 D_n(\sqrt{V_n})\longrightarrow\infty.
 \end{gathered}
\end{equation}
One Gaussian list recovers every column of $\mathbf S$ with probability
$1-p_n-o(1)$, with the same asymptotic time and space as
Theorem~\ref{up:energy}, normalized by the dimension $n$ of one column.
Dependence among the columns and on $\mathbf A$ is allowed.
\end{corollary}
\begin{proof}
Use the common certificate $V=V_n$ and
$B=\lfloor\sqrt{nV_n}\rfloor$ for every column.
As in the proof of Theorem~\ref{up:energy},
Proposition~\ref{up:prefix-interface} supplies the geometry event.
On this event, Theorem~\ref{ue:list-engine} constructs one Gaussian
list with the required conditional similarity and implementation error.
The final prefix gives~\eqref{ur:mass-contract}, and
\eqref{ue:signal-margin} ensures~\eqref{ur:list-demand} for some
$\delta=o(1)$, as verified in the proof of Theorem~\ref{ue:main}.
The construction and its error bounds depend on the common operator
and certificate.
Theorem~\ref{ur:shared-list} therefore gives one event on which this
list recovers every label whose secret and error satisfy that certificate.
Intersect this event with the simultaneous energy event
in~\eqref{ua:frodo-energy} and apply a union bound, including the
geometry, auxiliary, and implementation errors.
The resulting success probability is $1-p_n-o(1)$ under the stated
joint law of the columns and $\mathbf A$.
Scoring the columns successively reuses the same storage and
multiplies the scoring time by at most $t_n$.
Since $t_n$ is polynomial in $n$, the stated exponent is preserved.
\end{proof}

For a concrete sequence, take
\begin{equation}\label{ua:frodo-q}
 b_{q,n}=\lceil\kappa\log_2n\rceil,\qquad q=2^{b_{q,n}},
 \qquad \kappa>1/2.
\end{equation}
If the entries of $\mathbf S$ and $\mathbf E$ are bounded in absolute
value by public positive integers $H_s,H_e=(\ln n)^{O(1)}$, respectively, use
\[
 B=H_s,\qquad V=\max\{1,H_s^2+H_e^2\},\qquad V_n=V,\qquad p_n=0.
\]
These bounds hold simultaneously for all columns. Using $B=H_s$ in
Theorem~\ref{ue:main}, the normalized candidate radius is
$U=\sqrt{V+3B^2/n}=\sqrt V(1+O(1/n))$, since $B^2\le V$.
Thus $D_n(U)\sim D_n(\sqrt V)\sim\ln\ln n$, and the expected time and space are
\begin{equation}\label{ua:frodo-polylog}
 2^{(1/2+o(1))n/\ln\ln n}.
\end{equation}
This covers any prescribed family of coefficient tables with such
support. The same complexity also holds for centered subgaussian
coefficients of polylogarithmic width: the energy estimates of
Section~\ref{up:dist:energy-distributions}, followed by a union bound
over columns, give a polylogarithmic $V_n$ and $p_n=o(1)$.
Packing and unpacking the public matrix preserve the residues in
\eqref{ua:frodo-relation}.

\begin{corollary}[FrodoKEM family specialization]
\label{ua:frodo-headline}
Let $q=2^{\lceil\kappa\log_2n\rceil}$ with fixed $\kappa>1/2$,
and let $t_n\le n^c$ be the prescribed number of columns. Suppose a
public polylogarithmic $V_n\ge1$ gives the simultaneous joint energy
bound $nV_n$ with failure probability $o(1)$.
For uniform $\mathbf A\in(\mathbb Z/q\mathbb Z)^{n\times n}$,
one Gaussian list recovers all columns of $\mathbf S$ from
$\mathbf B=\mathbf A\mathbf S+\mathbf E$ with probability $1-o(1)$,
in expected time and space $2^{(1/2+o(1))n/\ln\ln n}$.
\end{corollary}
\begin{proof}
The candidate-energy substitution of Remark~\ref{ue:substitution}
gives $U_n=\Theta(\sqrt{V_n})$ and $1\le U_n=(\ln n)^{O(1)}$.
Since $\ln(q/\sqrt n)=(\kappa-1/2)\ln n+O(1)$,
\[
 D_n=\ln\frac{\ln(q/\sqrt n)}{1+\ln U_n}
     =\ln\ln n+O(\ln\ln\ln n)\sim\ln\ln n.
\]
Corollary~\ref{ua:frodo-corollary} gives simultaneous recovery on
the single energy event. The polynomial number of columns changes
neither the leading exponent nor the reusable list space.
This proves Corollary~\ref{ui:frodo-cor}.
\end{proof}

\section{Kyber/ML-KEM: Recovery over Prime Cyclotomic Modules}
\label{um:section}

Multiplication in a polynomial ring creates dependencies among matrix
entries. A nonzero vector can now have an image supported on a proper
ideal, so the uniform-image calculation of Section~\ref{ut:section} no
longer applies directly. We group the dual vectors by their image ideal.
Large ideals give enough randomness. Small ideals force long integer
representatives. These two estimates establish the same prefix condition
as before. For fixed module rank and splitting degree, polynomial moduli
and polylogarithmically bounded coefficients then give recovery in
expected time and space $2^{(1/2+o(1))n/\ln\ln n}$.

\subsection{The public operator and the modulus restriction}

\begin{definition}[Cyclotomic modules over a prime modulus]
\label{ua:mlkem-definition}
Fix a positive integer $k$, let $d$ tend to infinity through powers of two,
and put $n=kd$.  For an odd prime $q$, write
\[
 O=\mathbb Z[X]/(X^d+1),\qquad R_q=O/qO,\qquad
 \delta=\operatorname{ord}_{2d}(q).
\]
Choose a family with constant $\delta$ and
$\ln q\leq(\ln n)^c$ for a fixed $c$. Sample
$\mathbf A\leftarrow\operatorname{Unif}(R_q^{k\times k})$. For module vectors
$\mathbf s_R,\mathbf e_R\in O^k$ and $\mathbf y_R\in R_q^k$, write
$\mathbf s=\operatorname{coeff}(\mathbf s_R)$,
$\mathbf e=\operatorname{coeff}(\mathbf e_R)$, and
$\mathbf y=\operatorname{lift}_q^+(\operatorname{coeff}(\mathbf y_R))$,
concatenating the
coefficients of the $k$ polynomial components. Let $\mathbf M$ be the
$n\times n$ coefficient matrix of multiplication by $\mathbf A$.
Thus the $k$ ring equations
\begin{equation}
 \mathbf y_R=\mathbf A\mathbf s_R+\mathbf e_R\pmod q
 \label{um:model}
\end{equation}
are precisely $\mathbf y=\mathbf M\mathbf s+\mathbf e\pmod q$ in coefficient coordinates.
Order the columns of $\mathbf M$ first by secret polynomial and then by coefficient degree.
All norms are Euclidean norms of these integer coefficient vectors.
\end{definition}
For $k=1$, this is the algebraic form of the ring-LWE
relation~\cite{LPR13}. Ring vectors and matrices give the module
form~\cite{LS15}. The same public relation
underlies Kyber/ML-KEM~\cite[Section~5.1, Algorithm~13]{FIPS203}. Its bounded
coefficient distributions will enter after the geometric argument.
Fixing the splitting degree forces $q=\Omega(n)$. We prove this
restriction in Lemma~\ref{um:modulus}, where it determines the
complexity of the bounded-coefficient case.

\subsection{Ideal images in the coefficient basis}

A coefficient prefix can end inside a polynomial block. Its image under
ring multiplication is controlled by an ideal, but its allowed inputs
need not form a submodule. The mass bound must therefore track both the
image ideal and the dimension of a partial block. We use the ring
decomposition to count ideals and the original coefficient basis to
measure lengths.

The Gaussian mass calculation is related to ring
regularity~\cite[Theorem~4.1]{LPR13T}, cyclotomic leftover
hashing~\cite{BL26}, and its Gaussian module counterpart~\cite{AFLLW26}.
Here we need the estimate for every prescribed coefficient prefix,
including those with an incomplete polynomial block.

For $0\leq m\leq n$, let $\mathbf M_m$ consist of the first $m$ columns
of $\mathbf M$, and define the systematic prefix lattice
\begin{equation}
 \Lambda_m=\{(\mathbf u,\mathbf v)\in\mathbb Z^n\times\mathbb Z^m:
                 \mathbf M_m^T\mathbf u+\mathbf v=0\pmod q\}.
 \label{um:prefix}
\end{equation}
In particular, $\Lambda_0=\mathbb Z^n$.  The identity block makes the
row map surjective and gives the disjoint dual coset decomposition
\begin{equation}
 \Lambda_m^*=
 \bigcup_{\mathbf z\in\mathbb F_q^m}
 \left(\mathbb Z^{n+m}
       +\frac1q(\mathbf M_m\mathbf z,\mathbf z)\right).
 \label{um:dual}
\end{equation}
For $\mathbf z\in\mathbb F_q^m$, let
$\widetilde{\mathbf z}\in R_q^k$ be the module vector obtained by
padding $\mathbf z$ with zeros and grouping the coefficients into
polynomial blocks. The first component $\mathbf M_m\mathbf z$ is
$\operatorname{coeff}(\mathbf A\widetilde{\mathbf z})$.

For $T>0$, put
\[
 a_T=\rho_{1/T}(\mathbb Z),\qquad
 b_T=\rho_{q/T}(\mathbb Z).
\]
If $V\leq\mathbb F_q^D$ is a subspace of dimension $v$, then
\begin{equation}
 \sum_{\mathbf z\in V}\rho_{1/T}(\mathbb Z^D+\mathbf z/q)
       \leq a_T^{D-v}b_T^v.
 \label{um:subspace}
\end{equation}
Indeed, select $v$ information coordinates on which projection is a
bijection.  Every remaining coordinate contributes a shifted Gaussian mass
at most $a_T$, by one-dimensional Poisson summation and the triangle
inequality.  Summing the information coordinates over all their residues
then gives $b_T^v$.  We shall also use $b_T\leq1+q/T\leq2q/T$ for $T\leq q$.

\begin{lemma}
\label{um:ideal-norm}
Let $I\lhd R_q$ be an ideal of dimension $w$ over $\mathbb F_q$.  If
$z\in\mathbb Z[X]$ is nonzero, has degree less than $d$, and satisfies
$z\bmod q\in I$, then
\begin{equation}
 \|\operatorname{coeff}(z)\|\geq q^{1-w/d}.
 \label{um:norm}
\end{equation}
\end{lemma}
\begin{proof}
Let $\mathbf M_z$ be the integral negacyclic multiplication matrix of $z$. Its
image modulo $q$ lies in $I$ because $I$ is an ideal, so its rank modulo
$q$ is at most $w$.
Smith normal form implies $q^{d-w}\mid\det\mathbf M_z$. As $d$ is a power of
two, $X^d+1$ is irreducible over $\mathbb Q$, hence $\det\mathbf M_z\ne0$.
Each column of $\mathbf M_z$ is a signed permutation of the coefficient vector
of $z$.  Hadamard's inequality yields
$q^{d-w}\leq|\det\mathbf M_z|\leq\|\operatorname{coeff}(z)\|^d$.
\end{proof}

Write $m=jd+s$, with $0\leq s<d$, using $j=k,s=0$ at $m=n$.
For a nonzero prefix $\mathbf z\in\mathbb F_q^m$ with
$\widetilde{\mathbf z}=(z_1,\ldots,z_k)\in R_q^k$, define
\[
 g=\gcd(X^d+1,z_1,\ldots,z_k),\qquad
 w=d-\deg g,
 \qquad h_g=jw+\max\{0,s-d+w\},
\]
where the gcd is monic.  A complete polynomial block in the ideal $(g)$
has $w$ free coefficients.  A partial block of degree less than $s$ is an
ordinary multiple of $g$ of degree less than $s$, and contributes
$\max\{0,s-\deg g\}$ dimensions.  Thus $h_g$ is exactly the dimension
of the subspace of prefixes all of whose blocks lie in $(g)$, and
\begin{equation}
 h_g\leq mw/d,\qquad h_g\leq kw.
 \label{um:prefix-dimension}
\end{equation}

For any fixed prefix with this exact gcd, each ring component of
$\mathbf A\widetilde{\mathbf z}$
is uniform in $(g)$, and the $k$ components are independent.  In fact,
the map
\[
 R_q^k\longrightarrow R_q,\qquad
 (a_1,\ldots,a_k)\longmapsto\sum_{t=1}^k a_tz_t
\]
has image $(g)$ and equal-size fibres. Independence follows from the
independent rows of the ring matrix $\mathbf A$.

Let $\mathcal E_g$ be the exact-gcd class of nonzero prefixes, and write
\begin{equation}
 D_g(T)=\sum_{\mathbf z\in\mathcal E_g}
  \mathbb E_{\mathbf A}\!\left[
    \rho_{1/T}(\mathbb Z^n+\mathbf M_m\mathbf z/q)
    \rho_{1/T}(\mathbb Z^m+\mathbf z/q)\right].
 \label{um:class-mass}
\end{equation}
The uniform-image statement and~\eqref{um:subspace} bound the expected
first factor by $q^{-kw}a_T^{n-kw}b_T^{kw}$.  Enlarge the systematic
prefix class to its containing subspace of dimension $h_g$ and apply
\eqref{um:subspace} once more.  For $T\leq q$ this proves
\begin{align}
 D_g(T)
 &\leq q^{-kw}a_T^{n+m-h_g-kw}b_T^{h_g+kw}
 \notag\\
 &\leq q^{h_g}a_T^{n+m-h_g-kw}(2/T)^{h_g+kw}.
 \label{um:dimension-mass}
\end{align}
The first estimate uses the dimension of the image ideal to quantify
the averaging over uniform outputs. Larger ideals provide more
randomness and will be controlled by this bound. Small ideals instead
force long nonzero integer representatives, giving a complementary
geometric estimate.

For any point of a nonzero systematic coset, multiplying by $q$ gives a
nonzero integer polynomial block whose residue lies in $(g)$.
Lemma~\ref{um:ideal-norm} therefore
bounds the norm of this point below by $q^{-w/d}$.  Split the Gaussian
exponent into two equal parts, use this norm bound on one part, and bound
the other shifted mass by the unshifted mass.  Each first-component
shifted mass is at most $a_T^n$, and there are at most $q^{h_g}$ prefixes.
Consequently,
\begin{equation}
 D_g(T)\leq q^{h_g}a_T^n a_{T/\sqrt2}^m
          \exp\left(-\frac\pi2 T^2q^{-2w/d}\right).
 \label{um:norm-mass}
\end{equation}

\subsection{Smoothing every prescribed prefix}

We now combine the two mass estimates. The logarithmic condition below
ensures that the count of large ideals is absorbed by their Gaussian
decay. The recovery schedule will satisfy it automatically.

\begin{proposition}
\label{um:geometry}
Assume the public model above.  Let $L_0$ satisfy
\begin{equation}
 \ln n=o(L_0),\qquad L_0=o(n),\qquad
 \ln\left(\frac{n\ln q}{L_0}\right)=o(\ln L_0),
 \qquad \varepsilon_0=4n e^{-L_0}.
 \label{um:geometry-scale}
\end{equation}
Choose deterministic pairs $(m_i,T_i)_{i=1}^r$, independent of $\mathbf A$, with
$1\leq m_i\leq n$, $8\sqrt{L_0}\leq T_i\leq q$, and
\begin{equation}
 \max_i\frac{m_i\ln(q/T_i)}{n\ln T_i}=o(1).
 \label{um:prefix-ratio}
\end{equation}
Then, for all sufficiently large $n$,
\begin{equation}
 \Pr_{\mathbf A}\!\left[\exists i:
       \eta_{\varepsilon_0}(\Lambda_{m_i})>T_i\right]
       \leq\frac{r}{2n}e^{-2L_0}.
 \label{um:geometry-probability}
\end{equation}
\end{proposition}
\begin{proof}
Fix one of the prescribed pairs and abbreviate it by $(m,T)$.
Because $T^2\geq64L_0$ and $\ln n=o(L_0)$, every integer theta factor
in~\eqref{um:dimension-mass} and~\eqref{um:norm-mass} is $1+o(1)$,
uniformly in the ideal and the pair.  For example,
$a_{T/\sqrt2}=1+O(e^{-\pi T^2/2})$, and its exponent is at most $2n$.
Equations~\eqref{um:dimension-mass}, \eqref{um:prefix-dimension}, and
\eqref{um:prefix-ratio} imply
\begin{align}
 \ln D_g(T)
 &\leq o(1)+h_g\ln(q/T)-kw\ln T+(h_g+kw)\ln2
 \notag\\
 &\leq o(1)-(1-o(1))kw\ln T.
 \label{um:dimension-decay}
\end{align}
All error terms in this proof are uniform in the prescribed pairs.

First consider large image ideals. Set $w_*=8L_0/(k\ln T)=o(d)$.
The polynomial $X^d+1$ is squarefree
modulo $q$, and every irreducible factor has degree $\delta$.  The number
of dimension-$w$ ideals is therefore
$\binom{d/\delta}{w/\delta}$ when $\delta\mid w$, and zero otherwise.
For $w\geq w_*$, it is at most $\binom dw$, and
\[
 \ln\binom dw\leq w\ln(ed/w_*)
 \leq w\ln\left(\frac{e n\ln q}{L_0}\right)
 =o(w\ln T).
\]
Here the last equality follows from~\eqref{um:geometry-scale} and
$\ln T\geq\tfrac12\ln L_0+\ln8$.  Summing
\eqref{um:dimension-decay} over ideals and at most $d$ dimensions yields
\begin{equation}
 \sum_{w\geq w_*}\ \sum_{d-\deg g=w}D_g(T)\leq e^{-6L_0}.
 \label{um:high-dimension}
\end{equation}

For the small ideals, suppose $0<w<w_*$. A nonempty exact-gcd class has $h_g>0$.
The prefix dimension formula then forces $m>d-w>d/2$ eventually.
Using~\eqref{um:prefix-ratio} and fixed $k$ gives
\[
 \ln(q/T)=o(\ln T),\qquad \ln q=(1+o(1))\ln T,
\]
and hence
\begin{equation}
 q^{-2w/d}\geq
 \exp\left(-\frac{16L_0}{n}\frac{\ln q}{\ln T}\right)=1-o(1).
 \label{um:low-norm}
\end{equation}
Since $2d\mid q^\delta-1$, the number of dimension-$w$ ideals is at most
$(d/\delta)^{w/\delta}\leq q^w$.  Along with $h_g\leq kw$, this bounds
the logarithm of the product of the ideal count and the number of prefixes
by
\[
 (k+1)w_*\ln q
 =\left(\frac{8(k+1)}k+o(1)\right)L_0
 \leq(16+o(1))L_0.
\]
The negative exponent in~\eqref{um:norm-mass} is at most
$-(32\pi-o(1))L_0$, by~\eqref{um:low-norm}.  After summing over the
possible dimensions, we obtain
\begin{equation}
 \sum_{0<w<w_*}\ \sum_{d-\deg g=w}D_g(T)\leq e^{-80L_0}.
 \label{um:low-dimension}
\end{equation}
Together, \eqref{um:dual}, \eqref{um:high-dimension}, and
\eqref{um:low-dimension} give
\[
 \mathbb E_{\mathbf A}\!\left[\rho_{1/T}
       (\Lambda_m^*\setminus\mathbb Z^{n+m})\right]\leq e^{-3L_0}.
\]
Lemma~\ref{up:mass-to-smoothing} now proves~\eqref{um:geometry-probability}.
\end{proof}

\subsection{Recovery and the Kyber/ML-KEM relation}
\label{ua:mlkem}

The geometry now connects to the recovery theorem in one step. We state
the consequence for arbitrary low-energy inputs, then specialize the
parameters and coefficient distributions.

\begin{theorem}[Recovery over a prime modulus]
\label{um:recovery}
\label{ua:mlkem-corollary}
In Definition~\ref{ua:mlkem-definition}, let a deterministic rational
$V\ge1$ of polynomial bit length satisfy
\begin{equation}
 \Pr[\|\mathbf s\|^2+\|\mathbf e\|^2>nV]\le p_E.
 \label{um:energy}
\end{equation}
Set $U=\sqrt V$ and assume
\begin{equation}
 D=\ln\frac{\ln(q/\sqrt n)}{1+\ln U}\longrightarrow\infty.
 \label{um:window}
\end{equation}
The algorithm of Theorem~\ref{ue:main} recovers the original integer
vector $\mathbf s$ from the $k$ ring equations with probability $1-p_E-o(1)$.
Its expected time and space are
\begin{equation}
 2^{(1/2+o(1))n/D}.
 \label{um:complexity}
\end{equation}
The joint law of $\mathbf s,\mathbf e$ may depend on $\mathbf A$.
\end{theorem}
\begin{proof}
Remark~\ref{ue:substitution} gives the candidate bound and the same
asymptotic depth.

For the schedule of Theorem~\ref{ue:main}, the bounds
from~\eqref{ue:accuracy}--\eqref{ue:geometry-scale} give
$L_0=\Theta(n/D)$, $T_i=t_i/2\ge8\sqrt{L_0}$, and
\eqref{um:prefix-ratio}. Since $\ln q\le(\ln n)^c$ and $U\ge1$,
we have $D=O(\ln\ln n)$ and
\[
 \ln\left(\frac{n\ln q}{L_0}\right)=O(\ln\ln n)=o(\ln L_0).
\]
Thus all hypotheses of Proposition~\ref{um:geometry} hold for this
schedule. Its failure probability is at most $r e^{-2L_0}/(2n)=o(1)$.
Theorem~\ref{ue:main} gives recovery uniformly over the stated
energy event.
\end{proof}

To specialize the exponent, we first make explicit the restriction
imposed by the fixed splitting degree.

\begin{lemma}
\label{um:modulus}
In this family, $q=\Omega(n)$. Consequently, if
\begin{equation}
 \ln(q/\sqrt n)=(\ln n)^{\beta+o(1)}
 \label{um:beta}
\end{equation}
for a fixed $\beta>0$, then $\beta\geq1$.
\end{lemma}
\begin{proof}
Write $2d=2^u$. The order $\delta$ is a power of two. If $\delta=1$,
then $q\equiv1\pmod{2d}$, and $q\geq2d+1$. Otherwise write
$\delta=2^a$, with $a\geq1$. The elementary valuation identity gives
\[
 u\leq v_2(q^{2^a}-1)
   =v_2(q-1)+v_2(q+1)+a-1.
\]
Exactly one of $v_2(q-1)$ and $v_2(q+1)$ equals one. It follows that
$q\equiv\pm1\pmod{2d/\delta}$ and $q\geq2d/\delta-1$. Both $k$ and
$\delta$ are fixed, so $q=\Omega(n)$. Therefore
$\ln(q/\sqrt n)\geq\tfrac12\ln n-O(1)$, which excludes $\beta<1$.
\end{proof}

For fixed coefficient bounds and
$\ln(q/\sqrt n)=(\ln n)^{\beta+o(1)}$, the expected time and space are
\begin{equation}
 2^{(1/(2\beta)+o(1))n/\ln\ln n}
 \label{um:bounded-complexity}
\end{equation}
with $\beta\ge1$ by Lemma~\ref{um:modulus}.
A useful broader case has $\ln(q/\sqrt n)=\Theta(\ln n)$ and
$\|\mathbf s\|_\infty+\|\mathbf e\|_\infty=(\ln n)^{O(1)}$. The bounded-coefficient
estimate in Section~\ref{up:dist:energy-distributions} gives
$D=(1+o(1))\ln\ln n$, and hence expected time and space
\begin{equation}
 2^{(1/2+o(1))n/\ln\ln n}.
 \label{ua:mlkem-polylog}
\end{equation}
\begin{corollary}
\label{ua:mlkem-existence}
Fix $k$. For every power of two $d\ge2$, let $q_d$ be the least prime
satisfying
\begin{equation}
 q_d\equiv d+1\pmod{2d}.
 \label{ua:mlkem-quadratic}
\end{equation}
Then Definition~\ref{ua:mlkem-definition} holds with $\delta=2$.
Polylogarithmically bounded coefficients give the expected time and space
in~\eqref{ua:mlkem-polylog}.
\end{corollary}
\begin{proof}
The residue class is coprime to $2d$, and the least-prime theorem
of~\cite[Theorem~1.1]{Xyl11} gives
$d+1\le q_d\le C(2d)^{5.18}$ for an effective absolute constant $C$.
Moreover $(d+1)^2\equiv1\pmod{2d}$ and $d+1\not\equiv1\pmod{2d}$,
so $\operatorname{ord}_{2d}(q_d)=2$ and all irreducible factors have
degree two. For $n=kd$, these bounds give
$\ln(q_d/\sqrt n)=\Theta(\ln n)$. The preceding specialization applies.
Enumerating candidates up to the polynomial upper bound and testing
primality also constructs this modulus sequence in polynomial time.
\end{proof}

For centered binomial secret and error coefficients with public positive
integer parameters $\eta_s,\eta_e=(\ln n)^{O(1)}$, the respective supports
are $[-\eta_s,\eta_s]$ and $[-\eta_e,\eta_e]$. Take
\[
 B=\eta_s,\qquad V=\max\{1,\eta_s^2+\eta_e^2\}.
\]
The joint secret and error energy is then deterministically at most $nV$.
The normalized candidate radius in~\eqref{ue:window} is
$U=\sqrt{V+3B^2/n}=\sqrt V(1+O(1/n))$.
When $\ln(q/\sqrt n)=\Theta(\ln n)$, we obtain
$D=(1+o(1))\ln\ln n$ and~\eqref{ua:mlkem-polylog}.
For independent coordinates, Corollary~\ref{up:dist:gaussian-cbd} also
gives a probabilistic energy bound.
Kyber/ML-KEM key generation uses $\eta_s=\eta_e=\eta_1$.
The parameter $\eta_2$ is used for encryption errors.

Public data in an invertible transform representation are
first returned to coefficient form, as in
\cite[Sections~2.4.5 and~4.3]{FIPS203}. This preserves the uniform
distribution of the public matrix and provides the coefficient
representation used throughout the recovery argument.

\begin{corollary}[Kyber/ML-KEM family specialization]
\label{ua:kyber-headline}
Let $n=kd$, with fixed $k$ and $d$ growing through powers of two.
Take prime $q=n^{\kappa+o(1)}$, fixed $\kappa\ge1$, and fixed
$\operatorname{ord}_{2d}(q)$. For uniform
$\mathbf A\in R_q^{k\times k}$, suppose a public polylogarithmic
$V\ge1$ gives a joint secret and error energy bound $nV$ with probability
$1-o(1)$. Then the original coefficient secret is recovered from
$\mathbf t=\mathbf A\mathbf s_R+\mathbf e_R\pmod q$ with probability $1-o(1)$,
in expected time and space $2^{(1/2+o(1))n/\ln\ln n}$.
\end{corollary}
\begin{proof}
The fixed order is the fixed splitting degree in
Definition~\ref{ua:mlkem-definition}. With $U=\sqrt V$,
$1+\ln U=O(\ln\ln n)$ and
$\ln(q/\sqrt n)=(\kappa-1/2+o(1))\ln n$. Hence
\[
 D=\ln\frac{\ln(q/\sqrt n)}{1+\ln U}
   =\ln\ln n+O(\ln\ln\ln n)\sim\ln\ln n.
\]
Theorem~\ref{um:recovery} applies with $p_E=o(1)$ to the $k$ ring
equations, which supply exactly $n$ coefficient observations.
This proves Corollary~\ref{ui:kyber-cor}.
\end{proof}

\section{SABER: Recovery from Rounded Module Products}
\label{ud:section}
\label{ua:saber}

Rounding a module product gives a bounded integer residue, so the common
recovery theorem applies once we prove smoothing for its coefficient
prefixes. The new issue is the modulus: modulo two,
$X^d+1=(X+1)^d$, and the ideals form a single nested chain. We use this
chain to count prefix vectors and bound their Gaussian mass directly in
coefficient coordinates.

\subsection{The rounded relation and its energy}

Fix a positive module rank $k$, let $d=2^\nu$ tend to infinity, and put
$n=kd$. Set
\[
 O=\mathbb Z[x]/(x^d+1),\qquad R_q=O/qO,\qquad
 q=2^{b_q},\qquad \mathbf A\leftarrow\operatorname{Unif}(R_q^{k\times k}).
\]
For the secret module vector $\mathbf s_R\in O^k$, write
$\mathbf s=\operatorname{coeff}(\mathbf s_R)\in\mathbb Z^n$,
concatenating the coefficients of its $k$ polynomial components.
Let $\mathbf M$ be the coefficient matrix of multiplication by
$\mathbf A^T$, so that
$\operatorname{coeff}(\mathbf A^T\mathbf s_R)=\mathbf M\mathbf s\pmod q$.
Coordinates are ordered by polynomial block and then coefficient degree.
All norms below use this basis.

For $p=2^{b_p}$ with $0\le b_p\le b_q$, set $\Delta=q/p$ and
$h_{\rm rnd}=\lfloor\Delta/2\rfloor$. The public module vector is
$\mathbf b\in R_p^k$, with $R_p=O/pO$. Write
$\mathbf b^+=\operatorname{lift}_p^+(\operatorname{coeff}(\mathbf b))\in\mathbb Z^n$.
Writing $[z]_q$ for the least nonnegative residue, define its coefficients by
\begin{equation}
 b_j^+=\left\lfloor\frac{[(\mathbf M\mathbf s)_j+h_{\rm rnd}]_q}{\Delta}\right\rfloor,
 \qquad 0\le b_j^+<p.
 \label{ud:rnd:map}
\end{equation}
For $\Delta\ge2$, this is the public rounding map in the SABER
specification~\cite[Section~2.3 and Algorithm~1]{SaberSpec}.
The case $\Delta=1$ is the unrounded extension.

\begin{lemma}
\label{ud:rnd:lift}
The computable integer lift $\mathbf y=\Delta\mathbf b^+$ satisfies
\begin{equation}
 \mathbf y=\mathbf M\mathbf s+\mathbf e_{\rm rnd}\pmod q,\qquad
 e_{{\rm rnd},j}=h_{\rm rnd}-[(\mathbf M\mathbf s)_j+h_{\rm rnd}]_\Delta.
 \label{ud:rnd:residue}
\end{equation}
For $\Delta\ge2$,
\begin{equation}
 e_{{\rm rnd},j}\in\{1-\Delta/2,\ldots,\Delta/2\},\qquad
 \|\mathbf e_{\rm rnd}\|^2\le n\Delta^2/4.
 \label{ud:rnd:energy-rounding}
\end{equation}
For $\Delta=1$, the residue is zero. More generally, an offset
$0\le h_{\rm rnd}<\Delta$ gives
$h_{\rm rnd}-\Delta+1\le e_{{\rm rnd},j}\le h_{\rm rnd}$.
\end{lemma}
\begin{proof}
Euclidean division gives
$[(\mathbf M\mathbf s)_j+h_{\rm rnd}]_q=\Delta b_j^++r_j$ with
$0\le r_j<\Delta$. Since $\Delta\mid q$, its remainder is
$r_j=[(\mathbf M\mathbf s)_j+h_{\rm rnd}]_\Delta$.
Rearranging proves the identity and bounds.
\end{proof}

Thus, if $|s_j|\le B$ and $\|\mathbf s\|^2\le nV_s$, the common recovery
construction uses
\begin{equation}
 V=\max\{1,V_s+r_{\rm rnd}^2\},\qquad R^2=nV+3B^2,
 \qquad r_{\rm rnd}=\lfloor\Delta/2\rfloor.
 \label{ud:rnd:energy}
\end{equation}
The residue bound holds for every matrix and secret. In particular,
the specification's $\operatorname{CBD}_\eta$ secret has $B=\eta$
(its parameter $\mu=2\eta$), so $V_s=\eta^2$ is a deterministic
choice. Independent coefficients also admit the energy bounds of
Corollary~\ref{up:dist:gaussian-cbd}.

The resulting expected recovery time and space are
$2^{(1/2+o(1))n/\ln\ln n}$ when $q=n^{\kappa+o(1)}$,
$\kappa>1/2$, and both the secret bound and rounding step are
polylogarithmic. The remaining task is to prove prefix smoothing
for this public matrix distribution.

\subsection{Ideals and coefficient prefixes}
\label{ud:alg:section}

The element $1-x$ generates the unique maximal ideal of $R_q$.
We first record its ideal powers and the dimensions seen by a partial
polynomial block. The underlying ring structure is
classical~\cite{Dinh05,LM17}. The prefix counts will determine the
Gaussian estimate.

\begin{lemma}
\label{ud:alg:chain}
There is a unit $u\in O^\times$ with $2=u(1-x)^d$. The ideals of
$R_q$ are exactly
\[
 I_j=(1-x)^jR_q\quad(0\le j\le b_q d),\qquad
 |I_j|=2^{b_q d-j},\qquad I_{b_q d}=0.
\]
Every nonzero element has a unique valuation $j<b_q d$, meaning that
it is $(1-x)^j$ times a unit. For $j=ad+r$ with $0\le r<d$,
\begin{equation}
 (1-x)^jO=2^a(1-x)^rO,\qquad I_j=2^a(1-x)^rR_q.
 \label{ud:alg:valuation}
\end{equation}
\end{lemma}
\begin{proof}
If $\zeta$ is the image of $x$ in the cyclotomic field, then
\[
 2=\Phi_{2d}(1)=\prod_{\substack{1\le b<2d\\b\text{ odd}}}(1-\zeta^b).
\]
For odd $b$, $(1-\zeta^b)/(1-\zeta)$ is an integral geometric sum
of norm one, hence a unit. This proves $2=u(1-x)^d$.
In $R_q$, $1-x$ is nilpotent and $R_q/(1-x)=\mathbb F_2$.
Every element outside $(1-x)R_q$ has the form $1+(1-x)b$ and is
invertible by a finite geometric series. Repeated division by $1-x$
therefore writes every nonzero element as a power of $1-x$ times
a unit. An element of least valuation generates any nonzero ideal.
Multiplication by $1-x$ on $O$ has determinant of absolute value
two, so $[O:(1-x)^jO]=2^j$. Together with $qO=(1-x)^{b_q d}O$,
this proves the cardinalities and uniqueness of the valuation.
\end{proof}

The quotient and remainder $a,r$ in~\eqref{ud:alg:valuation} give a
convenient description of this single chain. Define the binary code
\[
 C_r=(1+x)^r\mathbb F_2[x]/((1+x)^d)\subseteq\mathbb F_2^d.
\]
We identify ideals and codes with their coefficient embeddings when
viewing them as lattices or subsets of coefficient spaces. For
$0\le r\le d$, the code has dimension $d-r$ and
\begin{equation}
 (1-x)^rO=\{\mathbf z\in\mathbb Z^d:\mathbf z\bmod2\in C_r\}.
 \label{ud:alg:binary-preimage}
\end{equation}
Indeed, both lattices contain $2O$, have the same binary image, and
have index $2^r$ in $O$.

\begin{lemma}
\label{ud:alg:prefix}
Let $\operatorname{pr}_t$ select the first $t$ coefficients, and identify
$\mathbb F_2^t$ with vectors whose last $d-t$ coordinates are zero.
For $0\le r<d$ and $0\le t\le d$,
\begin{equation}
 \dim\operatorname{pr}_t C_r=\min(t,d-r),\qquad
 \dim(C_r\cap\mathbb F_2^t)=\max(0,t-r).
 \label{ud:alg:binary-prefix}
\end{equation}
Consequently, for $j=ad+r<b_q d$,
\begin{equation}
 \begin{split}
 |\operatorname{pr}_t I_j|&=2^{(b_q-a)t-\max(0,t-d+r)},\\
 |I_j\cap(\mathbb Z/q\mathbb Z)^t|&=2^{(b_q-a)t-\min(t,r)}.
 \end{split}
 \label{ud:alg:ideal-prefix}
\end{equation}
\end{lemma}
\begin{proof}
The polynomials $(1+x)^rx^b$, $0\le b<d-r$, form a basis of $C_r$.
For $t\le d-r$, their first $t$ coefficient rows contain a triangular
$t\times t$ submatrix with diagonal entries one. This gives the
projection rank, which stays $d-r$ for larger $t$. A polynomial of
degree less than $t$ belongs to $C_r$ exactly when it is
$(1+x)^rh$ with $\deg h<t-r$, giving the intersection rank.
By~\eqref{ud:alg:binary-preimage}, a vector of $I_j$ has coefficients
divisible by $2^a$, its next bit belongs to $C_r$, and its remaining
bits are free. Each allowed pattern on $t$ coordinates has
$2^{(b_q-a-1)t}$ higher-bit lifts, proving the cardinalities.
\end{proof}

Let $V_m\subseteq R_q^k$ consist of vectors supported on the first
$m$ coefficients. Write $m=hd+t$, $0\le t<d$, with $h=k,t=0$ when
$m=n$. For $j=ad+r<b_q d$, set
\begin{equation}
 \begin{split}
 v_j(m)&=h(d-r)+\max(0,t-r),\\
 c_j(m)&=m-v_j(m)=hr+\min(t,r),\\
 H_j(m)&=(b_q-a)m-c_j(m).
 \end{split}
 \label{ud:alg:prefix-dimensions}
\end{equation}
Here $v_j(m)$ counts the free bits in the lowest allowed binary
layer of the prefix, $c_j(m)$ counts its independent binary
constraints, and $H_j(m)$ counts all free bits, including the
unrestricted higher layers. The preceding lemma gives
\begin{equation}
 \begin{gathered}
 |I_j^k\cap V_m|=2^{H_j(m)},\\
 v_j(m)\le \frac{m(d-r)}d\le k(d-r),\qquad
 H_j(m)\le m(b_q-j/d).
 \end{gathered}
 \label{ud:alg:module-prefix}
\end{equation}
The two inequalities follow from
$\max(0,t-r)\le t(d-r)/d$ and $\min(t,r)\ge tr/d$.
These counts include prefixes ending inside a polynomial.

\begin{lemma}
\label{ud:alg:image}
If $\mathbf z_R=(z_1,\ldots,z_k)\in R_q^k$ has minimum component
valuation $j$, then the $k$ components of $\mathbf A^T\mathbf z_R$
are independent and uniform in $I_j$.
\end{lemma}
\begin{proof}
The map $\mathbf a=(a_1,\ldots,a_k)\mapsto\sum_ba_bz_b$ has image equal to the ideal generated
by the $z_b$, namely $I_j$, and has equally sized fibres. Apply this
to the independent columns of $\mathbf A$.
\end{proof}

\subsection{Gaussian mass of the prescribed prefixes}
\label{ud:geo:section}

Let $\mathbf M_m$ consist of the first $m$ columns of $\mathbf M$.
The prefix lattice and its dual are
\begin{equation}
 \begin{split}
 \Lambda_m&=\{(\mathbf u,\mathbf v)\in\mathbb Z^n\times\mathbb Z^m:
                    \mathbf M_m^T\mathbf u+\mathbf v=0\pmod q\},\\
 \Lambda_m^*&=\mathbb Z^{n+m}
       +q^{-1}\begin{pmatrix}\mathbf M_m\\\mathbf I_m\end{pmatrix}\mathbb Z^m.
 \end{split}
 \label{ud:geo:lattices}
\end{equation}
For a coefficient prefix $\mathbf z\in(\mathbb Z/q\mathbb Z)^m$, let
$\widetilde{\mathbf z}\in V_m$ be the module vector obtained by padding
with zeros and grouping the coefficients into polynomial blocks.
The systematic coordinates distinguish the dual cosets, and
$\mathbf M_m\mathbf z=\operatorname{coeff}(\mathbf A^T\widetilde{\mathbf z})$.
Thus the preceding ideal image lemma applies directly. Write
$\operatorname{coeff}_m$ for the first $m$ coordinates of the coefficient
embedding, so $\operatorname{coeff}_m(\widetilde{\mathbf z})=\mathbf z$.

Write $\theta(s)=\rho_s(\mathbb Z)$. We separate the Gaussian mass
of the prefix vectors from the average mass of one output polynomial:
\[
 \begin{aligned}
 G_{j,m}(T)&=\sum_{\mathbf z_R\in I_j^k\cap V_m}
                  \rho_{1/T}(\mathbb Z^m+\operatorname{coeff}_m(\mathbf z_R)/q),\\
 A_j(T)&=|I_j|^{-1}\rho_{q/T}((1-x)^jO)\quad(j<b_q d).
 \end{aligned}
\]

\begin{lemma}
\label{ud:alg:theta}
For $T>0$,
\begin{equation}
 \mathbb E_{\mathbf A}\!\left[\rho_{1/T}(\Lambda_m^*\setminus\mathbb Z^{n+m})\right]
 =\sum_{j=0}^{b_q d-1}A_j(T)^k
          \bigl(G_{j,m}(T)-G_{j+1,m}(T)\bigr).
 \label{ud:alg:theta-identity}
\end{equation}
Here $G_{b_q d,m}(T)=\theta(1/T)^m$. For $j=ad+r<b_q d$ and
$Q=2^{b_q-a}$,
\begin{equation}
 \begin{split}
 A_j(T)^k&\le2^{kr}Q^{-n}
       \theta(Q/T)^{n-kr}\theta(Q/(2T))^{kr},\\
 G_{j,m}(T)&\le\theta(Q/T)^{v_j(m)}
                   \theta(Q/(2T))^{c_j(m)}.
 \end{split}
 \label{ud:alg:theta-bounds}
\end{equation}
\end{lemma}
\begin{proof}
For prefix vectors of exact valuation $j$, the expected Gaussian mass of
each output polynomial is
\[
 \begin{gathered}
 |I_j|^{-1}\sum_{y\in I_j}\rho_{1/T}(\mathbb Z^d+\operatorname{coeff}(y)/q)\\
 =|I_j|^{-1}\rho_{q/T}((1-x)^jO)=A_j(T).
 \end{gathered}
\]
The $k$ output polynomials are independent. The systematic mass on
these vectors is $G_{j,m}-G_{j+1,m}$. Summing proves the identity.

For the bounds, write a coefficient vector in $(1-x)^jO$ as
$2^a\mathbf w$ with $\mathbf w\bmod2\in C_r$.
Its even and odd coordinate masses are
\[
 \alpha=\theta(Q/(2T)),\qquad
 \gamma=\theta(Q/T)-\alpha\le\alpha,
\]
where the inequality is the maximum property of a shifted Gaussian
theta series. Choose $d-r$ coordinates on which the binary code
projects bijectively. Summing over them contributes
$(\alpha+\gamma)^{d-r}$. Each remaining coordinate contributes at
most $\alpha$. Hence
\[
 \rho_{q/T}((1-x)^jO)
 \le\theta(Q/T)^{d-r}\theta(Q/(2T))^r.
\]
Divide by $|I_j|=Q^d/2^r$ and raise to the $k$th power for the first
bound. Apply the same argument to the binary prefix code
$C_r^{\,h}\times(C_r\cap\mathbb F_2^t)$, of dimension $v_j(m)$,
for the second. All these masses use the original coefficient norm.
\end{proof}

For a prefix vector of exact valuation $j=ad+r$, its additive order is
$Q=2^{b_q-a}$: after division by $2^a$, at least one coefficient
is odd. Denote the corresponding summand in
\eqref{ud:alg:theta-identity} by $D_j(m,T)$. Dropping the subtraction
and using~\eqref{ud:alg:theta-bounds} gives
\begin{equation}
 D_j(m,T)\le2^{kr}Q^{-n}
   \theta(Q/T)^{n+m-kr-c_j(m)}
   \theta(Q/(2T))^{kr+c_j(m)}.
 \label{ud:geo:stratum-bound}
\end{equation}
We next sum over the ideal chain, separating the layers into three
additive order ranges. The order two range is exceptional. A nonzero
systematic label of order two has a coordinate equal to $q/2$, so its
normalized prefix is at distance at least $1/2$ from $\mathbb Z^m$.
This distance gives Gaussian decay directly, while the intermediate
and large orders are controlled by~\eqref{ud:geo:stratum-bound}.

\begin{theorem}[Smoothing for dyadic module prefixes]
\label{ud:geo:smoothing}
Let $k$ be fixed, $n=kd\to\infty$ with $d$ a power of two, and
$q=2^{b_q}$. Choose $L_0$ with
\begin{equation}
 L_0\to\infty,\qquad L_0=o(n),\qquad \ln(nb_q)=o(L_0).
 \label{ud:geo:L-hypotheses}
\end{equation}
Before drawing $\mathbf A$, prescribe pairs $(m_i,T_i)_{i=1}^r$ satisfying
\begin{equation}
 1\le m_i\le n,\qquad 8\sqrt{L_0}\le T_i\le q,\qquad
 \max_i\frac{m_i\ln(q/T_i)}{n\ln T_i}=o(1).
 \label{ud:geo:width-hypotheses}
\end{equation}
For all sufficiently large $n$, uniformly in these pairs,
\begin{equation}
 \mathbb E_{\mathbf A}\!\left[\rho_{1/T_i}
   (\Lambda_{m_i}^*\setminus\mathbb Z^{n+m_i})\right]\le e^{-3L_0}.
 \label{ud:geo:expected-mass}
\end{equation}
Consequently, with $\varepsilon_0=4ne^{-L_0}$,
\begin{equation}
 \Pr_{\mathbf A}[\exists i:\eta_{\varepsilon_0}(\Lambda_{m_i})>T_i]
 \le\frac r{2n}e^{-2L_0}.
 \label{ud:geo:failure}
\end{equation}
\end{theorem}
\begin{proof}
Fix $(m,T)$. Write each index as $j=ad+\rho$, with $0\le\rho<d$,
and set $Q=2^{b_q-a}$.
There are at most $b_q d$ nonzero layers, which we divide into three ranges.

\emph{Order two.} For $Q=2$, put $u=d-\rho$ and $v=v_j(m)$.
The output factor in~\eqref{ud:alg:theta-bounds} is
\[
 2^{-ku}\theta(2/T)^{ku}\theta(1/T)^{k\rho}=2^{-ku}(1+o(1)).
\]
A nonzero coefficient prefix has the form $\mathbf z=(q/2)\mathbf w$
with a nonzero binary vector $\mathbf w$.
Every point of $\mathbb Z^m+\mathbf w/2$ has norm at least $1/2$.
Splitting its Gaussian exponent into equal halves yields
\begin{equation}
 \rho_{1/T}(\mathbb Z^m+\mathbf w/2)
 \le e^{-\pi T^2/8}\theta(\sqrt2/T)^m.
 \label{ud:geo:order-two-norm}
\end{equation}
The containing prefix code has $2^v$ vectors, and
$v\le ku$ by~\eqref{ud:alg:module-prefix}. Their count
cancels the factor $2^{-ku}$, giving
\begin{equation}
 D_j(m,T)\le(1+o(1))e^{-\pi T^2/8}
            \le(1+o(1))e^{-8\pi L_0}.
 \label{ud:geo:order-two-bound}
\end{equation}
The factors $1+o(1)$ are uniform, since
$\theta(c/T)=1+O(e^{-\pi T^2/c^2})$ for fixed $c>0$,
$T^2\ge64L_0$, and $\ln n=o(L_0)$.

\emph{Intermediate orders.} For $4\le Q\le T$, use
$\theta(s)\le e^{1/8}$ for $0<s\le1$ in
\eqref{ud:geo:stratum-bound} to obtain
\begin{equation}
 D_j(m,T)\le2^{k\rho}Q^{-n}e^{(n+m)/8}
             \le e^{-(\ln2-1/4)n}=e^{-\omega(L_0)}.
 \label{ud:geo:intermediate-bound}
\end{equation}
Here $k\rho\le n$, $m\le n$, and $L_0=o(n)$.

\emph{Large orders.} For $Q>T$, Poisson summation gives
$\theta(Q/T)\le2Q/T$, and monotonicity bounds the smaller theta
factor by the larger one. Thus
\begin{align}
 D_j(m,T)&\le2^{3n}Q^mT^{-(n+m)},\nonumber\\
 \ln D_j(m,T)&\le3n\ln2+m\ln(q/T)-n\ln T
              =-(1-o(1))n\ln T.
 \label{ud:geo:large-bound}
\end{align}
This is $e^{-\omega(L_0)}$ uniformly by the prescribed-pair conditions.

Summing over at most $b_q d$ layers proves the expected mass bound,
since $\ln(b_q d)=o(L_0)$. Lemma~\ref{up:mass-to-smoothing}
gives the simultaneous smoothing bound.
\end{proof}

\begin{corollary}
\label{ud:geo:schedule}
Assume~\eqref{ud:geo:L-hypotheses}. A deterministic schedule satisfying
\[
 \begin{gathered}
 m_i=\sum_{j\le i}b_j\le n,\qquad
 16\sqrt{L_0}\le t_i\le q,\\
 m_i\ln(q/t_i)=O(\ln N\ln t_i),\qquad \ln N=o(n)
 \end{gathered}
\]
obeys~\eqref{ud:geo:failure} with $T_i=t_i/2$.
\end{corollary}
\begin{proof}
The width bound gives $T_i\ge8\sqrt{L_0}$, and
\[
 \frac{m_i\ln(q/T_i)}{n\ln T_i}
 \le O(\ln N/n)\frac{\ln t_i}{\ln T_i}
       +\frac{\ln2}{\ln T_i}=o(1).
\]
Apply Theorem~\ref{ud:geo:smoothing}.
\end{proof}

\subsection{Recovery and the effect of rounding}

\begin{corollary}
\label{ud:recovery}
\label{ua:saber-recovery}
Consider the rounded relation~\eqref{ud:rnd:map}, with fixed module
rank $k$, $d\to\infty$ through powers of two, and
$\ln q\le(\ln n)^c$ for a fixed $c$. Fix a rational $V\ge1$ and an
integer $1\le B\le\lfloor\sqrt{nV}\rfloor$ of polynomial bit length
before drawing $\mathbf A$, and assume
\[
 \Pr[\|\mathbf s\|_\infty>B\ \text{or}\
       \|\mathbf s\|^2+\|\mathbf e_{\rm rnd}\|^2>nV]\le p_E.
\]
Assume also that
\begin{equation}
 R^2=nV+3B^2,\qquad U=\sqrt{R^2/n},\qquad
 D=\ln\frac{\ln(q/\sqrt n)}{1+\ln U}\longrightarrow\infty.
 \label{ud:depth}
\end{equation}
Then the original integer secret is recovered from exactly $k$ rounded
ring equations with probability at least $1-p_E-o(1)$, in expected time
and space $2^{(1/2+o(1))n/D}$. The secret may depend on $\mathbf A$.
\end{corollary}
\begin{proof}
Lemma~\ref{ud:rnd:lift} supplies the integer relation. For the
schedule of Theorem~\ref{ue:main},
the bounds in~\eqref{ue:geometry-scale} give
$L_0=\Theta(n/D)$, $L_0=o(n)$, and $\ln(nb_q)=o(L_0)$.
Corollary~\ref{ud:geo:schedule} therefore supplies its prefix smoothing
with failure probability at most $r e^{-2L_0}/(2n)=o(1)$.
Theorem~\ref{ue:main} now gives
recovery uniformly over the stated energy event.
\end{proof}

For the SABER-type support $|s_j|\le\eta$ and $\Delta\ge2$, take
$B=\eta$, $V=\eta^2+\Delta^2/4$, and $p_E=0$. In particular,
\begin{equation}
 \begin{gathered}
 q=n^{\kappa+o(1)},\quad \kappa>1/2,\quad
 \eta+\Delta=(\ln n)^{O(1)}\\
 \Longrightarrow\quad D=(1+o(1))\ln\ln n.
 \end{gathered}
 \label{ud:rounding-windows}
\end{equation}
These families are nonempty: choose
$q=2^{\lceil\kappa\log_2(kd)\rceil}$ and a fixed or polylogarithmic
power-of-two divisor $\Delta\ge2$ of $q$. The resulting expected time and
space are $2^{(1/2+o(1))n/\ln\ln n}$.
To quantify coarser rounding, keep $q=n^{\kappa+o(1)}$ with fixed
$\kappa>1/2$, and assume $\eta=O(\Delta)$. Then $U=\Theta(\Delta)$.
If $\ln\Delta=(\ln n)^{a+o(1)}$ for fixed $0<a<1$, the depth is
$D=(1-a+o(1))\ln\ln n$. If instead $\Delta=n^{a+o(1)}$ for fixed
$a>0$, then
\[
 D=\ln\frac{\kappa-1/2}{a}+o(1).
\]

\begin{corollary}[SABER family specialization]
\label{ua:saber-headline}
Let $n=kd$, with fixed $k$ and power-of-two $d\to\infty$, and
take power-of-two moduli $p\mid q$ with
$q=n^{\kappa+o(1)}$, fixed $\kappa>1/2$.
If $|s_j|\le\eta$ and $\Delta=q/p\ge2$ are polylogarithmic,
then the original integer secret is recovered from the $k$ rounded
ring equations for uniform $\mathbf A\in R_q^{k\times k}$ with
probability $1-o(1)$, in expected time and space
$2^{(1/2+o(1))n/\ln\ln n}$.
\end{corollary}
\begin{proof}
The public lift has an error satisfying
$\|\mathbf e_{\rm rnd}\|^2\le n\Delta^2/4$.
Thus $B=\eta$, $V=\eta^2+\Delta^2/4$, and
$U^2=V+3\eta^2/n$ satisfy the energy condition deterministically.
Here $1\le U=(\ln n)^{O(1)}$ and
$\ln(q/\sqrt n)=(\kappa-1/2+o(1))\ln n$, so
$D=\ln\ln n+O(\ln\ln\ln n)\sim\ln\ln n$.
Corollary~\ref{ud:recovery} supplies the result, including the
dependence of the rounding error on the public matrix and secret.
This proves Corollary~\ref{ui:saber-cor}.
\end{proof}

\section{NTRU LPRime: Recovery from Rounded Field Products}
\label{un:section}
\label{ua:lprime}

We use $p$ for the ring degree and $n=p$ for the recovery dimension.
A single ring product contains $n$ coefficient observations, but all
entries of its multiplication matrix depend on the same public
multiplier. In this case the dependence is easy to handle: over a
field, multiplying any fixed nonzero element by a uniform element
gives a uniform result. This is exactly the property required by
Lemma~\ref{ug:uniform-mass}.

\subsection{The public multiplier and its Gaussian mass}
\label{un:alg:section}
\label{un:geo:section}

Let $p\ge1$ be an integer, let $q\ge5$ be prime, and assume that
$X^p-X-1$ is irreducible modulo
$q$. Set $O=\mathbb Z[X]/(X^p-X-1)$ and
$R_q=O/qO=\mathbb F_q[X]/(X^p-X-1)$, and draw the ring element $G$
uniformly in $R_q$. In the basis $(1,X,\ldots,X^{p-1})$, let
$\mathbf M_G$ be the coefficient matrix of multiplication by $G$.
Thus $\operatorname{coeff}(Gc)=\mathbf M_G\operatorname{coeff}(c)$
for $c\in R_q$. If $\mathbf E_m$ appends zeros to a vector of length
$m$, let $\mathbf M_m$ consist of the first $m$ columns of $\mathbf M_G$.
Then
\begin{equation}\label{un:alg:actual-prefix}
 \mathbf M_m=\mathbf M_G\mathbf E_m,\qquad
 \mathbf M_m\mathbf z=\mathbf M_G\mathbf E_m\mathbf z.
\end{equation}
Thus the prefix lattice is
\begin{equation}\label{un:geo:lattices}
 \Lambda_m(G)=
 \{(\mathbf u,\mathbf v)\in\mathbb Z^n\times\mathbb Z^m:
       \mathbf M_m^T\mathbf u+\mathbf v=0\pmod q\}.
\end{equation}

\begin{lemma}
\label{un:alg:uniform}
For each prescribed $1\le m\le n$ and each fixed
$0\ne\mathbf z\in\mathbb F_q^m$, the vector
$\mathbf M_m\mathbf z$ is uniform in
$\mathbb F_q^n$.
\end{lemma}
\begin{proof}
The polynomial represented by $\mathbf E_m\mathbf z$ has degree less than $p$ and is
nonzero. It is therefore a nonzero element of the field $R_q$.
Multiplication by it permutes $R_q$.
\end{proof}

Applying Lemma~\ref{ug:uniform-mass} with
$\mathbf A_m=\mathbf M_m$ gives
\begin{equation}\label{un:geo:exact-mass}
 \mathbb E_G\!\left[\rho_{1/T}
   (\Lambda_m(G)^*\setminus\mathbb Z^{n+m})\right]
 =q^{-n}\theta(q/T)^n
       \bigl(\theta(q/T)^m-\theta(1/T)^m\bigr).
\end{equation}
Corollary~\ref{ug:uniform-geometry} now supplies the smoothing
bound for the sampler's prescribed prefixes. The argument takes place
entirely in coefficient coordinates and applies to any irreducible
defining polynomial.

\subsection{Rounding gives a small error}

The NTRU LPRime Core relation~\cite[Sections~2.1 and~2.4.2]{NTRUPrimeSpec}
uses a ternary secret of weight $w$, with $0\le w\le p$:
\begin{equation}\label{un:rnd:secret-set}
 \mathcal S_w=\{a\in O:\operatorname{coeff}(a)\in\{-1,0,1\}^n,
                  \ \|\operatorname{coeff}(a)\|_0=w\}.
\end{equation}
For $a\in\mathcal S_w$, write $\mathbf a=\operatorname{coeff}(a)$.
Products with $G$ use the reduction of $a$ modulo $q$.
For an integer $z$, let $\operatorname{cen}_q(z)$ denote its
representative in $[-(q-1)/2,(q-1)/2]$. Define nearest rounding to a
multiple of three by
\begin{equation}\label{un:rnd:scalar-map}
 \operatorname{rnd}_q(z)=
 3\left\lfloor\frac{\operatorname{cen}_q(z)+1}{3}\right\rfloor.
\end{equation}
For a ring element $c\in R_q$, define $\operatorname{Round}(c)\in O$
by applying $\operatorname{rnd}_q$ to each coefficient. Thus
$\operatorname{Round}:R_q\to O$ has an integer ring output.
The public ring element and its integer coefficient observation are
\begin{equation}\label{un:rnd:public-map}
 A=\operatorname{Round}(aG)\in O,\qquad
 \mathbf y=\operatorname{coeff}(A)\in\mathbb Z^n.
\end{equation}

\begin{lemma}\label{un:rnd:lift}
Let $x_j=\operatorname{cen}_q((\mathbf M_G\mathbf a)_j)$ and let $[\,\cdot\,]_3$
denote the least nonnegative residue modulo three. Then
\begin{equation}\label{un:rnd:residue}
 \mathbf y=\mathbf x+\mathbf e=\mathbf M_G\mathbf a+\mathbf e\pmod q,\qquad
 e_j=1-[x_j+1]_3\in\{-1,0,1\}.
\end{equation}
In particular,
\begin{equation}\label{un:rnd:joint-energy}
 \|\mathbf a\|^2+\|\mathbf e\|^2\le w+n.
\end{equation}
\end{lemma}
\begin{proof}
Write $x_j+1=3t_j+[x_j+1]_3$. Rounding returns $y_j=3t_j$,
so $y_j-x_j=1-[x_j+1]_3$. Summing squares proves the bound.
\end{proof}

We retain the rounded integer output when forming $\mathbf e$. For
$q\equiv5\pmod6$, its extreme values can lie one unit outside the
centered interval. Reducing $\mathbf y$ modulo $q$ leaves the same bounded
error valid in the modular equation.

\subsection{Recovery from the given product}

The energy bound supplies the common theorem with
\begin{equation}\label{un:candidate-radius}
 V=1+w/n,\qquad B=1,\qquad R^2=n+w+3.
\end{equation}

\begin{corollary}
\label{un:recovery}
\label{ua:lprime-recovery}
Fix a parameter sequence $(p,q,w)$ as above, with $n\to\infty$,
before drawing $G$, and assume
\[
 D_0=\ln\ln(q/\sqrt n)\longrightarrow\infty,\qquad
 \ln q\le(\ln n)^c
\]
for a fixed $c$. Given $G$ and $A=\operatorname{Round}(aG)$,
the algorithm recovers the integer coefficient vector $\mathbf a$, and
hence the ring element $a$, with probability
$1-o(1)$, in expected time and space $2^{(1/2+o(1))n/D_0}$.
The secret may depend on $G$, provided $a\in\mathcal S_w$.
\end{corollary}
\begin{proof}
Lemma~\ref{un:rnd:lift} gives the lifted relation with $h=n$
and the deterministic certificate $V=1+w/n$, $B=1$.
Thus $R^2=n+w+3$ and $U^2=1+w/n+3/n=O(1)$, giving
$D=D_0-O(1)\sim D_0$.
For every prescribed prefix, Lemma~\ref{un:alg:uniform} gives
a uniform image for each fixed nonzero label.
The schedule bounds~\eqref{ue:geometry-scale} therefore allow
Corollary~\ref{ug:uniform-geometry} to be applied with $T_i=t_i/2$.
Its geometry failure probability is at most $r e^{-2L_0}/(2n)=o(1)$.
Theorem~\ref{ue:main} gives the stated recovery probability and
resource bounds, including secrets that depend on $G$.
\end{proof}

In particular, $q=\Theta(n)$ gives $D_0\sim\ln\ln n$.

For the LPRime family, the supplied degree $p$ is also prime and the
parameters satisfy the specification's arithmetic requirements.
For prime $p$, irreducibility of $f=X^p-X-1$ is checked by
$f\mid X^{q^p}-X$ and $\gcd(f,X^q-X)=1$, using modular polynomial
arithmetic.

\begin{corollary}[NTRU LPRime family specialization]
\label{ua:lprime-headline}
Let $(p,q,w)$ be a supplied growing family satisfying the LPRime
arithmetic requirements, with primes $p,q$, $q=\Theta(p)$,
$X^p-X-1$ irreducible modulo $q$, and $0\le w\le p$.
For uniform $G\in R_q$ and $a\in\mathcal S_w$, the original
ternary secret is recovered from $A=\operatorname{Round}(aG)$
with probability $1-o(1)$, in expected time and space
$2^{(1/2+o(1))n/\ln\ln n}$, where $n=p$.
\end{corollary}
\begin{proof}
Lemma~\ref{un:rnd:lift} gives deterministic joint energy $w+n$.
The candidate radius in~\eqref{un:candidate-radius} has
$U^2=1+w/n+3/n=O(1)$. Since $q=\Theta(n)$,
\[
 D_0=\ln\ln(q/\sqrt n)=\ln\ln n+O(1),\qquad
 D=D_0-O(1)\sim\ln\ln n.
\]
Corollary~\ref{un:recovery} applies to the $n$ coefficient
observations in this one rounded product. It recovers the integer
coefficient vector, hence $a$ itself.
This proves Corollary~\ref{ui:lprime-cor}.
\end{proof}

\section{Dilithium/ML-DSA: Recovery from Compressed Public Keys}
\label{umldsa:section}

The ML-DSA public key contains the high bits of a noisy module
product. Its discarded low bits become part of the error in the lifted relation,
so the same prime-module geometry used for Kyber/ML-KEM applies.
Selecting a square block of public equations recovers the original
vector $\mathbf s_1$; the remaining short algebraic components can
then be completed coefficientwise.

\subsection{The public relation and its integer lift}

We use the key generation relation and the \texttt{Power2Round} map
of~\cite[Algorithms~6 and~35]{FIPS204}. Write $m$ for the ring degree
and $d$ for the number of discarded public-key bits, and set
\[
 R=\mathbb Z[X]/(X^m+1),\qquad R_q=R/qR,\qquad \Delta=2^d.
\]
In the specified parameter sets, $m=256$, $d=13$, and the module
dimensions $(k,\ell)$ are $(4,4),(6,5),(8,7)$.
The growing family below keeps the module dimensions fixed while the
ring degree and modulus grow.

\begin{definition}[ML-DSA public family]
\label{umldsa:family}
Fix integers $k\ge\ell\ge1$. Let $m\to\infty$ through powers of
two, put $n=\ell m$, and take a supplied prime modulus
$q\equiv1\pmod{2m}$ with $q\le n^{C_q}$ for fixed $C_q$.
Choose public integers $\eta\ge1$, $d\ge1$, with
$\eta,\Delta=2^d=(\ln n)^{O(1)}$, before the public draw.
The bit count $d$ may be fixed; if it grows, it satisfies
$d=O(\ln\ln n)$.
Draw $\mathbf A\leftarrow\operatorname{Unif}(R_q^{k\times\ell})$ and let
$\mathbf s_1\in R^\ell$, $\mathbf s_2\in R^k$ have coefficient
absolute values at most $\eta$. They may depend on $\mathbf A$.
The public data are $(\mathbf A,\mathbf t_1)$, where
\begin{equation}
 \mathbf t=\mathbf A\mathbf s_1+\mathbf s_2\pmod q,
 \qquad (\mathbf t_1,\mathbf t_0)
          =\operatorname{Power2Round}_{d}(\mathbf t)
 \label{umldsa:public-relation}
\end{equation}
coefficientwise, with $\mathbf t\in R_q^k$ and integer outputs
$\mathbf t_1,\mathbf t_0\in R^k$. The map
$\operatorname{Power2Round}_d$ is the standard function parameterized
by its discarded-bit count $d$, with the scalar convention below.
\end{definition}

For $r\in\mathbb Z_q$, let $r^+\in\{0,\ldots,q-1\}$ be its
least nonnegative representative and define
\begin{equation}
 r_1=\left\lfloor\frac{r^++\Delta/2-1}{\Delta}\right\rfloor,
 \qquad r_0=r^+-\Delta r_1,
 \qquad \operatorname{Power2Round}_{d}(r)=(r_1,r_0).
 \label{umldsa:power2round}
\end{equation}
Thus $-\Delta/2<r_0\le\Delta/2$, equivalently
$-2^{d-1}<r_0\le2^{d-1}$, and
$r=\Delta r_1+r_0\pmod q$, including the tie convention in the
specification. The first output is stored as an integer.
In the encoded public key, $\mathbf A$ is reconstructed from its
public seed; an NTT representation is returned to coefficient form.

\begin{lemma}[Lifted error bound]
\label{umldsa:lift}
Define
$\mathbf y_R:=2^d\mathbf t_1$ and
$\mathbf e_R:=\mathbf s_2-\mathbf t_0$.
Then
\begin{equation}
 \mathbf y_R=\mathbf A\mathbf s_1+\mathbf e_R\pmod q,
 \qquad \|\mathbf e_R\|_\infty\le\eta+\Delta/2.
 \label{umldsa:lifted-relation}
\end{equation}
Fix any set $J\subseteq[k]$ of $\ell$ rows before drawing
$\mathbf A$, and let $\mathbf A_J$ denote the corresponding restriction.
Set
$\mathbf y:=\operatorname{coeff}((\mathbf y_R)_J)$,
$\mathbf e:=\operatorname{coeff}((\mathbf e_R)_J)$, and
$\mathbf s:=\operatorname{coeff}(\mathbf s_1)$.
Then $\mathbf A_J$ is uniform in
$R_q^{\ell\times\ell}$, and its coefficient matrix $\mathbf M_J$
gives $n=\ell m$ observations with deterministic certificate
\begin{equation}
 \begin{gathered}
 \mathbf y=\mathbf M_J\mathbf s+\mathbf e\pmod q,\\
 \|\mathbf s\|_2^2+\|\mathbf e\|_2^2\le nV,
 \qquad V=\eta^2+(\eta+\Delta/2)^2,
 \qquad B=\eta.
 \end{gathered}
 \label{umldsa:energy}
\end{equation}
\end{lemma}
\begin{proof}
Subtract $\mathbf t_0$ from~\eqref{umldsa:public-relation} and use
\eqref{umldsa:power2round}. The coefficient bound on $\mathbf s_2$
and $|r_0|\le\Delta/2$ give~\eqref{umldsa:lifted-relation}.
A predetermined restriction to $\ell$ rows of a uniform ring matrix
is a uniform square ring matrix. Both $\mathbf s_1$ and the
restricted error have $n$ coefficients, so summing their squared
bounds gives~\eqref{umldsa:energy}. These statements hold for every
realized secret and rounding output.
\end{proof}

\subsection{Recovery of the original secret component}

\begin{theorem}[ML-DSA secret recovery]
\label{umldsa:recovery}
For the family in Definition~\ref{umldsa:family}, a randomized
algorithm recovers the original integer coefficient vector of
$\mathbf s_1$ from $(\mathbf A,\mathbf t_1)$ with probability
$1-o(1)$ over the public matrix and independent algorithmic
randomness. Its expected time and space are
\begin{equation}
 2^{(1/2+o(1))n/\ln\ln n},\qquad n=\ell m.
 \label{umldsa:complexity}
\end{equation}
\end{theorem}
\begin{proof}
Fix $J=\{1,\ldots,\ell\}$, so the algorithm uses $\ell$ of the
available $k$ ring equations. Lemma~\ref{umldsa:lift} supplies the
uniform square module operator and its energy certificate with
$p_E=0$. The modulus congruence gives splitting degree
$\operatorname{ord}_{2m}(q)=1$, so
Theorem~\ref{um:recovery} applies with its degree parameter $d$
set to the present $m$ and its module rank set to $\ell$.
The dependence of $\mathbf e$ on the public matrix and the secret
is included in that theorem.

For the depth in Theorem~\ref{um:recovery}, put
\[
 U=\sqrt V,\qquad
 D=\ln\frac{\ln(q/\sqrt n)}{1+\ln U}.
\]
The bounds $q\ge2m+1$ and $q\le n^{C_q}$ imply
$\ln(q/\sqrt n)=\Theta(\ln n)$, since $\ell$ is fixed.
Moreover $1\le U=(\ln n)^{O(1)}$, whence
\[
 D=\ln\ln n+O(\ln\ln\ln n)\sim\ln\ln n.
\]
The theorem therefore recovers the original $\mathbf s_1$ with the
stated probability and resource bound. Forming the public lift,
selecting the rows, and reading the remaining public data cost only
polynomial time. This proves Corollary~\ref{ui:mldsa-cor}.
\end{proof}

\subsection{Completion of the short algebraic components}

The recovered vector also determines a compatible choice of the
other short algebraic components of the public-key relation.

\begin{proposition}[Coefficientwise completion]
\label{umldsa:completion}
Given $(\mathbf A,\mathbf t_1)$ and the recovered original
$\mathbf s_1$, a deterministic polynomial-time algorithm finds
$\mathbf s_2',\mathbf t_0'\in R^k$ such that
\begin{equation}
 \begin{gathered}
 \|\mathbf s_2'\|_\infty\le\eta,\qquad
 -\Delta/2<(t_0')_j\le\Delta/2\quad\text{for every coefficient }j,\\
 \operatorname{Power2Round}_{d}(\mathbf A\mathbf s_1+\mathbf s_2')
       =(\mathbf t_1,\mathbf t_0'),\\
 \mathbf A\mathbf s_1+\mathbf s_2'
       =\Delta\mathbf t_1+\mathbf t_0'\pmod q.
 \end{gathered}
 \label{umldsa:completion-output}
\end{equation}
Thus the tuple $(\mathbf s_1,\mathbf s_2',\mathbf t_0')$ has the
specified coefficient bounds and reproduces the same compressed
public key. Its construction and all checks use polynomial bit complexity.
\end{proposition}
\begin{proof}
Compute $\mathbf a=\mathbf A\mathbf s_1\pmod q$. For each
coefficient position $j$, enumerate $z\in\{-\eta,\ldots,\eta\}$
and retain the first one for which the first component of
$\operatorname{Power2Round}_{d}(a_j+z)$ equals $(t_1)_j$.
The original coefficient $(s_2)_j$ is a valid choice, so this search
always succeeds. Set $(s_2')_j=z$ and let $(t_0')_j$ be the second
component of the same rounding operation. Equation~\eqref{umldsa:power2round}
gives all the asserted bounds and identities. There are $km$
positions and $2\eta+1$ trials per position, with $O(\ln q+\ln\Delta)$-bit
rounding arithmetic. Since $\eta$ is polylogarithmic, the work is
polynomial. Recomputing the rounded product verifies the output exactly.
\end{proof}

\section{Recovering Short Pairs from Public Quotients}
\label{uq:overview}

For a public matrix $\mathbf H\in(\mathbb Z/q\mathbb Z)^{m\times m}$,
define
\[
 \Lambda_{\mathbf H}=\ker_q[\mathbf H\mid-\mathbf I_m]
 =\{(\mathbf u,\mathbf v)\in\mathbb Z^{2m}:
                    \mathbf H\mathbf u=\mathbf v\pmod q\}.
\]
The Falcon parameter family uses the multiplication matrix of $h$.
Cyclic NTRU uses that of $\widehat h=3^{-1}h\bmod q$.
In both cases, recovery seeks a nonzero short pair passing public
algebraic checks.

\paragraph{Algorithmic picture.}
The algorithm fills the coordinates of $\mathbf u$ one at a time,
completes $\mathbf v$ from the public relation, and checks the resulting
pair. Dual Gaussian lists estimate the weights used to choose each
coordinate. The public mass bound sets the required accuracy and
number of repetitions.

An anchor is a fixed positive coordinate $u_i=a$, with
$1\le a\le\lfloor B\rfloor$. Since $B<q/2$, every nonzero pair of
norm at most $B$ has $\mathbf u\ne\mathbf0$, so either it or its
negative agrees with some anchor. For each position $i$, fix an order
for the other coordinates. These $m$ orders use at most $m$ sets of
free coordinates each, hence at most $m^2$ sets in total. Each set has
its own dual Gaussian list, reused across coordinate values and
repetitions.

Suppose the current free set is $S$ and the next coordinate is $u_j$.
For a candidate value $b$, let $A_b$ contain all lattice pairs that
agree with the coordinates already fixed and with $u_j=b$. Its
Gaussian mass sums the weights of all possible completions. The
ideal choice is
\[
 \Pr[u_j=b\mid\text{fixed coordinates}]
 \propto\rho_{s_{\rm search}}(A_b),
\]
over integers $b\in[-\lfloor B\rfloor,\lfloor B\rfloor]$ whose fixed
coordinates have squared norm at most $B^2$. All these child slices
have the direction lattice for $S\setminus\{j\}$. Their estimates use
the same list with different shifts in the cosine average. The
common mass denominator cancels when the weights are normalized.
The algorithm samples one value from the estimated weights and
continues along that path.

Once $\mathbf u$ is complete, set
$\mathbf v=\operatorname{lift}_q^{\rm cen}(\mathbf H\mathbf u)$.
Retain the pair and its negative if the lattice, nonzero, and norm
checks pass. Repeat the search $T$ times for every anchor using the
same lists and fresh choices. Scan the resulting table with the
public verifier to obtain a pair satisfying the application's
algebraic conditions.

For a fixed short witness, the ideal transition probabilities are
bounded below by ratios of successive slice masses. Their product
telescopes to a lower bound from the witness's Gaussian weight and
the anchor mass. The estimates preserve this bound up to a constant
factor, and repetition gives high success probability.
Algorithm~\ref{alg:affine-slice-search} gives the search procedure.
Sections~\ref{uf:persistent-lists} and~\ref{uf:affine-recovery}
of Appendix~\ref{uf:section} give the list construction and search
analysis.

\subsection{From Gaussian mass to verifiable search}

The following criterion combines the public mass bound with the list
construction, search, and verification steps above.

Fix computable rational constants
$0<c<1$, $1/2<\beta<1$, and $A_q,A_0,C_B,c_{\rm search}>0$.
Use the finite schedule of Theorem~\ref{uf:persistent-theorem}, with
\[
 \begin{gathered}
 0\le\ln\ln m-D_{\rm Q}\le2^{-m^2},\qquad
 (1-2^{-m^2})q/m^\beta\le s_{\rm mass}^2\le q/m^\beta,\\
 s_{\rm search}^2=c_{\rm search} qD_{\rm Q}/m.
 \end{gathered}
\]
The quantities $D_{\rm Q},s_{\rm mass}^2$ are positive dyadic rationals.
The widths $s_{\rm mass}$ and $s_{\rm search}$ control the mass bound
and the search distribution, respectively. The constant $c_{\rm search}>0$ is a
fixed scheduling constant.

The supplied public bound $A_{\rm mass}$ on
$\rho_{s_{\rm mass}}(\Lambda_{\mathbf H})$ sets the list accuracy
and number of search repetitions.

\begin{theorem}[Mass criterion]
\label{uq:mass-search-framework}
Let $m\to\infty$ and
\[
 2m+1\le q\le m^{A_q},\qquad
 \ln(q/m)\le(\ln m)^{1-c}.
\]
Given $\mathbf H$, public rational $A_{\rm mass}\ge1$ and $B_{\rm Q}^2$
of polynomial bit length, and a polynomial-time verifier $\mathsf V$,
suppose
\[
 \begin{gathered}
 \rho_{s_{\rm mass}}(\Lambda_{\mathbf H})\le A_{\rm mass},\qquad
 \ln A_{\rm mass}\le A_0m/D_{\rm Q},\\
 s_{\rm mass}^2\ge8s_{\rm search}^2,
 \end{gathered}
\]
and $1\le B_{\rm Q}\le C_B\sqrt q<q/2$.
If $\mathsf V$ accepts some nonzero
$\mathbf z\in\Lambda_{\mathbf H}$ with $\|\mathbf z\|\le B_{\rm Q}$,
then for every fixed $K>0$ a finite randomized algorithm outputs such
a vector with probability at least $1-m^{-K}$, in time and space
$2^{O(m/D_{\rm Q})}=2^{O(m/\ln\ln m)}$.
\end{theorem}
\begin{proof}
Monotonicity gives $\rho_{s_{\rm search}}(\Lambda_{\mathbf H})\le A_{\rm mass}$.
For the search of Theorem~\ref{uf:candidate-table}, take
$A=A_{\rm mass}$, $B=B_{\rm Q}$, and
\[
 p_B=A_{\rm mass}^{-1}e^{-\pi B_{\rm Q}^2/s_{\rm search}^2}
 \ge \exp\!\left(-(A_0+\pi C_B^2/c_{\rm search})m/D_{\rm Q}\right).
\]
Allocate $m^{-(K+2)}$ to each of preprocessing, list implementation,
categorical sampling, query error, and missed-witness error.
With $H=\lfloor B_{\rm Q}\rfloor$, set
$\xi=p_B/[200m(2H+1)]$.
Equations~\eqref{uf:search-repetitions} and~\eqref{uf:query-list-size}
allow a list length $N=p_B^{-2}\operatorname{poly}(m)$ and a repetition
count $T=p_B^{-1}\operatorname{poly}(m)$ at these error budgets.
Each slice query scans $N$ samples. There are at most
$m^2H(2H+1)T$ queries, giving total query cost
$NT\operatorname{poly}(m)$, including finite precision.
Since $\ln(1/p_B)=O(m/D_{\rm Q})$ and $H$ is polynomial in $m$,
this cost is $\exp(O(m/D_{\rm Q}))$, and
$\ln(1/\xi)=O(m/D_{\rm Q})$.
Choose the constants $a,A_N$ in
Theorem~\ref{uf:persistent-theorem} so that its lists meet these
bounds and $\delta_{\rm out}\le\xi/16$.
Section~\ref{uf:affine-recovery} implements these bounds by rounding
accuracy thresholds downward and list sizes and repetition counts
upward, with fixed constant margins, using the finite sampling tools
of Section~\ref{ue:common-sampling}.

Fix one accepted witness for the analysis.
Theorem~\ref{uf:candidate-table} puts it and its negative in the
candidate table except with the sum of the five budgets, at most
$m^{-K}$ for large $m$. Scan the table and return the first entry
accepted by $\mathsf V$. Lattice membership, nonzeroness, and the norm
bound are checked exactly. Retaining both signs permits an arbitrary
verifier. The table and all checks cost $\exp(O(m/D_{\rm Q}))$.
\end{proof}

Algorithm~\ref{alg:affine-slice-search} uses the prepared dual lists
to produce the candidate table.

\begin{algorithm}[H]
\caption{\textsc{AffineSliceSearch}: short pair candidate table}
\label{alg:affine-slice-search}
\KwIn{A public matrix $\mathbf H\in(\mathbb Z/q\mathbb Z)^{m\times m}$, dual Gaussian lists, a norm bound $1\le B<q/2$, and repetition count $T$.}
\KwOut{A table $\mathcal C(\mathbf H)$ of nonzero short pairs in $\Lambda_{\mathbf H}$.}
Initialize $\mathcal C(\mathbf H)\gets\varnothing$\;
\ForEach{coordinate $i\in[m]$ and positive anchor $a\in[1,\lfloor B\rfloor]$}{
  \For{$t\gets1$ \KwTo $T$}{
    Set $u_i=a$, leave the other $\mathbf u$ coordinates free, and use fresh randomness for this run\;
    \While{some $\mathbf u$ coordinate remains free}{
      Evaluate the affine slice masses for next values in
      $[-\lfloor B\rfloor,\lfloor B\rfloor]$ whose fixed coordinates have squared norm at most $B^2$\;
      \If{the total nonnegative certified mass is zero}{
        Declare this run failed and continue with the next repetition\;
      }
      Sample the next value proportionally to the nonnegative certified masses\;
    }
    Set $\mathbf v\gets\operatorname{lift}_q^{\rm cen}(\mathbf H\mathbf u)$\;
    \If{$(\mathbf u,\mathbf v)\in\Lambda_{\mathbf H}\setminus\{\mathbf0\}$ and $\|(\mathbf u,\mathbf v)\|\le B$, checked exactly}{
      Add $(\mathbf u,\mathbf v)$ and its negative to $\mathcal C(\mathbf H)$\;
    }
  }
}
\Return{$\mathcal C(\mathbf H)$}\;
\end{algorithm}

\subsection{Falcon: recovery from the key generation quality bound}

In this subsection, the Falcon parameter family uses
$R=\mathbb Z[x]/(x^n+1)$ for a power of two $n$ and a prime
$q\equiv1\pmod{2n}$. A quality-qualified public key admits the
algebraic and quality witness in~\eqref{ui:falcon-quality-family}.
The quality quantity $\gamma_q$ is defined in~\eqref{uf:falcon-quality}.

\begin{theorem}[Falcon recovery]
\label{uf:falcon-quality-recovery}
Fix rational constants $C\ge1$ and $0<c<1$. Let $n=2^k\to\infty$
and let $q\equiv1\pmod{2n}$ be prime, with
\[
 \ln(q/n)\le(\ln n)^{1-c}.
\]
For every fixed $K>0$, there is a classical randomized
algorithm with the following guarantee for all sufficiently large $n$.
Given only $(n,q,h)$, $C$, and the prescribed public parameters,
suppose $h\in R_q$ admits $f,g,F_0,G_0\in R$ satisfying
\[
 \begin{gathered}
 f\bmod q\in R_q^\times,\qquad h=gf^{-1}\pmod q,\\
 fG_0-gF_0=q,\qquad \gamma_q(f,g)^2\le C^2q.
 \end{gathered}
\]
The algorithm returns $u,v,F,G\in R$ such that
\begin{equation}
 \begin{gathered}
 u\bmod q\in R_q^\times,\qquad vu^{-1}=h\pmod q,\\
 uG-vF=q,\qquad \gamma_q(u,v)^2\le C^2q,
 \end{gathered}
 \label{uf:quality-recovery-output}
\end{equation}
with probability at least $1-n^{-K}$ over its own randomness,
uniformly for every such public key. In particular,
$\|(u,v)\|_2\le C\sqrt q$, and the completed coefficient basis has
maximum Gram--Schmidt norm at most $C\sqrt q$.
Its time and space are $2^{O(n/\ln\ln n)}$.
The hidden constants may depend on $C,c,K$ and the fixed public
schedule constants.
\end{theorem}

\begin{proof}
The quality witness supplies a full basis with Gram--Schmidt lengths
at most $C\sqrt q$. The public lattice satisfies
$q\Lambda_h^*=\mathbf J\Lambda_h$ for an orthogonal map $\mathbf J$.
Its reversed dual basis therefore supplies another basis of $\Lambda_h$
with Gram--Schmidt lengths at least $\sqrt q/C$. Successive
one-dimensional Gaussian sums give
\[
 \rho_s(\Lambda_h)\le\theta(Cs/\sqrt q)^{2n},\qquad
 \theta(t)=\sum_{a\in\mathbb Z}e^{-\pi a^2/t^2}.
\]
At $s_{\rm mass}^2\le q/n^\beta$, this gives the public mass budget
$A_{\rm mass}=2$ for all sufficiently large $n$.
Theorem~\ref{uq:mass-search-framework}, with exact algebraic and
quality verification, then proves recovery.
Appendix~\ref{uf:quality-proof} gives the full proof. Its shared
ingredients are the lattice duality in
Lemma~\ref{uf:public-isoduality} and the integer basis completion
procedure in Appendix~\ref{uf:completion}.
Since $D_{\rm Q}\sim\ln\ln n$, this is the probability and
resource bound of Corollary~\ref{ui:falcon-cor}, with the same
algebraic and quality guarantees on the output.
\end{proof}
Every distribution supported on quality-qualified public keys
inherits the same success probability and complexity. This includes
KeyGen outputs satisfying the exact algebraic and quality checks above.

\subsection{Cyclic NTRU: distributional short key recovery}
\label{uq:ntru-distributional}

Use the NTRU-HPS/HRSS families in Section~\ref{ui:families}, with
the raw and accepted laws of Definition~\ref{uf:ntru-proposal-laws}.
The full cyclic lattice may contain additional points in the rational
span of the short pair. Their contribution is measured by
\[
 \Delta_{\rm sat}=[R_N:fR_N+gR_N+kR_N],\qquad k=(\widehat hf-g)/q.
\]
Proposition~\ref{uf:ntru-cyclic-geometry} gives
\[
 \rho_{s_{\rm mass}}(\Lambda_{\widehat h}^{\rm cyc})
 \le\Delta_{\rm sat}\,
 \theta(s_{\rm mass}/\sqrt{G_{\min}})^N
 \theta(s_{\rm mass}\sqrt{G_{\max}}/q)^N.
\]
Here $G_{\min},G_{\max}$ are scalar spectral extrema of
$|f(\zeta)|^2+|g(\zeta)|^2$ over $\zeta^N=1$, not ring elements.
The search needs only a total logarithmic budget $O(N/D_N)$.
Corollaries~\ref{uf:ntru-weak-mass-budget}
and~\ref{uf:ntru-weak-spectral-budget} make this pointwise criterion
explicit. For instance, it is enough that
\[
 G_{\min}\ge(q/N^\beta)\ln D_N,\qquad
 G_{\max}\le qN^\beta/\ln D_N,\qquad
 \ln\Delta_{\rm sat}\le C_{\rm sat} N/D_N.
\]

For the HPS and HRSS samplers, the analysis establishes
\[
 G_{\min}\ge q/N^\alpha,\qquad
 G_{\max}\le C_{\rm sp}q\ln N,\qquad
 \ln\Delta_{\rm sat}\le N/D_N,
 \qquad 2/3<\alpha<\beta<1.
\]
These bounds satisfy the sufficient conditions above for all
sufficiently large $N$ and give a fixed public mass budget.
The proof combines the spectral estimates with a bound on the
logarithm of the saturation index, then conditions on key generation
acceptance. Appendix~\ref{uf:ntru-accepted-bounds} collects this
joint event and its failure probability.
Fix rational $2/3<\alpha<\beta<1$, positive integers $K,K_{\rm sp}$,
and write
\begin{align}
 r_{\rm sp}(N)&=C\left(N^{1-3\alpha/2}+N^{-\alpha/2}
                                  +N^{-K_{\rm sp}}\right),\nonumber\\
 r_{\rm sat}(N)&=64D_N\ln(3N)e^{-\sqrt{\ln N}/4}+12N^{-1/2},
 \qquad a=86/256.
 \label{uq:ntru-input-failure}
\end{align}
Here $D_N$ is the prescribed positive dyadic approximation to
$\ln\ln N$. The effective constant $C$ depends only on the fixed
modulus window and the chosen spectral parameters.

\begin{theorem}[HRSS recovery]
\label{uf:ntru-hrss-distributional}
Let $N$ run along the supplied arithmetic sequence
\eqref{uf:ntru-arithmetic}, and use its HRSS modulus and correctness
conditions. Generate $h$ from the accepted law $\mu_N^{\rm HRSS}$.
For all sufficiently large $N$, a randomized algorithm using
only $(N,q,h)$ and the fixed public parameters returns a pair
accepted by $\mathsf{Valid}_{\mathsf{HRSS}}$ and its short key core
with probability at least
\begin{equation}
 1-\frac{r_{\rm sp}(N)+r_{\rm sat}(N)}{(1-a^{N-1})^2}-N^{-K}
 =1-o(1).
 \label{uf:ntru-hrss-success}
\end{equation}
Its time and space are $2^{O(N/\ln\ln N)}$.
\end{theorem}
\begin{proof}
Proposition~\ref{uf:ntru-accepted-geometry} proves that, except with
the first failure term in~\eqref{uf:ntru-hrss-success}, the generating
pair satisfies
$G_{\min}\ge q/N^\alpha$, $G_{\max}\le C_{\rm sp}q\ln N$,
and $\ln\Delta_{\rm sat}\le N/D_N$.
Because $\alpha<\beta$, these imply the weak spectral conditions for all sufficiently large $N$:
$N^{\beta-\alpha}\ge\ln D_N$ and
$C_{\rm sp}\ln N\ln D_N\le N^\beta$.
Corollary~\ref{uf:ntru-weak-spectral-budget} therefore supplies the
public integer budget
\begin{equation}
 A_{\rm mass}=4^{\lceil3N/D_N\rceil},\qquad
 \ln A_{\rm mass}\le6N/D_N.
 \label{uf:ntru-distributional-budget}
\end{equation}
It has $O(N/D_N)$ bits and is fixed before the key is sampled.

Apply Theorem~\ref{uq:mass-search-framework} with $m=N$, $c_{\rm search}=1$,
$A_0=6$, $B_{\rm Q}^2=C_{\rm len}^2q$, and the actual dyadic width
$s_{\rm mass}^2$ from the finite schedule. The modulus $q=\Theta(N)$
meets its window, and
\[
 \frac{s_{\rm mass}^2}{s_{\rm search}^2}
 \ge(1-2^{-N^2})\frac{N^{1-\beta}}{D_N}\longrightarrow\infty.
\]
Lemma~\ref{uf:ntru-short-check} supplies a fixed rational
$C_{\rm len}$, a valid nonzero witness, and the exact verifier.
The radius satisfies $C_{\rm len}\sqrt q<q/2$ eventually. For each fixed witness
on the good event, fresh algorithmic randomness succeeds with
probability at least $1-N^{-K}$. Averaging over the accepted input
law proves~\eqref{uf:ntru-hrss-success}. The algorithm uses the public budget, radius, and quotient.
\end{proof}

\begin{theorem}[HPS recovery]
\label{uf:ntru-hps-distributional}
Use the supplied arithmetic sequence~\eqref{uf:ntru-arithmetic}
with the HPS weight conditions, and generate $h$ from the ideal
accepted law $\mu_N^{\rm HPS}$ based on uniform sampling of fixed
type vectors. A randomized algorithm with
the same public inputs returns a pair accepted by
$\mathsf{Valid}_{\mathsf{HPS}}$ and its short key core with
probability at least
\begin{equation}
 1-\frac{r_{\rm sp}(N)+r_{\rm sat}(N)}{1-a^{N-1}}-N^{-K}
 =1-o(1),
 \label{uf:ntru-hps-success}
\end{equation}
in time and space $2^{O(N/\ln\ln N)}$.
\end{theorem}
\begin{proof}
The independent ternary and uniform fixed type blocks satisfy
the spectral bound of Proposition~\ref{uf:ntru-sampler-spectrum}
and the saturation bound of Proposition~\ref{uf:ntru-log-saturation}.
Lemma~\ref{uf:ntru-algebraic-acceptance} gives acceptance
probability $1-a^{N-1}$.
Proposition~\ref{uf:ntru-accepted-geometry} therefore bounds
the input failure probability by
$(r_{\rm sp}(N)+r_{\rm sat}(N))/(1-a^{N-1})$.

On the resulting event, the spectral and saturation bounds give
the public budget~\eqref{uf:ntru-distributional-budget},
as in the proof of Theorem~\ref{uf:ntru-hrss-distributional}.
Apply Theorem~\ref{uq:mass-search-framework} with $m=N$,
$c_{\rm search}=1$, $A_0=6$, and $B_{\rm Q}^2=C_{\rm len}^2q$.
The modulus $q=\Theta(N)$ meets its window, and the finite widths satisfy
$s_{\rm mass}^2/s_{\rm search}^2\to\infty$.
Lemma~\ref{uf:ntru-short-check} supplies the HPS witness and its exact
verifier, with $C_{\rm len}\sqrt q<q/2$ eventually.
For every fixed input on this event, the search succeeds with
probability at least $1-N^{-K}$.
Averaging over the accepted HPS law and adding the input failure
proves~\eqref{uf:ntru-hps-success} and the stated resource bounds.
\end{proof}

\begin{corollary}[HPS recovery with sorting labels]
\label{uf:ntru-hps-tags-distributional}
For the same HPS parameter sequence, generate $h$ using the sorting
sampler with label length $b_N$ and accepted law
$\mu_{N,b_N}^{\rm HPS}$. The algorithm succeeds with probability
at least
\begin{equation}
 1-\frac{r_{\rm sp}(N)+r_{\rm sat}(N)+\varepsilon_{\rm FT}(N)}
              {1-a^{N-1}}-N^{-K}.
 \label{uf:ntru-hps-tag-success}
\end{equation}
For $b_N=\lceil(K_{\rm tag}+3)\log_2N\rceil$, with any fixed
positive integer $K_{\rm tag}$, this is $1-o(1)$ and the time and
space bound remains $2^{O(N/\ln\ln N)}$.
\end{corollary}
\begin{proof}
Lemma~\ref{uf:ntru-fixed-type-tv} pays $\varepsilon_{\rm FT}$ once
for the entire raw good event. Algebraic acceptance still depends
only on the nonzero $f$ block and has probability $1-a^{N-1}$.
Thus Proposition~\ref{uf:ntru-accepted-geometry} gives exactly the
input failure term displayed above. The same public search budget
and verifier give the remaining $N^{-K}$ term. The chosen labels
have $O(\log N)$ bits and $\varepsilon_{\rm FT}=o(1)$.

For all three accepted laws, $r_{\rm sp}+r_{\rm sat}=o(1)$ and
$D_N\sim\ln\ln N$. The two theorems and this corollary therefore
prove Corollary~\ref{ui:ntru-cor} for the NTRU-HPS/HRSS family.
\end{proof}

\section{Conclusion and Further Discussion}
\label{uz:conclusion}

We developed two frameworks for recovery in growing families associated
with the seven lattice candidates in NIST's third round. Gaussian
estimates play a central role in both frameworks. For noisy or rounded
linear relations, they reveal the energy difference created by a
coordinate guess. The joint energy of the secret and error determines
the sampling depth and the leading recovery exponent. For quotient
relations, estimates of the mass in affine slices guide the search for
a short pair satisfying the public checks. A bound on the total
Gaussian mass controls the probability of reaching such a pair.

The common approach separates the recovery procedure from the geometry
of the public operator. The applications supply this geometry through
uniform images, ideals, additive orders, key quality, or spectral
estimates. This separation gives a way to study further public
relations by identifying the energy and mass bounds they require.
Appendix~\ref{ux:section} also examines how dimension and energy together
determine the cost of composing a reduction with the recovery algorithm.
We next describe an application to homomorphic encryption and several
directions for further work.

\subsection{Further Applications to Homomorphic Encryption}
\label{uz:fhe}

The framework for linear relations also applies to RLWE public keys in
homomorphic encryption. BGV, BFV, and CKKS use closely related public
relations~\cite{BGV12,FV12,CKKS17}.
Let $N$ be a power of two, and write
\[
 R=\mathbb Z[X]/(X^N+1),\qquad R_Q=R/QR.
\]
Here $Q$ is the modulus of the public key relation. It corresponds to
$q_L$ in the CKKS key generation described in~\cite[Section~3.4]{CKKS17}.
In a residue number system representation, write $Q=\prod_j q_j$ with
distinct primes $q_j$. These primes are individual components, whereas
the CKKS notation $q_\ell$ denotes the full modulus at level $\ell$.
More generally, the argument below uses any known prime divisor
$q_i\mid Q$. Write $R_{q_i}=R/q_iR$ for the selected component.

Table~\ref{uz:fhe-table} gives the three public relations and their
normalization at the selected prime. In the table, $a,b\in R_Q$ retain
their public key meanings, and $a$ is uniform in $R_Q$. The parameter
$t$ is the plaintext modulus of BGV. We use $\tilde a$ for the normalized
multiplier and $y$ for the normalized public observation.

\begin{table}[!htbp]
\centering
\papertablesetup
\caption{Public key relations in homomorphic encryption and their
normalization at a selected prime.}
\label{uz:fhe-table}
\begin{tabularx}{\linewidth}{@{}P{.11\linewidth}P{.25\linewidth}Y@{}}
\toprule
Scheme & Public relation modulo $Q$ & Normalization modulo $q_i$\\
\midrule
BGV & $b=as+te$
& Set $\tilde a=t^{-1}(a\bmod q_i)$ and
$y=t^{-1}(b\bmod q_i)$ when $t$ is invertible modulo $q_i$.\\
\addlinespace[4pt]
BFV & $b=-(as+e)$
& Set $\tilde a=a\bmod q_i$ and $y=-b\bmod q_i$.\\
\addlinespace[4pt]
CKKS & $b=-as+e$
& Set $\tilde a=-a\bmod q_i$ and $y=b\bmod q_i$.\\
\bottomrule
\end{tabularx}
\end{table}
\FloatBarrier

After reduction and normalization, each relation takes the form
\[
 y=\tilde a s+e\pmod{q_i}.
\]
These transformations make $\tilde a$ uniform in $R_{q_i}$ and retain
the same integer secret and error. Reduction to a prime component
introduces no additional rounding error. This gives a module instance
of rank one to which the recovery theorem can be applied, with ring
degree and secret coefficient dimension both equal to $N$.
One ring equation supplies all $N$ coefficient observations.

Let $N$ grow through powers of two, take $q_i\equiv1\pmod{2N}$, and
assume $\ln q_i\le(\ln N)^c$ for a fixed $c$.
Choose a public rational energy bound $V\ge1$ of polynomial bit length
before the public draw, with
\[
 \Pr[\|s\|_2^2+\|e\|_2^2>NV]\le p_E,\qquad
 D_i=\ln\frac{\ln(q_i/\sqrt N)}{1+\tfrac12\ln V}.
\]
When $D_i\to\infty$, Theorem~\ref{um:recovery} gives recovery of the
original integer secret with probability $1-p_E-o(1)$ in expected
time and space
\[
 2^{(1/2+o(1))N/D_i}.
\]
For $q_i=N^{\kappa+o(1)}$ with fixed $\kappa\ge1$ and polylogarithmic
$V$, this becomes $2^{(1/2+o(1))N/\ln\ln N}$.
The selected prime determines the recovery depth. Reading the full
public input adds polynomial work when the number of components and
their bit lengths are polynomial in $N$.

In the ring LWE version of BGV's basic key generation, $t=2$
\cite[Section~3.1]{BGV12}. The division by $t$ in the table is performed
modulo the selected prime, so the energy bound uses $e$.
When this division is unavailable, set $\tilde a=a\bmod q_i$ and
$y=b\bmod q_i$. The resulting relation has effective error $te$, and
the energy certificate uses $\|s\|_2^2+\|te\|_2^2$.
For CKKS, the energy comes from the secret and the public key error in
KeyGen~\cite[Section~3.4]{CKKS17}. The plaintext encoding scale and
later rescaling operations do not enter this public relation.
Recovery identifies the secret of the selected public key, which can
then be used with the corresponding decryption algorithm.

\subsection{Further Directions}
\label{uz:directions}

\paragraph{Implementation and concrete costs.}
Our algorithms use conservative parameter choices to ensure provable
recovery. A natural next step is to implement these algorithms with suitable parameters and
measure their performance on small instances.
For linear recovery, the experiments can measure the score gap between
adjacent candidates. For quotient recovery, they can compare estimated
slice masses with direct calculations and record the number of trials
needed to find a valid pair.
These measurements would support a cost model for tuning block sizes,
Gaussian widths, list lengths, and search repetitions in practice.

\paragraph{Heuristic improvements.}
Gaussian peak estimation could accelerate linear recovery by locating
each secret coordinate from a few score evaluations, followed by a
local search around the predicted value.

Extreme pruning could reduce the search in quotient recovery through
tighter norm budgets at intermediate depths. Estimated Gaussian
masses could then guide which of the remaining partial assignments
to retain. Varying the pruning thresholds and randomizing branch
selection across restarts could explore alternative paths, with
every completed pair checked by the full public verifier.

\paragraph{Other public distributions and rings.}
Further geometric estimates could extend the frameworks to public
operators generated by seed expansion, other coefficient distributions,
and reducible polynomial rings beyond the cyclotomic cases treated here.
Appendix~\ref{uk:section} gives conditions for transferring a geometric
guarantee between public distributions. Another route is to estimate the
Gaussian mass directly under the new distribution. For linear relations,
the main task is to control the required coefficient prefixes. For
quotient relations, it is to obtain a useful mass bound and an efficient
test for the desired output.

\paragraph{Gram matrix public keys.}
HAWK~\cite{DPPvW22} provides a contrast to the seven NIST lattice
families considered here. Its public key has the form
$\mathbf Q=\mathbf B^*\mathbf B$, where
$\mathbf B\in\operatorname{SL}_2(R)$ is a secret basis over a cyclotomic
ring $R$ whose degree is a power of two, and $*$ denotes conjugate
transpose. Recovery seeks a basis $\mathbf B'$ with the same Gram
matrix. The linear framework requires a noisy or rounded linear
relation with a public operator, while HAWK publishes the quadratic
Gram relation above. The quotient framework requires a formulation
as a search for a short pair in a public modular lattice, together
with a suitable Gaussian mass bound.

HAWK lies outside the scope of our two frameworks, but not beyond the
reach of lattice cryptanalysis. Straznickas and Weis reduce HAWK key
recovery in degree $n$ to polynomially many calls to exact SVP in
dimension at most $n/2+1$~\cite[Theorem~6.1]{SW26}. Following this
attack, HAWK was withdrawn from NIST's additional signature
standardization process on 29 July 2026. It remains interesting to determine
whether Gaussian mass estimates can also guide recovery from this
quadratic public relation and what geometric information such a search
would require.

Together, these directions connect the geometric framework with the
concrete behavior of recovery algorithms and with further forms of
public information.

\clearpage
\phantomsection

\clearpage
\appendix
\phantomsection
\addcontentsline{toc}{part}{Appendices}
\section{Finite Implementation}
\label{ue:finite}

We give finite sampling tools for both recovery frameworks, followed
by the schedule and numerical precision for linear recovery.
Sections~\ref{uf:persistent-lists} and~\ref{uf:affine-recovery} specify
the calls and budgets for quotient recovery.

\subsection{Common sampling tools}
\label{ue:common-sampling}

\paragraph{Finite random choices.}
To choose uniformly from $\{0,\ldots,S-1\}$, draw
$\lceil\log_2S\rceil$ unbiased bits and reject values at least $S$.
Each trial succeeds with probability at least $1/2$. For at most $J\ge1$
calls and $0<\nu_{\rm bits}<1$, a cap of
$\lceil\log_2(J/\nu_{\rm bits})\rceil$ trials per
call gives total abort probability at most $\nu_{\rm bits}$.
Successful calls are exactly uniform. For $S=1$, return zero.
For $K$ nonnegative rational weights with positive sum, partition
$[0,1)$ by their normalized cumulative sums and locate $U/2^a$
for a uniform integer $0\le U<2^a$. This uses $a$ random bits and
has total variation error at most $K2^{-a}$. A zero sum is a failure.

\begin{lemma}[Scalar Gaussian masses and samples]
\label{uimpl:scalar-coset}
For rational $\lambda,a,c,s^2,\zeta$ with $\lambda,s^2>0$ and
$0<\zeta<1/2$, put $\sigma=s/\lambda$ and
$Z=\rho_s(a+\lambda\mathbb Z-c)$. One can compute a positive dyadic
interval $[Z_-,Z_+]$ containing $Z$ with
$Z_+-Z_-\le\zeta Z_-$, and sample within total variation distance
$\zeta$ of $D_{a+\lambda\mathbb Z,s,c}$, in
\[
 \operatorname{poly}\bigl(L_{\rm enc},1+\sigma+\sigma^{-1},
                             \ln(1/\zeta)\bigr)
\]
bit operations. Here $L_{\rm enc}$ bounds the input encoding lengths.
\end{lemma}
\begin{proof}
Write $(c-a)/\lambda=k+\delta$, with $k\in\mathbb Z$ and
$|\delta|\le1/2$. Then $Z=\sum_j e^{-\pi(j-\delta)^2/\sigma^2}$.
Set $A=\lceil\sigma^{-2}\rceil$ and $\ell=2^{-2A}$.
The nearest term and a Gaussian integral bound give, for $T\ge1$,
\[
 Z\ge e^{-\pi/(4\sigma^2)}\ge\ell,\qquad
 \sum_{|j|>T}e^{-\pi(j-\delta)^2/\sigma^2}
 \le(1+4\sigma)e^{-\pi T^2/(8\sigma^2)}.
\]
For the tail bound, use $|j-\delta|\ge|j|/2$ and split the exponent.
Choose an integer
$T=O(1+\sigma\sqrt{A+\ln((1+\sigma)/\zeta)})$
so the tail is at most $\zeta\ell/100$. Evaluate the retained terms
in nonnegative dyadic intervals with a common denominator and total
width at most $\zeta\ell/100$. Sum lower endpoints for $Z_-$ and
upper endpoints plus a tail bound of at most $2\zeta\ell/100$ for
$Z_+$. Then $Z_-\ge(1-2\zeta/100)\ell>0$ and the required relative
width follows.

Use the lower endpoints as sampling weights and allocate
$\zeta/100$ to their finite random choice. Truncation and weight
approximation contribute at most $3\zeta/100$, giving total error
less than $\zeta$. Return $a+\lambda(k+j)$, whose distance from $c$
is at most $\lambda(T+1/2)$. All bounds use outward rounding.
The precision is $O(A+\ln(T+1)+\ln(1/\zeta))$ bits. Certified
exponential evaluation and the bound on $T$ give the stated cost.
\end{proof}

\begin{lemma}[Auxiliary Gaussian masses and samples]
\label{uimpl:product-cosets}
\label{ue:direct-sampling}
Let $q\ge2$ be an integer, $P$ a prime, and
$\Gamma=q\mathbb Z^b+(q/P)\mathbf z\mathbb Z$ for a nonzero
$\mathbf z\in\mathbb F_P^b$. For rational $\mathbf c,s^2,\xi$ with
$s^2>0$ and $0<\xi<1/2$, one can enclose
$\rho_s(\Gamma-\mathbf c)$ in a positive dyadic interval of relative
width at most $\xi$, and sample within total variation distance $\xi$
of $D_{\Gamma,s,\mathbf c}$, in
\[
 P\operatorname{poly}\bigl(b,L_{\rm enc},1+s/q+q/s,
                              \ln P,\ln(1/\xi)\bigr)
\]
bit operations, where $L_{\rm enc}$ bounds the input encoding lengths.
\end{lemma}
\begin{proof}
The nonzero vector $\mathbf z$ has order $P$, so $\Gamma$ is the
disjoint union of the cosets $qk\mathbf z/P+q\mathbb Z^b$.
Their masses are
\begin{equation}\label{uf:explicit-shifted-dgs}
 W_k(\mathbf c)=\prod_{j=1}^b
      \rho_s(q\mathbb Z+qkz_j/P-c_j),\qquad 0\le k<P.
\end{equation}
Evaluate each factor by Lemma~\ref{uimpl:scalar-coset} at accuracy
$\zeta=\xi/(100Pb)$. Products of the interval endpoints enclose each
$W_k$ with relative width at most $(1+\zeta)^b-1\le2b\zeta$.
Their sums give the mass interval. Select a coset using the lower
weights, then sample its coordinates independently at accuracy
$\zeta$. With $\xi/10$ for the finite coset choice, the total error
is at most $2b\zeta+\xi/10+b\zeta<\xi$.
There are $Pb$ scalar mass calculations and $b$ samples. A common
binary precision keeps products and cumulative weights within the
stated cost.
\end{proof}

Both sampling bounds hold for every rational center and include
truncation and finite random choices. Lemma~\ref{up:adaptive-error}
therefore bounds the error of adaptive calls with fresh randomness
by the sum of their assigned accuracies.

\subsection{Finite implementation for linear recovery}
\label{uimpl:linear}

\paragraph{A deterministic finite schedule.}
\label{uimpl:linear-schedule}
All calculations precede the public draw and use certified intervals.
A failed check reports that the parameters are outside the proved range.
Choose an integer $f_{\rm sq}\in[K(R^2/n)D,K(R^2/n)D+2]$ by
rounding the upper endpoint of an interval of width less than one.
Check $f_{\rm sq}\ge4$, set $f^2=f_{\rm sq}$, and choose dyadic
$\bar\tau\in[\tau/4,\tau]$, where $\tau=e^{-\pi R^2/f_{\rm sq}}$.
Put
\[
 \bar g=\bar\tau/f_{\rm sq},\qquad
 \varepsilon=\frac{\bar\tau^2}{256f_{\rm sq}q^4},\qquad
 \varepsilon_0=\varepsilon/5.
\]
Since $f_{\rm sq}^{-1}\le1-e^{-\pi/f_{\rm sq}}\le4f_{\rm sq}^{-1}$,
we have $g/16\le\bar g\le g$.
Set $L=\ln(4n/\varepsilon)$ and $L_0=L+\ln5$, and check
$L>0$ and $q>\sqrt L$ before evaluating $W$.
Recompute the schedule with these tolerances, to which the prefix
condition~\eqref{ue:geometry} refers. For $r$, floor a certified lower
bound within $1/4$ of~\eqref{ue:schedule}. This can decrease $r$ by
one and gives $512L_0\le s_0^2<2048L_0$.
Take a lower bound $\widetilde W$ within $1/4$ of $W$, check
$\widetilde W>C_0+1$ and $1\le r\le n$, and set
\[
 N=2^{\lceil n/[2(\widetilde W-C_0-1)]\rceil},\qquad U_0=N/2.
\]
Check $N\ge4$. Compute rational bounds and capacities
\[
 \begin{gathered}
 0\le\widetilde y_i-y_i\le n^{-3},\qquad
 0\le\ln U_0-\widetilde\ell\le n^{-3},\qquad
 \widetilde B_i=\left\lfloor\frac{\widetilde\ell}{\widetilde y_i}\right\rfloor,\\
 \widetilde b_*=\min\left\{
   \left\lfloor\tfrac12\min_i\widetilde B_i\right\rfloor,
   \left\lfloor n/r\right\rfloor\right\}.
 \end{gathered}
\]
Check $\widetilde b_*\ge1$ and $\sum_i\widetilde B_i\ge n$.
Assign $\widetilde b_*$ coordinates to each level and fill the rest
without exceeding $\widetilde B_i$. The rounding directions preserve
the width inequalities, and $r\widetilde b_*\le n$ and
$\widetilde B_i\ge2\widetilde b_*$ hold by construction.
The continuous approximation loses at most
$O(rn^{-3}\ln U_0)=o(1)$ total capacity and the floors lose at most
$r$. The margin $C_0\ge C_{\rm cap}+10$ ensures that all checks pass
for sufficiently large members of the family. Lemma~\ref{ue:allocation}
gives $\widetilde b_*=\Omega(\ln N/\ln q)$.
All squared widths $s_i^2=q^2 2^{i-r}/f_{\rm sq}$ are rational.

\paragraph{Coordinate size and score precision.}
Reduce centers modulo $q$ and obtain a rational basis of $\Gamma_i$
from the Hermite normal form of $[P\mathbf I_{b_i}\mid\mathbf z_i]$,
scaled by $q/P$. Abort if an initial or centered shifted sample has
norm exceeding $R_{\rm cap}=2^{n^2}$. The inclusions
$q\mathbb Z^b\subseteq\Gamma_i\subseteq(q/P)\mathbb Z^b$ give the
uniform ideal tail bound
\[
 e^{-\pi R_{\rm cap}^2/(2s^2)}
 (1+\sqrt2sP/q)^b e^{\pi bq^2/(4s^2)}.
\]
For the denominator, use a point of $q\mathbb Z^b$ within
$q\sqrt b/2$ of the center. For the numerator, split the exponent
and sum over the containing fine product lattice. Poisson summation
and $\theta(v)\le1+v$ bound a shifted scalar mass at width to spacing
ratio $v$ by $1+v$.
The initial product Gaussian obeys the same type of bound.
In the modulus window, all cutoff events together cost at most
$e^{-n^2}$. Applying the same cutoff to both processes preserves
their total variation bound. Retained coordinates have polynomial
bit length through $r=O(\ln q)$ levels.

Compute phases modulo $q$ exactly and each cosine to error
$\bar g/64$ with a common binary denominator.
Lemma~\ref{up:concentration}, with $u=g/16$, gives failure at most
$2H e^{-Ng^2/512}$ for $H=n(2B+1)q^h$.
The errors from concentration, Gaussian approximation, conditional
means, and cosine evaluation have total at most
\[
 g/16+\bar g/1280+\bar g/128+\bar g/64<g/8.
\]
Thus two score errors leave a strict margin below the adjacent gap
$g$, with precision $O(\ln(1/\bar g))=O(n/D)$ bits.

\paragraph{Auxiliary vectors and initialization.}
\label{uimpl:linear-initialization}
Choose a nonzero vector in $\mathbb F_P^{b_i}$ by expanding $U+1$
in base $P$, with $U$ uniform in $\{0,\ldots,P^{b_i}-2\}$.
The rule of Section~\ref{ue:common-sampling}, with independent random
bits and total budget $\nu_z=n^{-10}$ over $r$ levels, gives
independent uniform vectors conditional on no abort.
For the $M_{\rm init}=3^rhN$ initial coordinates, use
Lemma~\ref{uimpl:scalar-coset} with $\lambda=1$, $a=c=0$, and
$\xi_{\rm init}=2^{-n^2}/(10M_{\rm init})$.
The total error is at most $\nu_{\rm init}=2^{-n^2}$.
Since $s_0=O(\sqrt{n/D})$ and
$\ln(1/\xi_{\rm init})=O(n^2)$, initialization costs
$M_{\rm init}\operatorname{poly}(n)=N^{1+o(1)}$.

\paragraph{The two implementations.}
\label{uimpl:linear-samplers}
Theorem~\ref{ue:list-engine} verifies the ADS15 width condition for
every center and gives the error~\eqref{ue:tv}. The polynomial bit
lengths above absorb the input factors in
Lemma~\ref{up:shifted-sampler}, preserving~\eqref{ue:refined-cost}.
For its computable uniform expected cost bound $\bar T_n$, stop after
$n^{K_t}\bar T_n$ steps. Markov's inequality adds at most $n^{-K_t}$
failure and preserves the exponent.

For direct sampling, apply Lemma~\ref{uimpl:product-cosets} to each
$\Gamma_i$ at width $s_{i-1}$. There are
$J=(3^{r+1}-3)N/2$ shifted calls. Give these calls and each initial
scalar draw accuracy $\xi=\nu/[10(J+M_{\rm init})]$, with
$\nu=n^{-10}$. Lemma~\ref{up:adaptive-error} gives total error at
most $\nu$. Since $1\le s_{i-1}\le q$, each shifted call costs
$P\operatorname{poly}(n,q,\ln P,\ln(1/\xi))$.
The numerical factor is absorbed by
$\ln q=\operatorname{polylog}(n)=o(n/D)$, and enumeration contributes
the factor $P$, giving~\eqref{ue:finite-cost}.
The complete implementation budgets in~\eqref{ue:failure} are
\[
 \delta_{\rm impl}\le
 \begin{cases}
  \nu_z+\nu_{\rm DGS}+e^{-n^2}+n^{-K_t},&\text{ADS15 sampler},\\
  \nu_z+\nu+e^{-n^2},&\text{coset enumeration}.
 \end{cases}
\]

\section{Proofs for Quotient Recovery}
\label{uf:section}

We first construct dual Gaussian lists and use them to search affine
slices of a public lattice. We then bound its Gaussian mass and verify
the output for the Falcon parameter family and cyclic NTRU. These are the three steps
of the recovery method in Section~\ref{uq:overview}.

In the generic construction and search proof below, $n$ denotes the
block length $m$ of Section~\ref{uq:overview}; the full pair lattice has
dimension $2n$. We use $\mathbf H_h$ and $\mathcal C(h)$ as local
notations for that section's public matrix $\mathbf H$ and candidate
table $\mathcal C(\mathbf H)$, even for an arbitrary matrix.
In the family applications, the matrix is coefficient multiplication
by $h$ for Falcon or by $\widehat h$ for cyclic NTRU.

\subsection{Persistent Gaussian lists}
\label{uf:persistent-lists}

Let $\mathbf H_h$ be an integral representative of any public $n\times n$
matrix modulo $q$, and set
$\Lambda_h=\ker_q[\mathbf H_h\mid-\mathbf I_n]$.
One random rank-one extension will serve every coordinate prefix.
This shared extension lets a single smoothing event control all lists.
We remove it at the end by taking differences in predetermined pairs.
The construction applies to both multiplication matrices used later.

For $S\subseteq[n]$, let $\mathbf H_{h,S}$ contain the columns of $\mathbf H_h$
indexed by $S$, and put
\[
 L_S=\{(\mathbf u_S,\mathbf v)\in\mathbb Z^{|S|+n}:
             \mathbf H_{h,S}\mathbf u_S-\mathbf v=\mathbf0\pmod q\}.
\]
Let
$\mathbf U_S:\mathbb R^{|S|}\oplus\mathbb R^n\to
\mathbb R^n\oplus\mathbb R^{|S|}$ be the orthogonal block swap
$\mathbf U_S(\mathbf u_S,\mathbf v)=(\mathbf v,\mathbf u_S)$, and define
\begin{equation}
 K_S:=\mathbf U_S(qL_S^*)
   =\ker_q[\mathbf H_{h,S}^T\mid\mathbf I_{|S|}]
   \subseteq\mathbb Z^{n+|S|}.
 \label{uf:KS}
\end{equation}
Dualizing gives
\begin{equation}
 qK_S^*=\mathbf U_SL_S,
 \qquad \rho_{q/u}(qK_S^*)\le\rho_{q/u}(\Lambda_h).
 \label{uf:dual-interface}
\end{equation}
The inequality uses the zero-padded inclusion $L_S\subseteq\Lambda_h$.

To remove an extension $\Lambda'\supseteq\Lambda$ of index
$[\Lambda':\Lambda]=P$, use the following fixed-pair procedure.
Given a list $(\mathbf X_1,\ldots,\mathbf X_{2M})$, pair entries
$(\mathbf X_{2j-1},\mathbf X_{2j})$. A pair succeeds when its difference
lies in $\Lambda$. Scan pairs in an order fixed independently of their
values and return the first $N$ successful differences. If fewer than
$N$ pairs succeed, declare failure.

For the analysis, put $\mu=D_{\Lambda',s}$ and let $K$ be the law of
$\mathbf X-\mathbf X'$ conditioned on membership in $\Lambda$, for
independent $\mathbf X,\mathbf X'\leftarrow\mu$.
Complete a deficient output with independent $K$-samples to define a
reference list of length $N$. This completion describes the reference
law on the failure event.

\begin{lemma}[Quotient removal]
\label{uf:fixed-pair-removal}
Suppose the input list is strongly conditionally $\delta$-similar to
$\mu^{2M}$ and
\[
 \eta_\varepsilon(\Lambda')\le s/\sqrt2,
 \qquad 0<\varepsilon\le\tfrac12.
\]
Then the completed reference list is strongly conditionally
\begin{equation}
 4\delta+\log\frac{1+\varepsilon}{1-\varepsilon}
 \label{uf:fixed-pair-error}
\end{equation}
similar to $D_{\Lambda,\sqrt2s}^{N}$.  Moreover, for $0<\zeta<1$,
\begin{equation}
 M\ge8e^{2\delta}P\bigl(N+\ln(1/\zeta)\bigr)
 \label{uf:fixed-pair-count}
\end{equation}
implies that the finite procedure and completed reference disagree with
probability at most $\zeta$.
The similarity conclusion also holds after conditioning on any
positive-probability event determined by the vector of pair-success
indicators, in particular on obtaining at least $N$ successes.
\end{lemma}
\begin{proof}
Write $p_{\mathbf c}=\Pr[\mathbf X\in\Lambda+\mathbf c]$ for each quotient
class represented by $\mathbf c$. The ideal success probability is
$C=\sum_{\mathbf c}p_{\mathbf c}^2\ge1/P$.
For $\mathbf z\in\Lambda$, completing the square gives
\[
 K(\mathbf z)\ \propto\
 \rho_{\sqrt2s}(\mathbf z)\rho_{s/\sqrt2}(\Lambda'-\mathbf z/2).
\]
Poisson summation and the smoothing hypothesis place the second factor,
uniformly in $\mathbf z$, in one multiplicative interval of ratio
$(1+\varepsilon)/(1-\varepsilon)$.  Thus $K$ has the stated pointwise
distance from $D_{\Lambda,\sqrt2s}$.

Given all other original entries, each fixed pair is
$2\delta$-similar to $\mu^{\otimes2}$.  Conditioning that pair on
success pays the same normalization factor once more, so its successful
difference is $4\delta$-similar to $K$. To pass to the output list,
condition first on the complete vector of success indicators and on
the original entries of all other pairs. The selected pair positions
are now fixed, and conditioning the current pair on success gives
exactly the bound just proved. Averaging over these conditions gives
the same bound conditional on the other output differences. Padding
entries are independent $K$-samples. An additional event determined
by the success indicators is already fixed under this conditioning,
so averaging over its allowed indicator vectors preserves the same
bound. This proves~\eqref{uf:fixed-pair-error} and its conditional form.
Finally, each
pair has conditional success probability at least $e^{-2\delta}/P$.
The conditional Chernoff bound and~\eqref{uf:fixed-pair-count} bound the
shortage event by $\zeta$.
\end{proof}

For the list construction, fix the following finite widths and
accuracy parameters.
Fix computable rational constants $0<c<1$, $1/2<\beta<1$,
and positive $a,K,A_q,A_0,A_N,A_{\mathcal S},c_{\rm search}$.
Let $n,q$ satisfy
\[
 2n+1\le q\le n^{A_q},\qquad
 \ln(q/n)\le(\ln n)^{1-c}.
\]
Choose a positive dyadic rational $D$ with
$0\le\ln\ln n-D\le2^{-n^2}$ and a positive dyadic rational
$s_{\rm mass}^2$ with
\[
 (1-2^{-n^2})q/n^\beta\le s_{\rm mass}^2\le q/n^\beta,
\]
and put
\[
 s_{\rm search}^2=c_{\rm search} qD/n,\qquad t=q/s_{\rm search}.
\]
\begin{samepage}
\begin{theorem}[Persistent lists]
\label{uf:persistent-theorem}
Use the parameters above. For a public matrix $\mathbf H_h$, suppose a
computable rational upper bound $A_{\rm mass}$ is supplied with
\begin{equation}
 \rho_{s_{\rm mass}}(\Lambda_h)\le A_{\rm mass},
 \qquad \ln A_{\rm mass}\le A_0n/D,
 \qquad s_{\rm mass}^2\ge8s_{\rm search}^2.
 \label{uf:list-geometry}
\end{equation}
Let $\mathcal S$ contain at most $n^{A_{\mathcal S}}$ coordinate sets,
fixed before the sampler randomness, and let $N_{\rm out}\ge1$ be an
integer with $\ln N_{\rm out}\le A_Nn/D$.
For all sufficiently large $n$, a finite randomized procedure produces
an $N_{\rm out}$-entry list in each $L_S^*$ with the following guarantees.
\begin{itemize}
\item \emph{List law.} Conditional on successful reference preprocessing,
each reference list is strongly conditionally
\[
 \delta_{\rm out}\le n^{-K}e^{-an/D}
\]
similar to $D_{L_S^*,1/s_{\rm search}}^{N_{\rm out}}$. Different lists may be
dependent. Reference lists use the analytical completion in
Lemma~\ref{uf:fixed-pair-removal}.
\item \emph{Implementation.} Preprocessing failure and total variation
error for the complete construction are each at most $n^{-K}$.
Total time and space over all prescribed sets are $\exp(O(n/D))$.
\end{itemize}
\end{theorem}
\end{samepage}
\begin{proof}
\emph{One mask for all prefixes.}
Write $u_0=q/s_{\rm mass}$. Choose an auxiliary prime $P_0>q$, draw
independently
\[
 \mathbf a_0\leftarrow\mathbb F_{P_0}^n\setminus\{\mathbf0\},
 \qquad \mathbf b_0\leftarrow\mathbb F_{P_0}^n,
\]
and, using fixed integer representatives, set
\[
 \boldsymbol\gamma_S=\frac q{P_0}(\mathbf a_0,(\mathbf b_0)_S),
 \qquad \widetilde K_S=K_S+\mathbb Z\boldsymbol\gamma_S.
\]
The quotient $\widetilde K_S/K_S$ has order exactly $P_0$.
If $S_{i-1}\subseteq S_i$, deletion of the new coordinates maps
$\widetilde K_{S_i}$ onto $\widetilde K_{S_{i-1}}$, with kernel
$q\mathbb Z^{|S_i\setminus S_{i-1}|}$. Indeed, surjectivity follows by
extending the integral $K_{S_{i-1}}$ part while retaining the same
multiple of $\boldsymbol\gamma$. If an element lies in the kernel, its first $n$
coordinates are an integer vector plus $kq\mathbf a_0/P_0$. Since $\mathbf a_0\ne\mathbf0$,
this forces $k=0\pmod{P_0}$, after which the systematic equation gives
the stated kernel.

Fix $j$ with $(\mathbf a_0)_j\ne0\pmod{P_0}$. For
$\mathbf x=(\mathbf x_0,\mathbf x_S)\in\widetilde K_S$, its common quotient label is
\[
 \ell_S(\mathbf x)=(q(\mathbf a_0)_j)^{-1}P_0(\mathbf x_0)_j\pmod{P_0}.
\]
For the full set, direct duality gives
\[
 \widetilde K_{[n]}^*
 =\{\mathbf y\in K_{[n]}^*:
       \langle q\mathbf y,(\mathbf a_0,\mathbf b_0)\rangle=0\pmod{P_0}\}.
\]
If $q\mathbf y\not\equiv\mathbf0\pmod{P_0}$, its survival probability is at most
$1/P_0$: use uniform $\mathbf b_0$ if the last block is nonzero and otherwise
uniform nonzero $\mathbf a_0$, for which the probability is
$(P_0^{n-1}-1)/(P_0^n-1)<1/P_0$. Since $qK_{[n]}^*$ is orthogonally
equal to $\Lambda_h$,
\begin{equation}
 \mathbb E\!\left[\rho_{1/u_0}
     (\widetilde K_{[n]}^*\setminus\{0\})\right]
 \le\frac{\rho_{s_{\rm mass}}(\Lambda_h)-1}{P_0}
       +\theta(s_{\rm mass}/P_0)^{2n}-1.
 \label{uf:mask-expectation}
\end{equation}

Choose a positive dyadic $\delta$ and set $\varepsilon_\star$ by
\begin{equation}
 \begin{gathered}
 \tfrac12 n^{-K-10}e^{-(a+10)n/D}
 \le\delta\le n^{-K-10}e^{-(a+10)n/D},\\
 \varepsilon_\star=\frac{\delta}{10^4 4^r},
 \end{gathered}
 \label{uf:fine-parameters}
\end{equation}
where $r$ is fixed below. Choose $P_0=\exp(O(n/D))$ large enough that
each term on the right of~\eqref{uf:mask-expectation} is at most
$\varepsilon_\star n^{-K-10}/2$.  This is possible because
$\ln A_{\rm mass}\le A_0n/D$. A prime in the required constant-factor
interval exists and may be chosen above the polynomially bounded $q$.
Markov's inequality gives
\begin{equation}
 \eta_{\varepsilon_\star}(\widetilde K_{[n]})\le u_0
 \label{uf:global-smoothing}
\end{equation}
except with probability $n^{-K-10}$. Projection of the full persistent
lattice onto every $\widetilde K_S$ is surjective, so zero-padding embeds
$\widetilde K_S^*$ into $\widetilde K_{[n]}^*$. The single event
\eqref{uf:global-smoothing} therefore implies
\begin{equation}
 \eta_{\varepsilon_\star}(\widetilde K_S)\le u_0
 \quad\text{for every }S\subseteq[n],
 \label{uf:all-smoothing}
\end{equation}
without a union bound over coordinate sets.

\smallskip\noindent
\emph{Widths and block capacity.}
Choose the width schedule
\[
 r=\left\lfloor\log_2\frac{t^2}{4u_0^2}\right\rfloor,
 \quad s_0=\frac{t}{\sqrt2\,2^{r/2}},
 \quad s_i=2^{i/2}s_0,
 \quad u_i=s_{i-1}/\sqrt2.
\]
The final condition in~\eqref{uf:list-geometry} makes $r\ge1$.
Then $s_0\in[\sqrt2u_0,2u_0)$, $s_r=t/\sqrt2$, and
\[
 r=(1-\beta)\log_2n-\log_2D+O(1),
 \qquad u_i\ge u_0.
\]

For a compact capacity calculation put
\[
 x_i=\ln(q/u_i),\qquad x_i\le X_i\le x_i+2,
 \qquad H_0=\sum_{i=1}^r\frac1{X_i+4},
 \qquad \ell=\ln(4/\varepsilon_\star).
\]
Put
\[
 L_P=\max\{\ell,4(n+r)/H_0,2\ln n+20\},
 \qquad X_P=\lceil e^{L_P}\rceil.
\]
By Bertrand's postulate, choose a prime $P$ with
$X_P<P<2X_P$. Then $\ln P\ge L_P$ and
$\ln P=L_P+O(1)$. Set
\[
 B_i=\left\lfloor\frac{\ln P}{2(X_i+4)}\right\rfloor.
\]
Then $\sum_iB_i\ge n$. Moreover, for every $1\le d\le B_i$,
writing $y=(\ln P)/d$ gives $y\ge2(x_i+4)$ and
$\ell/d\le y$, whence
\begin{equation}
 16qP^{-1/d}\sqrt{\max\{1,\ell/d\}}\le u_i,
 \qquad P^{1/d}>16,
 \qquad \sqrt d\,P^{1/d}\le P.
 \label{uf:capacity-condition}
\end{equation}
The first inequality follows by taking logarithms. The bounds $y\ge8$ and
$d\le n$ give the other two. The arithmetic progression $x_i$ has endpoints
\[
 x_1=\tfrac12(\ln(q/n)+(1-\beta)\ln n)+O(1),\qquad
 x_r=\tfrac12(\ln(q/n)+\ln D)+O(1).
\]
Integral comparison and the modulus window give
\begin{equation}
 H_0=\Theta\!\left(
 \ln\frac{1+\ln(q/n)+(1-\beta)\ln n}
          {1+\ln(q/n)+\ln D}\right)=\Theta(D).
 \label{uf:harmonic-capacity}
\end{equation}
Since $\ell=O(n/D)$, this proves $\ln P=O(n/D)$.

\smallskip\noindent
\emph{Certified auxiliary lifts.}
For each $S$, choose a chain
$\varnothing=S_0\subseteq S_1\subseteq\cdots\subseteq S_r=S$ with
$d_i=|S_i\setminus S_{i-1}|\le B_i$, allowing $d_i=0$.
Such a chain exists by the capacity bound.
For $d_i>0$, independently and uniformly draw
$\mathbf z_i\leftarrow\mathbb F_P^{d_i}\setminus\{\mathbf0\}$ and define
\[
 \Gamma_i=q\mathbb Z^{d_i}+\frac qP\mathbb Z\mathbf z_i,
 \qquad K_i'=\widetilde K_{S_i}+(0\oplus\Gamma_i).
\]
There are exact sequences
\begin{equation}
 0\longrightarrow\Gamma_i\longrightarrow K_i'
 \longrightarrow\widetilde K_{S_{i-1}}\longrightarrow0,
 \qquad K_i'/\widetilde K_{S_i}\simeq\mathbb Z/P\mathbb Z.
 \label{uf:temporary-extension}
\end{equation}
The first map deletes the new block. The second assertion follows from
$\Gamma_i/q\mathbb Z^{d_i}\simeq\mathbb Z/P\mathbb Z$.

Lemma~\ref{uimpl:product-cosets} implements each shifted draw on
$\Gamma_i$ at width $s_{i-1}$. Reduce the center modulo $q$ and
translate the output back. To certify smoothing, use Poisson summation:
\begin{equation}
 F_i:=\rho_{1/u_i}(\Gamma_i^*\setminus\{0\})
 =\frac{q^{d_i}}{Pu_i^{d_i}}\rho_{u_i}(\Gamma_i)-1.
 \label{uf:aux-certificate}
\end{equation}
Each shifted scalar mass is at most the centered mass, so
$\theta(v)\le1+v$ gives $F_i+1\le(1+q/u_i)^{d_i}$.
Take a rational upper bound $B_{\rm cert}$ on this expression with
$\ln B_{\rm cert}=O(n\ln n)$. Enclose the mass and its scale factor
in positive intervals of relative width at most
$\varepsilon_\star/(64B_{\rm cert})$, using
Lemma~\ref{uimpl:product-cosets} for the mass. Outward rounding gives
an interval for $F_i$ of width at most $\varepsilon_\star/4$.
Accept when its upper endpoint is at most $\varepsilon_\star$.
Every accepted lattice satisfies $\eta_{\varepsilon_\star}(\Gamma_i)\le u_i$,
and every lattice with $F_i\le\varepsilon_\star/2$ passes.

For $d_i\ge2$, the capacity condition~\eqref{uf:capacity-condition}
and~\eqref{up:hgh-smoothing} give
$\eta_{\varepsilon_\star/2}(\Gamma_i)\le u_i$ with probability at
least $1-2^{1-d_i}\ge1/2$. For $d_i=1$,
$\Gamma_i=(q/P)\mathbb Z$. Put $v=u_iP/q$.
The capacity condition gives $v\ge16\sqrt{\ell}$, so
\[
 \rho_{1/u_i}(\Gamma_i^*\setminus\{0\})
 =2\sum_{j\ge1}e^{-\pi v^2j^2}
 \le\frac{2e^{-\pi v^2}}{1-e^{-\pi v^2}}
 \le\varepsilon_\star/2.
\]
Thus every such draw passes. A cap of $R_{\rm aux}=O(\ln n)$
independent trials, chosen so that
$|\mathcal S|r2^{-R_{\rm aux}}\le n^{-K-10}$, bounds total retry
failure by $n^{-K-10}$.
For $d_i=0$, skip the lift and retain the prescribed number of
differences in consecutive pairs. The quotient has size one.

\smallskip\noindent
\emph{List recursion and quotient removal.}
The initial persistent lattice is
$J=\widetilde K_\varnothing
=\mathbb Z^n+(q/P_0)\mathbb Z\mathbf a_0$. Choose $k$ uniformly modulo $P_0$
and then sample the coordinates independently from
$D_{\mathbb Z+qk(\mathbf a_0)_j/P_0,s_0}$, using the uniform rule
and Lemma~\ref{uimpl:scalar-coset} with spacing one.
For the reference mixture, Poisson summation gives, uniformly in $k$,
\[
 \rho_{s_0}(\mathbb Z^n+qk\mathbf a_0/P_0)
 =s_0^n\bigl(1+O(ne^{-\pi s_0^2})\bigr).
\]
Thus one reference vector is pointwise
$\delta_0=O(ne^{-\pi s_0^2})$-similar to $D_{J,s_0}$. Independent
repetition is strongly conditionally $\delta_0$-similar.

Put
\[
 \zeta=n^{-K-A_{\mathcal S}-10},\qquad
 \overline\delta_i=4^i\delta_0+5(4^i-1)\varepsilon_\star.
\]
The choices in~\eqref{uf:fine-parameters} give
\(\overline\delta_r<\delta/16\) for all sufficiently large \(n\).
Choose an integer $M_{\rm pair}$ satisfying
\begin{equation}
 M_{\rm pair}\ge
 \max\left\{
 P,\;16P_0\bigl(N_{\rm out}+\ln(1/\zeta)\bigr)
 \right\}.
 \label{uf:working-list-size}
\end{equation}
Round a certified upper bound on the right side upward, with total
additive error at most two.
Start with \(2\cdot3^rM_{\rm pair}\) initial vectors.  At level \(i\),
use the primitive Gaussian lift along the first exact sequence in
\eqref{uf:temporary-extension}, with
\eqref{uf:explicit-shifted-dgs}, and then use exactly the ordered
quotient combination rule of Theorem~\ref{ue:list-engine} for
\(K_i'/\widetilde K_{S_i}\).  Its hypotheses hold because
\[
 \eta_{\varepsilon_\star}(\Gamma_i)\le u_i<s_{i-1},
 \qquad
 \eta_{\varepsilon_\star}(\widetilde K_{S_i})
       \le u_0\le u_i=s_{i-1}/\sqrt2.
\]
The lifting and ordered pairing bounds
\eqref{up:del-lift}--\eqref{up:del-pair} give, after level
\(i\), a list of length \(2\cdot3^{r-i}M_{\rm pair}\), width \(s_i\),
and strong conditional error at most \(\overline\delta_i\).  The same
estimate covers \(d_i=0\).  Hence the recursion produces
\(2M_{\rm pair}\) vectors in \(\widetilde K_S\) at width
\(s=t/\sqrt2\).

Apply Lemma~\ref{uf:fixed-pair-removal} with
\[
 \Lambda=K_S,\quad \Lambda'=\widetilde K_S,\quad
 P=P_0,\quad s=t/\sqrt2,\quad
 N=N_{\rm out},\quad \varepsilon=\varepsilon_\star.
\]
The explicit condition \(s_{\rm mass}^2\ge8s_{\rm search}^2\) gives
\[
 \eta_{\varepsilon_\star}(\widetilde K_S)
 \le u_0=q/s_{\rm mass}\le t/2=s/\sqrt2.
\]
Since \(\overline\delta_r<1/4\), the choice
\eqref{uf:working-list-size} implies
\eqref{uf:fixed-pair-count}.  The completed reference list is therefore
strongly conditionally
\[
 4\overline\delta_r+
 \log\frac{1+\varepsilon_\star}{1-\varepsilon_\star}
 <\delta
\]
similar to \(D_{K_S,t}^{N_{\rm out}}\), and its mismatch with the finite
list is at most \(\zeta\).  A union bound over
\(|\mathcal S|\le n^{A_{\mathcal S}}\) charges all such shortages to the
preprocessing failure.  Applying \(q^{-1}\mathbf U_S^{-1}\) and
\eqref{uf:KS} gives the reference list in \(L_S^*\) at width
\(1/s_{\rm search}\).

Successful reference preprocessing means the intersection of the
global smoothing event~\eqref{uf:global-smoothing}, successful
certified auxiliary selection within every retry cap, and the absence
of a final fixed-pair shortage for every prescribed set. First fix
the common mask and all selected auxiliary lattices. Conditional on
these data, different list constructions use fresh independent
randomness. The conditional form of
Lemma~\ref{uf:fixed-pair-removal} preserves each list's strong
conditional bound when its shortage event is excluded. Conditioning
on the other lists' success has no further effect at fixed data.
Finally, mixing over these data preserves the uniform pointwise
bound, since the target $D_{L_S^*,1/s_{\rm search}}$ is independent of them.

\smallskip\noindent
\emph{Finite precision and cost.}
Rational comparisons with powers of two determine $r$ exactly,
so all squared widths are rational. Bound $X_i$ upward and choose
$X_P$ within two above $e^{L_P}$. Replace $\ln P$ in $B_i$ by a
rational lower bound within $n^{-3}$. This loses at most one per
block from total capacity at least $2n+r$ and preserves
\eqref{uf:capacity-condition}.

The mask choices, auxiliary trials, initial coset and scalar choices,
and complete shifted draws together number at most
\[
 J_{\rm calls}=2+|\mathcal S|
       \bigl(rR_{\rm aux}+2(n+r+1)3^rM_{\rm pair}\bigr).
\]
Choose dyadic $\nu\in[n^{-K-10}/2,n^{-K-10}]$. Give each Gaussian
call accuracy $\nu/(2J_{\rm calls})$ and use total uniform choice
budget $\nu/2$ in Section~\ref{ue:common-sampling}.
Lemma~\ref{up:adaptive-error} gives total implementation error at
most $\nu$. The bad mask, exhausted retries, and pair shortages
have total reference probability at most $3n^{-K-10}$.

Initialization has $s_0^2\in[2q^2/s_{\rm mass}^2,4q^2/s_{\rm mass}^2)$
and spacing one. Lifting and certification have spacing $q$ and
widths $s_{i-1}$ and $u_i$. These width ratios and their inverses
have polynomial magnitude. Also
$\ln P_0,\ln P,\ln M_{\rm pair},\ln(1/\varepsilon_\star),
\ln(J_{\rm calls}/\nu)=O(n/D)$.
The scalar cutoffs and rational lattice denominators give polynomial
bit lengths through $r=O(\ln n)$ levels. The extra certification
precision has polynomial length as well. The total work is at most
\[
 2(1+|\mathcal S|)r3^rM_{\rm pair}P\,\operatorname{poly}
 (n,\ln q,\ln P,\ln P_0,\ln(1/\varepsilon_\star)),
\]
and the final pair scan is linear in its input list. This proves the theorem.
\end{proof}

\subsection{Affine slice search and witness containment}
\label{uf:affine-recovery}

For each anchor position $i\in[n]$, fix an order on
$[n]\setminus\{i\}$. Let $\mathcal S$ collect the sets of coordinates
left free along these orders. Each order has at most $n$ stages,
so $|\mathcal S|\le n^2$. For each $S$, use a dual Gaussian list
$\mathbf Y_1,\ldots,\mathbf Y_N\in L_S^*$ from
Theorem~\ref{uf:persistent-theorem}.

Embed $L_S$ in the subspace $\mathcal U_S\subseteq\mathbb R^{2n}$
of the free coordinates. Fixing the other coordinates gives a slice
$A=\mathbf s+L_S$. Decompose
$\mathbf s=\mathbf s_U+\mathbf s_\perp$ along
$\mathcal U_S\oplus\mathcal U_S^\perp$. Poisson summation gives
\begin{equation}
 \frac{\rho_{s_{\rm search}}(A)}{\rho_{s_{\rm search}}(L_S)}
 =e^{-\pi\|\mathbf s_\perp\|^2/s_{\rm search}^2}
   \mathbb E_{\mathbf Y\leftarrow D_{L_S^*,1/s_{\rm search}}}
      \left[\cos(2\pi\langle\mathbf Y,\mathbf s_U\rangle)\right].
 \label{uf:slice-poisson}
\end{equation}
The corresponding sample estimate is
\[
 \widehat w(A)=
 e^{-\pi\|\mathbf s_\perp\|^2/s_{\rm search}^2}
 \frac1N\sum_{k=1}^N
 \cos(2\pi\langle\mathbf Y_k,\mathbf s_U\rangle).
\]
For fixed $S$, changing the assigned coordinate values changes
$\mathbf s$ while reusing the same samples. Candidate children share
the denominator $\rho_{s_{\rm search}}(L_S)$, which cancels when
their weights are normalized. The procedure below evaluates this
estimate with certified numerical error and clamps it at zero.

Since $q\mathbb Z^{n+|S|}\subseteq L_S$, we have
$L_S^*\subseteq q^{-1}\mathbb Z^{n+|S|}$. Thus the cosine phase
depends only on $\mathbf s_U\bmod q$. There are at most $n^2q^{2n}$
phases across all lists. Simultaneous concentration uses the logarithm
of this count to bound the required list length. The search evaluates
the queries encountered on its sampled paths.

Define
\[
 \mathcal{W}_B(h)=\{\mathbf0\ne\mathbf z=(\mathbf u,\mathbf v)\in\Lambda_h:
                          \|\mathbf z\|\le B\}.
\]

\paragraph{The search procedure.}
Fix $h$, $s_{\rm search}>0$, $A\ge1$, $B>0$, and
$0<\varepsilon_{\rm qry},\varepsilon_{\rm hit}<1$ such that
\[
 1\le B<q/2,
 \qquad \rho_{s_{\rm search}}(\Lambda_h)\le A,
\]
and put
\begin{equation}
 H=\lfloor B\rfloor,
 \qquad p_B=A^{-1}e^{-\pi B^2/s_{\rm search}^2},
 \qquad \xi=\frac{p_B}{200n(2H+1)}.
 \label{uf:witness-mass}
\end{equation}
Let $0<\nu_{\rm list},\nu_{\rm cat}<1$ and put
$\nu=\nu_{\rm list}+\nu_{\rm cat}<1$. Given an $N$-entry list for
each $S\in\mathcal S$, construct the table $\mathcal C(h)$ as follows.
\begin{enumerate}
\item Evaluate every use of~\eqref{uf:slice-poisson} with certified numerical
error at most $\xi/4$, compute rational estimates with that certified
absolute error, and clamp them at zero.

\item For each $i\in[n]$ and $a\in[1,H]$, run the affine slice search from
the anchor $u_i=a$. At every later coordinate retain only values in
$[-H,H]$ consistent with the partial energy bound. Choose the next value
with probability proportional to the resulting nonnegative rational
weights. A zero total weight is a declared failure.
\item Run every anchored branch independently
\begin{equation}
 T=\left\lceil C_3p_B^{-1}
                    \ln(1/\varepsilon_{\rm hit})\right\rceil
 \label{uf:search-repetitions}
\end{equation}
times, completing all runs even if one has already found a verified output.
There are at most $J_{\rm cat}=n^2HT$ adaptive choices from the rational
weights. Use the finite weight rule in Section~\ref{ue:common-sampling}
with accuracy $\xi_{\rm cat}=\nu_{\rm cat}/J_{\rm cat}$ for each choice.
Lemma~\ref{up:adaptive-error} gives total error at most $\nu_{\rm cat}$.
\item Once all $\mathbf u$ coordinates are fixed, set
$\mathbf v=\operatorname{lift}_q^{\rm cen}(\mathbf H_h\mathbf u)$.
Verify lattice
membership, nonzeroness, and the norm bound exactly. Retain every
verified pair and also its negative. Call the resulting table
$\mathcal C(h)$.
\end{enumerate}

\begin{theorem}[Witness containment]
\label{uf:candidate-table}
Take $s_{\rm search}^2,A,B^2$ and the error budgets to be positive rationals
given in binary. Let $L_{\rm enc}$ bound their bit lengths and those
of the list coordinates.
There are absolute constants $C_2,C_3>0$ for which the procedure above
has the following guarantee. Suppose that:
\begin{enumerate}
\item the lists have preprocessing failure at most $\varepsilon_{\rm pre}$.
\item their reference laws are strongly conditionally $\delta$-similar
to the corresponding $D_{L_S^*,1/s_{\rm search}}^N$, and their joint
implementation has total variation error at most $\nu_{\rm list}$.
\item the accuracy and list length satisfy
\begin{equation}
 \delta\le\xi/16,\qquad
 N\ge C_2\xi^{-2}
       \left(2n\ln q+\ln\frac{2n^2}{\varepsilon_{\rm qry}}\right).
 \label{uf:query-list-size}
\end{equation}
\end{enumerate}
Then, for every fixed $\mathbf z\in\mathcal{W}_B(h)$,
\begin{equation}
 \Pr[\mathcal C(h)\cap\{\mathbf z,-\mathbf z\}=\varnothing]
 \le \varepsilon_{\rm pre}+\nu+
       \varepsilon_{\rm qry}+\varepsilon_{\rm hit}.
 \label{uf:fixed-witness}
\end{equation}
Moreover,
\begin{equation}
 |\mathcal C(h)|\le2nHT.
 \label{uf:table-size}
\end{equation}
Writing $T_{\rm list}$ for the cost of constructing all the prescribed
lists, the list construction plus all slice queries takes
\[
 T_{\rm list}+\operatorname{poly}
 (n,H,\ln q,L_{\rm enc},\ln T,\ln(1/\xi),\ln(1/\nu_{\rm cat}))NT
\]
bit operations. Its space is bounded by the list storage, the bit
storage of the table, and polynomial working storage in these
parameters.
\end{theorem}
\begin{proof}
\emph{Uniform query accuracy.}
Certified lower bounds for the accuracy thresholds and upper bounds
for the required $N,T$ give finite choices with fixed constant margins.
In the public search
schedule, $L_{\rm enc}$ and $\ln(1/\xi)$ are polynomial in $n$.
For a reference list, Lemma~\ref{up:concentration}, followed by a union bound
over at most $n^2q^{2n}$ phases shows that
\eqref{uf:query-list-size} makes every empirical cosine accurate within
$3\xi/4$, except with probability $\varepsilon_{\rm qry}$. Certified
finite-precision evaluation uses the remaining $\xi/4$. The
total variation bound $\nu$ transfers the complete reference analysis
below to the implemented procedure.

\smallskip\noindent
\emph{A fixed witness path.}
Fix $\mathbf z=(\mathbf u,\mathbf v)\in\mathcal{W}_B(h)$.
If $\mathbf u=\mathbf0$, then
$\mathbf v\in q\mathbb Z^n$. Since $\|\mathbf v\|<q/2$,
this would force $\mathbf v=\mathbf0$, a contradiction.
Choose $i$ with $u_i\ne0$, and replace $\mathbf z$ by $-\mathbf z$ if necessary so that
$a=u_i\in[1,H]$. The anchored branch $(i,a)$ is therefore present.

Every prefix slice on the path to $\mathbf z$ contains $\mathbf z$. Since its direction
lattice embeds in $\Lambda_h$, the correct child's normalized mass is
at least
\[
 \frac{\rho_{s_{\rm search}}(\mathbf z)}{\rho_{s_{\rm search}}(\Lambda_h)}\ge p_B.
\]
If a level has at most $2H+1$ children, additive error at most $\xi$
changes the correct transition probability by a multiplicative factor
at least
\[
 \frac{1-\xi/p_B}{1+(2H+1)\xi/p_B}.
\]
The product of these factors over at most $n$ levels is bounded below
by an absolute positive constant. To account explicitly for the
retained candidates, let $C_j$ be their set at level $j$, let
$A_{j-1}$ be the parent slice, and let $A_j$ be the correct child.
The child slices are disjoint subsets of the parent, so the ideal
transition probability satisfies
\[
 P_j=\frac{\rho_{s_{\rm search}}(A_j)}
             {\sum_{b\in C_j}\rho_{s_{\rm search}}(A_{j,b})}
 \ge\frac{\rho_{s_{\rm search}}(A_j)}{\rho_{s_{\rm search}}(A_{j-1})}.
\]
Multiplying these lower bounds gives
\[
 \prod_jP_j\ge
 \frac{\rho_{s_{\rm search}}(A_{\rm leaf})}{\rho_{s_{\rm search}}(A_{\rm anchor})}
 \ge\frac{\rho_{s_{\rm search}}(\mathbf z)}{\rho_{s_{\rm search}}(\Lambda_h)}\ge p_B.
\]
Here the leaf contains the witness, while the anchor is a subset of
$\Lambda_h$. Once all $\mathbf u$ coordinates are fixed, the centered
representative completes $\mathbf v$ uniquely because $B<q/2$.
Hence one
run of the correct branch outputs $\mathbf z$ with probability at least
$c_2p_B$ for an absolute $c_2>0$.

\smallskip\noindent
\emph{Repetition and verification.}
The probability that all $T$ repetitions miss this fixed witness is at
most $(1-c_2p_B)^T\le e^{-c_2p_BT}\le\varepsilon_{\rm hit}$ after
choosing $C_3$ sufficiently large. Each coordinate of $\mathbf v$ lies in
$(-q/2,q/2)$, so the centered representative of
$\mathbf H_h\mathbf u\bmod q$ is exactly
$\mathbf v$ and the witness passes the final verification. Every anchored run
contributes at most one raw pair. Retaining both signs proves
\eqref{uf:table-size}.
\end{proof}

For a random witness, condition on the realized pair $(h,\mathbf z)$, apply
\eqref{uf:fixed-witness}, and average. The uniform bound for each fixed
witness gives the same guarantee under any distribution on such pairs.

\subsection{Falcon: quality, Gaussian mass, and basis completion}
\label{uf:regular-generator}

This subsection gives two ways to obtain the Gaussian mass bound
needed for recovery in the negacyclic ring. We first record the
shared lattice duality and basis completion facts. We then prove
Theorem~\ref{uf:falcon-quality-recovery} from Falcon's key generation
quality condition in Section~\ref{uf:quality-proof}. The final part
derives a separate recovery guarantee from spectral bounds on a
short pair.

Let
\[
 n=2^k,\qquad R=\mathbb Z[x]/(x^n+1),\qquad R_q=R/qR,
\]
where $q\equiv1\pmod{2n}$ is prime. We use the coefficient embedding
$\operatorname{coeff}:R\to\mathbb Z^n$ and its induced Euclidean
norm. For $h\in R_q$, let $\mathbf H_h$ be its negacyclic multiplication matrix,
so that $\operatorname{coeff}(hu)=\mathbf H_h\operatorname{coeff}(u)\pmod q$,
and put
\begin{equation}
 \Lambda_h
 =\operatorname{coeff}\{(u,v)\in R^2:hu-v=0\pmod q\}
 =\ker_q[\mathbf H_h\mid-\mathbf I_n].
 \label{uf:lattice}
\end{equation}
Here and below $\operatorname{coeff}$ acts blockwise on ring pairs and
sets. We identify a ring submodule with its coefficient lattice when
taking intersections, orthogonal projections, or Gaussian masses.

\begin{lemma}
\label{uf:public-isoduality}
Write $a^*=a(x^{-1})$ for coefficient conjugation. For every
$h\in R_q$, the orthogonal map
\begin{equation}
 \mathbf J\operatorname{coeff}(u,v)=\operatorname{coeff}(v^*,-u^*)
 \quad\text{satisfies}\quad q\Lambda_h^*=\mathbf J\Lambda_h.
 \label{uf:public-duality}
\end{equation}
\end{lemma}
\begin{proof}
Let $\mathbf H$ be an integer negacyclic multiplication matrix for $h$
and let $\mathbf S$ be the signed permutation matrix of conjugation.
Then $\mathbf H^T\mathbf S=\mathbf S\mathbf H$, and the systematic
form of~\eqref{uf:lattice} gives
\[
 q\Lambda_h^*=\{(\mathbf a,\mathbf b)\in\mathbb Z^{2n}:
               \mathbf a+\mathbf H^T\mathbf b=0\pmod q\}.
\]
For $(\mathbf u,\mathbf v)\in\Lambda_h$,
$\mathbf S\mathbf v-\mathbf H^T\mathbf S\mathbf u
=\mathbf S(\mathbf v-\mathbf H\mathbf u)=0\pmod q$.
Thus $\mathbf J\Lambda_h\subseteq q\Lambda_h^*$.
Both lattices have determinant $q^n$, proving equality.
\end{proof}

\subsubsection{Basis completion}
\label{uf:completion}

For the conventional row basis orientation, put
\begin{equation}
 \Lambda_h^{\rm std}=\operatorname{coeff}\{(a,b)\in R^2:a+bh=0\pmod q\}.
 \label{uf:standard-lattice}
\end{equation}
It is the image of~\eqref{uf:lattice} under the signed coordinate
permutation $\operatorname{coeff}(u,v)\mapsto\operatorname{coeff}(v,-u)$.
For $u,v,F,G\in R$, define
\begin{equation}
 \mathbf B(u,v,F,G)=
 \begin{pmatrix}v&-u\\G&-F\end{pmatrix}.
 \label{uf:ntru-basis}
\end{equation}
The equation
\begin{equation}
 uG-vF=q
 \label{uf:ntru-equation}
\end{equation}
makes its determinant equal to $q$. If $u$ is invertible modulo $q$ and
$vu^{-1}=h$, the coefficient vectors of both rows of~\eqref{uf:ntru-basis} lie in
$\Lambda_h^{\rm std}$. Their coefficient row lattice has determinant
$q^n$, equal to the index of~\eqref{uf:standard-lattice}, so they form
an NTRU basis of that lattice.

\begin{lemma}
\label{uf:polynomial-completion}
Given $q,u,v$, one can decide in time polynomial in their coefficient
bit length and in $n$ whether~\eqref{uf:ntru-equation} has an integer
solution, and return a solution of polynomial bit length whenever it
exists. If $u\bmod q$ is a unit and $vu^{-1}=h$, this solution gives
a basis of $\Lambda_h^{\rm std}$.
\end{lemma}
\begin{proof}
Let $\mathbf M_u,\mathbf M_v$ be the integer coefficient multiplication matrices,
defined by $\mathbf M_u\operatorname{coeff}(a)=\operatorname{coeff}(ua)$
and likewise for $v$.
The equation is the integer linear system
\[
 [\mathbf M_u\mid-\mathbf M_v]
 \begin{pmatrix}\operatorname{coeff}(G)\\
                 \operatorname{coeff}(F)\end{pmatrix}=q\mathbf e_0.
\]
Here $\mathbf e_0=\operatorname{coeff}(1)$.
Polynomial-time Smith or Hermite normal form algorithms, with their
transformation matrices, decide solvability and produce a solution
of polynomial bit length~\cite{KB79}. The determinant and public
relation argument above proves the basis claim.
\end{proof}

Define $\mathsf{CompleteAndCheck}_{q,B}(h,u,v)$ to
perform the unit, public relation, and norm checks below, solve the
integer system using Lemma~\ref{uf:polynomial-completion}, and reject
if a check fails or no solution exists. It accepts only after verifying
\begin{equation}
 \begin{gathered}
 u\bmod q\in R_q^\times,\qquad
 vu^{-1}=h\pmod q,
 \qquad \|(u,v)\|\le B,\\
 uG-vF=q\ \text{in }R.
 \end{gathered}
 \label{uf:completion-checks}
\end{equation}
All checks use exact integer or rational arithmetic. The procedure
decides the displayed conditions and returns an integer completion
for every accepted pair. For $B=O(\sqrt q)$, each candidate is
processed in $T_{\rm comp}(n,q)=\operatorname{poly}(n,\log q)$ time.

\subsubsection{Recovery from the key quality condition}
\label{uf:quality-proof}

For $a^*=a(x^{-1})$, define
\begin{equation}
 \gamma_q(u,v)=\max\left\{\|(v,-u)\|_2,
 \left\|\left(\frac{qu^*}{uu^*+vv^*},
                   \frac{qv^*}{uu^*+vv^*}\right)\right\|_2\right\}.
 \label{uf:falcon-quality}
\end{equation}
Compute the inverse in $R\otimes\mathbb Q$ and set
$\gamma_q(u,v)=+\infty$ if it does not exist. All norms are
Euclidean coefficient norms. Fix a rational constant $C\ge1$,
independent of $n$, and impose the exact quality predicate
$\gamma_q(u,v)^2\le C^2q$.
Falcon uses $C=117/100$ in this predicate
\cite[Equation~(3.28) and Algorithm~5]{FalconSpec}.
A quality-qualified key is one admitting the algebraic and quality
witness in~\eqref{ui:falcon-quality-family}.

\begin{lemma}
\label{uf:completion-quality}
Suppose $u\bmod q$ is a unit, $vu^{-1}=h\pmod q$, and
$uG-vF=q$ in $R$. Order the $2n$ coefficient rows as the $n$
shifts of $(v,-u)$ followed by the $n$ shifts of $(G,-F)$.
They form a basis of $\Lambda_h^{\rm std}$ with maximum
Gram--Schmidt norm at most $\gamma_q(u,v)$, for every integer
completion $(F,G)$. The quality predicate is decidable exactly
in polynomial bit complexity.
\end{lemma}
\begin{proof}
The ring matrix~\eqref{uf:ntru-basis} has determinant $q$, so its
full coefficient matrix has determinant of absolute value $q^n$.
Its rows lie in $\Lambda_h^{\rm std}$, whose determinant is also
$q^n$, and hence form a full basis.
Put $d=uu^*+vv^*$, $b_1=(v,-u)$, $b_2=(G,-F)$, and
$\mu=(Gv^*+Fu^*)/d$. The NTRU equation ensures that $u,v$ have
no common complex zero, so $d$ is invertible. It also gives
\begin{equation}
 b_2-\mu b_1=\left(\frac{qu^*}{d},\frac{qv^*}{d}\right).
 \label{uf:completion-projection}
\end{equation}
The right-hand side is orthogonal to every coefficient shift of $b_1$,
and $\mu b_1$ lies in their real span. All shifts of $b_1$ have
norm $\|b_1\|_2$, so every Gram--Schmidt length in the first block
is at most this value. After projecting away that block, the second
block consists of shifts of the right-hand side of
\eqref{uf:completion-projection}. Further orthogonalization cannot
increase their common norm. This proves the quality bound for any
completion.

For exact verification, form the integer multiplication matrix of
$d$. Rational linear algebra either detects singularity or computes
$d^{-1}$ with polynomial coefficient bit length. Exact squared
coefficient norms and comparison with the rational $C^2q$ therefore
take polynomial bit complexity.
\end{proof}

\begin{lemma}
\label{uf:gs-mass-product}
Let $L$ be a full-rank lattice of dimension $d$, with basis
$(b_1,\ldots,b_d)$ and Gram--Schmidt lengths
$r_i=\|\widetilde b_i\|_2>0$. For every $s>0$,
\[
 \rho_s(L)\le\prod_{i=1}^d\theta(s/r_i).
\]
\end{lemma}
\begin{proof}
For real Gram--Schmidt coefficients $\mu_{ji}$,
\[
 \left\|\sum_{i=1}^d z_ib_i\right\|_2^2
 =\sum_{i=1}^d r_i^2
       \left(z_i+\sum_{j>i}\mu_{ji}z_j\right)^2.
\]
Fixing $z_2,\ldots,z_d$, sum first over $z_1$. The shifted
one-dimensional Gaussian mass is at most $\theta(s/r_1)$ by
Poisson summation. Repeating for $z_2,\ldots,z_d$ gives the
product bound. Nonnegativity justifies the iterated sums.
\end{proof}

\begin{lemma}
\label{uf:isodual-quality-mass}
Let $L$ be a full-rank $d$-dimensional lattice with
$qL^*=JL$ for an orthogonal map $J$. If $L$ has a basis whose
Gram--Schmidt lengths are at most $C\sqrt q$, then for every $s>0$,
\[
 \rho_s(L)\le\theta(Cs/\sqrt q)^d.
\]
\end{lemma}
\begin{proof}
Let $(b_1^\vee,\ldots,b_d^\vee)$ be the dual basis, so that
$\langle b_i,b_j^\vee\rangle=\delta_{ij}$. The reversed dual
basis $(b_d^\vee,\ldots,b_1^\vee)$ has Gram--Schmidt lengths
$1/r_d,\ldots,1/r_1$. This follows by inverse transposition of
a QR decomposition and reversal of the orthogonalization order.
Hence $qJ^{-1}(b_d^\vee,\ldots,b_1^\vee)$ is a basis of $L$
with Gram--Schmidt lengths at least $\sqrt q/C$.
Lemma~\ref{uf:gs-mass-product} and monotonicity of $\theta$ prove
the claim.
\end{proof}

\begin{proposition}
\label{uf:falcon-quality-mass}
For every quality-qualified public key and every $s>0$,
\begin{equation}
 \rho_s(\Lambda_h)\le\theta(Cs/\sqrt q)^{2n}.
 \label{uf:quality-mass-bound}
\end{equation}
In particular, for fixed $1/2<\beta<1$ and the finite rational
width square $s_{\rm mass}^2\le q/n^\beta$, for all sufficiently large $n$,
\begin{equation}
 \begin{aligned}
 \rho_{s_{\rm mass}}(\Lambda_h)
 &\le\theta(Cn^{-\beta/2})^{2n}\\
 &\le\exp\!\left(6n\exp(-\pi n^\beta/C^2)\right)
 =1+n^{-\omega(1)}.
 \end{aligned}
 \label{uf:quality-search-mass}
\end{equation}
Thus the public mass budget $A_{\rm mass}=2$ works uniformly for these keys.
\end{proposition}
\begin{proof}
The orthogonal transformation $P(u,v)=(v,-u)$ maps $\Lambda_h$
to $\Lambda_h^{\rm std}$. Apply $P^{-1}$ to the full quality basis
from Lemma~\ref{uf:completion-quality}. Its Gram--Schmidt lengths
remain at most $C\sqrt q$.
Lemmas~\ref{uf:public-isoduality} and~\ref{uf:isodual-quality-mass}
give~\eqref{uf:quality-mass-bound} in dimension $2n$.
Eventually $Cn^{-\beta/2}\le1$. Using
$\theta(t)\le1+3e^{-\pi/t^2}$ for $0<t\le1$ then
proves~\eqref{uf:quality-search-mass}.
\end{proof}

\begin{proof}[Proof of Theorem~\ref{uf:falcon-quality-recovery}]
Apply Theorem~\ref{uq:mass-search-framework} with $m=n$, fixed rational
$1/2<\beta<1$, and $c_{\rm search}=1$. Use its finite schedule
\[
 \begin{gathered}
 0\le\ln\ln n-D_{\rm Q}\le2^{-n^2},\qquad
 (1-2^{-n^2})q/n^\beta\le s_{\rm mass}^2\le q/n^\beta,\\
 s_{\rm search}^2=qD_{\rm Q}/n.
 \end{gathered}
\]
The congruence on $q$ gives $q\ge2n+1$, and the modulus window
implies $q\le n^2$ eventually. Proposition~\ref{uf:falcon-quality-mass}
supplies $A_{\rm mass}=2$, with $\ln A_{\rm mass}=O(n/D_{\rm Q})$.
For finite arithmetic, the conditions $C^2\le n^\beta$ and
$6n\exp(-\pi n^\beta/C^2)\le\ln2$ can be certified by one-sided
rational approximations as in Appendix~\ref{ue:finite}. They use
only public constants and parameters.
Moreover,
\[
 \frac{s_{\rm mass}^2}{s_{\rm search}^2}
 \ge(1-2^{-n^2})\frac{n^{1-\beta}}{D_{\rm Q}}\longrightarrow\infty,
\]
so $s_{\rm mass}^2\ge8s_{\rm search}^2$ for all sufficiently large $n$.
Choose the public rational radius square $B_{\rm Q}^2=C^2q$.
Then $1\le B_{\rm Q}=C\sqrt q<q/2$ eventually.

The verifier first checks nonzeroness, modular invertibility of $u$,
$hu=v\pmod q$, and $\|(u,v)\|_2^2\le C^2q$. It then uses
Lemma~\ref{uf:polynomial-completion} to solve $uG-vF=q$ over $R$,
rejecting when no integer solution exists, and applies the exact
quality predicate $\gamma_q(u,v)^2\le C^2q$.
All candidate coefficients have absolute value at most $C\sqrt q$.
The completion and quality checks therefore have polynomial bit
complexity and return $F,G$ of polynomial bit length.

The witness in the hypothesis passes every check, so the common
search theorem returns an accepted candidate with probability at
least $1-n^{-K}$ in time and space $2^{O(n/\ln\ln n)}$.
Lemma~\ref{uf:completion-quality} gives the stated quality of its
full coefficient basis. The quality basis gives the mass bound,
and the algorithm uses the public budget $A_{\rm mass}=2$.
Its probability bound is uniform over all qualified public keys.
\end{proof}

\subsubsection{Recovery from spectral bounds}
\label{uf:spectral-proof}

We now derive a mass bound from the Fourier values of a short pair.
These values control lengths in the pair's polynomial span and in
the orthogonal projection. The resulting bound gives the recovery
theorem below and its distributional corollary.

\begin{definition}
\label{uf:falcon-witness}
Fix public constants $1/2<\alpha<\beta<1$, $C_{\rm sp}>0$, and
rational $C_{\rm len}>0$. For a fixed public quotient $h\in R_q$,
an admissible witness is a pair $(f,g)\in R^2$ for which there exist
$F_{\!N},G_{\!N}\in R$ satisfying
\begin{equation}
 \begin{gathered}
 f\bmod q\in R_q^\times,\qquad h=gf^{-1}\pmod q,\qquad
 fG_{\!N}-gF_{\!N}=q,\\
 \|(f,g)\|_2\le C_{\rm len}\sqrt q,
 \end{gathered}
 \label{uf:gq1}
\end{equation}
and whose spectrum satisfies
\begin{equation}
 \frac q{n^\alpha}
 \le G_{\min}:=\min_{\zeta^n=-1}
       (|f(\zeta)|^2+|g(\zeta)|^2)
 \le G_{\max}:=\max_{\zeta^n=-1}
       (|f(\zeta)|^2+|g(\zeta)|^2)
 \le C_{\rm sp}q\ln n.
 \label{uf:gq2}
\end{equation}
The extrema $G_{\min},G_{\max}$ range over the complex roots
of $x^n+1$.
The pair and its completion establish the existence of a valid target.
The algorithm takes the public quotient and bounds as input.
\end{definition}

The spectral bounds control the Gaussian mass in the span of the short
pair and in its orthogonal projection. The NTRU equation makes the
short pair sublattice primitive, so these two estimates together bound
the mass of the full public lattice.

\begin{lemma}
\label{uf:coprime}
If $f\bmod q$ is a unit, then $(f,g)=R$ if and only if there
exist $F_{\!N},G_{\!N}\in R$ with $fG_{\!N}-gF_{\!N}=q$.
In particular, every pair satisfying~\eqref{uf:gq1} has $(f,g)=R$.
\end{lemma}
\begin{proof}
Since $f$ is a unit modulo $q$, there are $a,k\in R$ with
$af+qk=1$.  Substituting $q=fG_{\!N}-gF_{\!N}$ gives
\[
 1=f(a+kG_{\!N})-g(kF_{\!N}),
\]
so $1\in(f,g)$.
Conversely, if $rf+tg=1$, take $G_{\!N}=qr$ and
$F_{\!N}=-qt$.
\end{proof}

With the constants of Definition~\ref{uf:falcon-witness}, put
\begin{equation}
 s_{\rm mass}^2=\frac q{n^\beta}.
 \label{uf:fstar}
\end{equation}

\begin{proposition}
\label{uf:source-geometry}
If~\eqref{uf:gq1}--\eqref{uf:gq2} hold, then, for every $s>0$,
\begin{equation}
 1\le \rho_s(\Lambda_h)
 \le \theta\!\left(\frac{s}{\sqrt{G_{\min}}}\right)^n
       \theta\!\left(\frac{s\sqrt{G_{\max}}}{q}\right)^n.
 \label{uf:master-mass}
\end{equation}
Consequently, for every fixed $c_0>0$,
\begin{equation}
 \rho_{c_0s_{\rm mass}}(\Lambda_h)=1+n^{-\omega(1)}.
 \label{uf:fstar-mass}
\end{equation}
If $0<D=D(n)\le1+\ln\ln n$ and
$s_{\rm search}^2=c_{\rm search} qD/n$ for fixed $c_{\rm search}>0$, then also
\begin{equation}
 \rho_{\sqrt2s_{\rm search}}(\Lambda_h)=1+n^{-\omega(1)}.
 \label{uf:coarse-mass}
\end{equation}
\end{proposition}
\begin{proof}
Let $K=\mathbb Q[x]/(x^n+1)$, let
$W=K\mathbin{\cdot}(f,g)\subset K^2$, and define $A:K\to K^2$ by
$A(a)=a(f,g)$. Write $\mathbf A$ for its coefficient matrix, so that
$\operatorname{coeff}(A(a))=\mathbf A\operatorname{coeff}(a)$.
The negacyclic Fourier transform is unitary in
coefficient norm, so the eigenvalues of $\mathbf A^T\mathbf A$ are precisely
$|f(\zeta)|^2+|g(\zeta)|^2$.

Lemma~\ref{uf:coprime} implies
\begin{equation}
 \Lambda_h\cap W=AR.
 \label{uf:intersection}
\end{equation}
Indeed, $AR\subseteq\Lambda_h$.  Conversely, if
$a(f,g)\in R^2$ and $rf+tg=1$ with $r,t\in R$, then
$a=r(af)+t(ag)\in R$.

Lemma~\ref{uf:public-isoduality} gives
$q\Lambda_h^*=\mathbf J\Lambda_h$, and $\mathbf JW=W^\perp$.
If $Q=\pi_{W^\perp}\Lambda_h$, duality and
\eqref{uf:intersection} give
\begin{equation}
 Q^*=\Lambda_h^*\cap W^\perp=q^{-1}\mathbf J(AR),
 \qquad Q\simeq q(AR)^*.
 \label{uf:orthogonal-quotient}
\end{equation}

For $a\in R$,
$\|A(a)\|^2\ge G_{\min}\|a\|^2$, and hence
\[
 \rho_s(AR)\le\theta(s/\sqrt{G_{\min}})^n.
\]
Writing $(AR)^*=\mathbf A(\mathbf A^T\mathbf A)^{-1}\mathbb Z^n$
shows that the vector indexed by $\mathbf z\in\mathbb Z^n$ has squared norm
$\mathbf z^T(\mathbf A^T\mathbf A)^{-1}\mathbf z\ge\|\mathbf z\|^2/G_{\max}$.
Equation
\eqref{uf:orthogonal-quotient} therefore gives
\[
 \rho_s(Q)\le\theta(s\sqrt{G_{\max}}/q)^n.
\]
Decompose $\Lambda_h$ into the fibres of its orthogonal projection onto
$Q$.  Every fibre is a translate of $AR$, and
$\rho_s(AR+\mathbf c)\le\rho_s(AR)$ by the Gaussian shift inequality recalled
in Section~\ref{ui:prelim}.  This proves~\eqref{uf:master-mass}.

For $0<x\le1$,
\[
 \theta(x)-1\le3e^{-\pi/x^2},\qquad
 \theta(x)^n\le\exp(3ne^{-\pi/x^2}).
\]
At $s=c_0s_{\rm mass}$, the squared theta arguments in
\eqref{uf:master-mass} are at most
\[
 c_0^2n^{\alpha-\beta},
 \qquad C_{\rm sp}c_0^2\frac{\ln n}{n^\beta},
\]
so both theta contributions equal $1+n^{-\omega(1)}$.  This proves
\eqref{uf:fstar-mass}.  Finally,
$2s_{\rm search}^2/s_{\rm mass}^2=2c_{\rm search} Dn^{\beta-1}=o(1)$. Monotonicity of
Gaussian mass and~\eqref{uf:fstar-mass} prove
\eqref{uf:coarse-mass}.
\end{proof}

\begin{theorem}[Negacyclic recovery]
\label{uf:falcon-pointwise}
Fix the public constants of Definition~\ref{uf:falcon-witness}. Along any
given sequence of powers of two $n$ and primes $q\equiv1\pmod{2n}$ with
\[
 \ln(q/n)\le(\ln n)^{1-c}
\]
for a fixed $c\in(0,1)$, let $h$ be any fixed public quotient admitting
an admissible witness. From $h$, the public parameters, and the
prescribed bounds alone, a randomized procedure returns
$u,v,F,G\in R$ satisfying
\[
 \begin{gathered}
 u\bmod q\in R_q^\times,\qquad vu^{-1}=h\pmod q,\\
 \|(u,v)\|_2\le C_{\rm len}\sqrt q,\qquad uG-vF=q,
 \end{gathered}
\]
and hence a certified NTRU basis of $\Lambda_h^{\rm std}$,
with probability $1-o(1)$ over its internal
randomness, uniformly over these instances.
Its time and space are
\begin{equation}
 \exp(O(n/\ln\ln n)).
 \label{uf:falcon-cost}
\end{equation}
\end{theorem}
\begin{proof}
Take the finite schedule of Theorem~\ref{uq:mass-search-framework}
with $m=n$ and $B_{\rm Q}=C_{\rm len}\sqrt q$.
The modulus window gives $q\le n^2$ eventually, and
$s_{\rm mass}^2/s_{\rm search}^2\to\infty$.
Proposition~\ref{uf:source-geometry} and monotonicity in the width
supply the public mass bound $A_{\rm mass}=2$ at these finite parameters.
Use $\mathsf{CompleteAndCheck}$ as the verifier:
Lemma~\ref{uf:polynomial-completion} proves that it runs in polynomial
time and accepts the admissible witness. The common theorem gives
the short pair, its completion, and the stated cost.
\end{proof}

\begin{corollary}
\label{uf:generator-corollary}
Under the parameter conditions of Theorem~\ref{uf:falcon-pointwise},
let a distribution on $(h,f,g)$ satisfy~\eqref{uf:gq1} for every
output and~\eqref{uf:gq2} except with probability $\omega_n$.
Then recovery of a short pair and certified integer NTRU basis succeeds with
probability at least $1-\omega_n-o(1)$ over the key and independent
algorithmic randomness, in time and space $2^{O(n/\ln\ln n)}$.
\end{corollary}
\begin{proof}
Apply the uniform guarantee of Theorem~\ref{uf:falcon-pointwise} to
each key on the spectral event and average.
\end{proof}

More generally, if a proposal law $\lambda$ is conditioned on an
acceptance event $\mathcal A$ of positive probability, then any
required event $\mathcal E$ satisfies
\begin{equation}
 \Pr_\lambda[\mathcal E^c\mid\mathcal A]
 \le\frac{\Pr_\lambda[\mathcal E^c]}{\Pr_\lambda[\mathcal A]}.
 \label{uf:generic-acceptance-transfer}
\end{equation}
Here $\mathcal A$ is the full acceptance event. A direct estimate
under the accepted law also supplies the corollary's hypothesis.

\subsection{Cyclic NTRU: saturation and recovery}
\label{uf:cyclic-ntru}

We first express the Gaussian mass of the cyclic public lattice in
terms of spectral bounds and a saturation index. We then give
conditions under which this mass bound yields recovery. For the HPS
and HRSS key distributions, the remaining parts establish the
spectral and saturation bounds and combine them after key generation
acceptance.

Let $N$ be an odd prime, put
\[
 R_N=\mathbb Z[x]/(x^N-1),\qquad
 S_N=\mathbb Z[x]/(\Phi_N),\qquad q=2^{b_q},
\]
and use the standard coefficient norm on $R_N\simeq\mathbb Z^N$.
The specification gives
\begin{equation}
 hf=3g\pmod{(q,x^N-1)}.
 \label{uf:ntru-public-relation}
\end{equation}
Since $3$ is a unit modulo $q$, the public scaling
\begin{equation}
 \widehat h=3^{-1}h\pmod q
 \label{uf:ntru-public-scaling}
\end{equation}
turns~\eqref{uf:ntru-public-relation} into
$\widehat hf=g\pmod q$. Let $\mathbf H_{\widehat h}^{\rm cyc}$ be the cyclic
multiplication matrix, acting on coefficient vectors, and set
\begin{equation}
 \Lambda_{\widehat h}^{\rm cyc}
 =\operatorname{coeff}\{(u,v)\in R_N^2:\widehat hu-v\in qR_N\}
 =\ker_q[\mathbf H_{\widehat h}^{\rm cyc}\mid-\mathbf I_N].
 \label{uf:ntru-cyclic-lattice}
\end{equation}
Write $\mathcal T$ for the nonzero ternary representatives of degree at
most $N-2$, $\mathcal T^+$ for those with nonnegative adjacent
correlation, and $\mathcal T(d)$ for those with exactly $d/2$ coefficients
of each sign.  For $\mathsf P\in\{\mathsf{HPS},\mathsf{HRSS}\}$, define
$\mathsf{Valid}_{\mathsf P}(h;u,v)$ by the following exact checks:
\begin{itemize}
\item for HPS, $u\in\mathcal T$, $v\in\mathcal T(q/8-2)$, and
      $\widehat hu=v\pmod q$.
\item for HRSS, $u\in\mathcal T^+$ and
      $v=(x-1)v_0$ as an equality between the standard integer
      representatives of degree less than $N$, for some
      $v_0\in\mathcal T^+$,
      together with $\widehat hu=v\pmod q$.
\item in either profile, $h(1)=0\pmod q$, and the required
      inverses of $u$ in $S_N/3S_N$ and $S_N/qS_N$, and the public inverse
      of $h$ in $S_N/qS_N$, exist.
\end{itemize}
The condition $h(1)=0\pmod q$ is equivalent to $h$ lying in the image
of the canonical $R_{q,0}$ public key encoding.
Every test is polynomial-time and exact.

For the profile interpretation in the next lemma, inherit
$\operatorname{ord}_N(2)=\operatorname{ord}_N(3)=N-1$ and the respective
HPS weight or HRSS correctness conditions from
\cite[Sections~1.3.2--1.3.3]{NTRUSpec}, stated explicitly
in~\eqref{uf:ntru-arithmetic}. The geometric statements below use
only the hypotheses displayed in each result.

\begin{lemma}
\label{uf:ntru-short-check}
For an HPS witness,
\[
 \|(f,g)\|_2^2\le (N-1)+(q/8-2),
\]
and for an HRSS witness, writing $g=(x-1)g_0$,
\[
 \|(f,g)\|_2^2\le3(N-1).
\]
Thus, along families with $q=\Theta(N)$, a public constant
$C_{\rm len}$ gives $\|(f,g)\|_2\le C_{\rm len}\sqrt q$, and this
radius is below $q/2$ for all sufficiently large $N$.  Any pair accepted
by $\mathsf{Valid}_{\mathsf P}$ yields the tuple
\[
 \bigl(u,f_p,h_q\bigr)
 =\bigl(u,\operatorname{S3\_inverse}(u),
           \operatorname{Sq\_inverse}(h)\bigr),
\]
consisting of the secret polynomial and the inverses used by DPKE
decryption and by KEM decapsulation on valid inputs. We call this tuple
the \emph{short key core}. Sampling seeds and KEM fallback strings are
separate data.
\end{lemma}
\begin{proof}
For HPS, $f$ is ternary and $g$ has exactly $q/8-2$ nonzero
coefficients.  For HRSS, $f,g_0\in\mathcal T^+$ and
\[
 \|(x-1)g_0\|_2^2
 =2\|g_0\|_2^2-2\sum_i(g_0)_i(g_0)_{i+1}
 \le2(N-1).
\]
For the output statement put $G=3v$.  In $S_N/qS_N$, the checked
congruence gives $hu=G$, and hence
$(Gu)^{-1}G^2=h$ and $(Gu)^{-1}u^2=h^{-1}$.  The checked
$R_{q,0}$ condition makes the first identity agree with the canonical
public representative, so the public-key computation reproduces the
same $h$.  The remaining claims follow from the exact membership and
unit tests in the specification~\cite[Sections~1.2, 1.8.3--1.8.4,
1.10.1, 1.11.1--1.11.2]{NTRUSpec}.
\end{proof}

\subsubsection{Saturation and Gaussian mass}

The short pair sublattice can have a nontrivial saturation index.
Fix an integer lift of $\widehat h$ and put
\begin{equation}
 k=\frac{\widehat hf-g}{q}\in R_N,
 \qquad K_N=\mathbb Q[x]/(x^N-1).
 \label{uf:ntru-k}
\end{equation}
Define the rational map $A(a)=a(f,g)$, its real image
\[
 W=A(K_N\otimes_{\mathbb Q}\mathbb R)\subset\mathbb R^{2N},
\]
and the fractional $R_N$-module
\[
 T=\{a\in K_N:af,ag,ak\in R_N\}.
\]
If the integer lift of $\widehat h$ is changed to $\widehat h+qb$ with
$b\in R_N$,
then $k$ changes to $k+bf$. Consequently both $T$ and the ideal
$fR_N+gR_N+kR_N$ are independent of the chosen lift.  Put
\[
 \Delta_{\rm sat}=[R_N:fR_N+gR_N+kR_N]\in\mathbb N\cup\{\infty\},
\]
with value $\infty$ when this ideal is not full rank.
For the $N$th roots of unity define the scalar spectral extrema
\begin{equation}
 G_{\min}=\min_{\zeta^N=1}(|f(\zeta)|^2+|g(\zeta)|^2),
 \qquad
 G_{\max}=\max_{\zeta^N=1}(|f(\zeta)|^2+|g(\zeta)|^2).
 \label{uf:ntru-spectrum}
\end{equation}

\begin{proposition}
\label{uf:ntru-cyclic-geometry}
If $G_{\min}>0$, then $R_N\subseteq T$ is a finite-index inclusion and
\begin{equation}
 \Lambda_{\widehat h}^{\rm cyc}\cap W=AT,
 \qquad
 \Delta_{\rm sat}=[T:R_N]=[R_N:fR_N+gR_N+kR_N].
 \label{uf:ntru-saturation}
\end{equation}
Given $f,g,k$, the integer $\Delta_{\rm sat}$ is computable from the rectangular
Smith form of the stacked multiplication matrices of $f,g,k$.
For every $s>0$,
\begin{equation}
 1\le\rho_s(\Lambda_{\widehat h}^{\rm cyc})
 \le\Delta_{\rm sat}\,
 \theta\!\left(\frac{s}{\sqrt{G_{\min}}}\right)^N
 \theta\!\left(\frac{s\sqrt{G_{\max}}}{q}\right)^N.
 \label{uf:ntru-master-mass}
\end{equation}
\end{proposition}
\begin{proof}
The spectral assumption makes $A$ injective.  If $A(a)=(af,ag)$ is an
integer point, then
$\widehat h(af)-ag=qak$, so it belongs to
$\Lambda_{\widehat h}^{\rm cyc}$ exactly when $a\in T$.  The unique
preimage of an integer point lies in $K_N$ because $A$ is rational and
full rank on its image, proving the first equality in
\eqref{uf:ntru-saturation}.

The full cyclic coefficient lattice is self-dual and multiplication
matrices satisfy $\mathbf M_a^T=\mathbf M_{\bar a}$. Therefore
\[
 T^*=\bar fR_N+\bar gR_N+\bar kR_N.
\]
Coefficient conjugation is an integral orthogonal permutation, which
proves the index identity and the asserted finite-index inclusion.
Equivalently, $\Delta_{\rm sat}$ is the gcd of the full-rank minors of the
rectangular matrix $[\mathbf M_f\ \mathbf M_g\ \mathbf M_k]$. This computation uses the witness
$(f,g,k)$. Recovery uses the public upper bound derived below.

Cyclic conjugation is orthogonal and satisfies $\mathbf M_a^T=\mathbf M_{\bar a}$,
so the argument of Proposition~\ref{uf:source-geometry} gives
$q(\Lambda_{\widehat h}^{\rm cyc})^*=\mathbf J\Lambda_{\widehat h}^{\rm cyc}$
and $\mathbf JW=W^\perp$ for
$\mathbf J\operatorname{coeff}(u,v)=\operatorname{coeff}(\bar v,-\bar u)$.
Thus, with $M=AT$ and $Q=\pi_{W^\perp}\Lambda_{\widehat h}^{\rm cyc}$,
we have $Q\simeq qM^*$ and
\[
 \rho_s(\Lambda_{\widehat h}^{\rm cyc})
 \le\rho_s(M)\rho_s(qM^*).
\]
The first factor is at most $\Delta_{\rm sat}\rho_s(AR_N)$, whereas
$M^*\subseteq(AR_N)^*$ bounds the second by $\rho_s(q(AR_N)^*)$.
This gives a single factor $\Delta_{\rm sat}$. The unitary
full cyclic Fourier transform bounds the two remaining masses by the
theta factors in~\eqref{uf:ntru-master-mass}.
\end{proof}

The mass estimate leads to three useful forms of the recovery
condition. The first bounds the total logarithmic mass directly.
The second gives sufficient bounds on the two spectral extrema and
the saturation index. The third uses the bounds established by the
sampling analysis below and supplies the budget for the general
generator theorem.

\begin{corollary}[Recovery from a total mass bound]
\label{uf:ntru-weak-mass-budget}
Use the public finite schedule $D_N\sim\ln\ln N$,
$(1-2^{-N^2})q/N^\beta\le s_{\rm mass}^2\le q/N^\beta$, and
$s_{\rm search}^2=qD_N/N$, for fixed rational $1/2<\beta<1$.
Suppose $G_{\min}>0$ and a fixed public rational $C_{\rm mass}\ge0$
satisfies
\begin{equation}
 \ln\Delta_{\rm sat}+
 N\ln\theta\!\left(\frac{s_{\rm mass}}{\sqrt{G_{\min}}}\right)+
 N\ln\theta\!\left(\frac{s_{\rm mass}\sqrt{G_{\max}}}{q}\right)
 \le C_{\rm mass}N/D_N.
 \label{uf:ntru-total-theta-budget}
\end{equation}
Then the public integer
$A_{\rm mass}=4^{\lceil(C_{\rm mass}+1)N/D_N\rceil}$ bounds
$\rho_{s_{\rm mass}}(\Lambda_{\widehat h}^{\rm cyc})$ and has bit length
$O(N/D_N)$. With the modulus window, witness, and radius conditions
of Theorem~\ref{uq:mass-search-framework}, this supplies recovery
with internal success probability $1-N^{-K}$ and time and
space $2^{O(N/\ln\ln N)}$.
\end{corollary}
\begin{proof}
Take logarithms in~\eqref{uf:ntru-master-mass}. The displayed budget
bounds the resulting expression, and $\ln4>1$ makes the public
integer an upper bound. Its logarithm is $O(N/D_N)$ with a constant
fixed before the instance is drawn. Moreover,
$s_{\rm mass}^2/s_{\rm search}^2\ge(1-2^{-N^2})N^{1-\beta}/D_N\to\infty$.
Theorem~\ref{uq:mass-search-framework} applies for large $N$.
The algorithm uses this public value of $A_{\rm mass}$.
\end{proof}

\begin{corollary}[Recovery from spectral and saturation bounds]
\label{uf:ntru-weak-spectral-budget}
In the schedule of Corollary~\ref{uf:ntru-weak-mass-budget}, the conditions
\begin{equation}
 G_{\min}\ge\frac q{N^\beta}\ln D_N,\qquad
 G_{\max}\le\frac{qN^\beta}{\ln D_N},\qquad
 \ln\Delta_{\rm sat}\le C_{\rm sat} N/D_N
 \label{uf:ntru-weak-extremes}
\end{equation}
with fixed public rational $C_{\rm sat}\ge0$ suffice, for large $N$,
with $C_{\rm mass}=C_{\rm sat}+1$ and
\begin{equation}
 A_{\rm mass}=4^{\lceil(C_{\rm sat}+2)N/D_N\rceil},\qquad
 \ln A_{\rm mass}\le2(C_{\rm sat}+3)N/D_N.
 \label{uf:ntru-weak-public-budget}
\end{equation}
\end{corollary}
\begin{proof}
Both squared theta arguments in~\eqref{uf:ntru-master-mass} are at
most $1/\ln D_N$. Eventually these arguments are at most one, and
$\ln\theta(t)\le3e^{-\pi/t^2}$ gives
\[
 \ln\rho_{s_{\rm mass}}(\Lambda_{\widehat h}^{\rm cyc})
 \le C_{\rm sat} N/D_N+6ND_N^{-\pi}
 \le(C_{\rm sat}+1)N/D_N.
\]
Corollary~\ref{uf:ntru-weak-mass-budget} gives the budget and recovery claim.
\end{proof}

Section~\ref{uq:ntru-distributional} uses
Corollary~\ref{uf:ntru-weak-spectral-budget} to obtain its public budget.
The next corollary applies the sampling estimates directly and supplies
the budget for Theorem~\ref{uf:ntru-generator-theorem}.

\begin{corollary}[Mass bound for the sampling estimates]
\label{uf:ntru-mass-budget}
Fix computable rational constants
$C_{\rm sp}>0$, $C_{\rm sat}\ge0$, and $2/3<\alpha<\beta<1$.
Let $D=D(N)$ be a public positive dyadic rational of polynomial bit
length with $D\le1+\ln\ln N$.  Put
$s_{\rm ref}^2=q/N^\beta$, and let $s_{\rm mass}^2$ be a positive dyadic rational
satisfying
\[
 (1-2^{-N^2})s_{\rm ref}^2\le s_{\rm mass}^2\le s_{\rm ref}^2,
\]
and suppose
\[
 G_{\min}\ge q/N^\alpha,\qquad
 G_{\max}\le C_{\rm sp}q\ln N,
 \qquad \ln\Delta_{\rm sat}\le C_{\rm sat} N/D.
\]
Then
\[
 \ln\rho_{s_{\rm ref}}(\Lambda_{\widehat h}^{\rm cyc})
 \le C_{\rm sat} N/D+o(1).
\]
Put $s_{\rm search}^2=qD/N$.  For all sufficiently large $N$,
$s_{\rm mass}^2\ge8s_{\rm search}^2$ and
\begin{equation}
 \rho_{s_{\rm search}}(\Lambda_{\widehat h}^{\rm cyc})
 \le \rho_{s_{\rm mass}}(\Lambda_{\widehat h}^{\rm cyc})
 \le \rho_{s_{\rm ref}}(\Lambda_{\widehat h}^{\rm cyc})
 \le \overline A_N,
 \qquad
 \overline A_N:=4^{\left\lceil(C_{\rm sat}+1)N/D\right\rceil},
 \label{uf:ntru-public-mass-bound}
\end{equation}
where $\overline A_N$ is an explicitly computable integer and
$\ln\overline A_N=O(N/D)$.
\end{corollary}
\begin{proof}
Evaluate the upper bound in~\eqref{uf:ntru-master-mass} at
$s=s_{\rm ref}$.  The squared arguments
of the two theta factors are then at most
$N^{\alpha-\beta}$ and $C_{\rm sp}(\ln N)/N^\beta$.
The standard bound
$\theta(x)^N\le\exp(3Ne^{-\pi/x^2})$ for $0<x\le1$
makes both contributions $1+N^{-\omega(1)}$.  The final statements
follow from monotonicity, the dyadic lower approximation, and
$s_{\rm mass}^2/s_{\rm search}^2\ge(1-2^{-N^2})N^{1-\beta}/D\to\infty$.
\end{proof}

\subsubsection{Recovery from geometric bounds}

The following theorem turns joint spectral and saturation bounds
into recovery for a key generator. The next parts verify these
bounds for the HPS and HRSS distributions and give explicit failure
probabilities.
Take a sequence of primes $N$ and power-of-two $q$ satisfying the profile's
arithmetic constraints and the modulus window below.

\begin{definition}[Conditions for recovery]
\label{uf:ntru-regular-generator}
Fix a profile $\mathsf P\in\{\mathsf{HPS},\mathsf{HRSS}\}$ and public
computable rational constants
\[
 2/3<\alpha<\beta<1,\qquad
 C_{\rm len},C_{\rm sp},C_{\rm sat},A_q>0,\qquad 0<c<1,
\]
and use the public positive dyadic schedule $D_N$ of polynomial bit length
satisfying
\[
 0\le\ln\ln N-D_N\le2^{-N^2}.
\]
A generator $\mathsf{CycNTRUGen}_{\mathsf P,N}$ is a distribution on
$(h,f,g,\mathsf{aux})$. It satisfies the recovery conditions when
every output satisfies
\[
 \mathsf{Valid}_{\mathsf P}(h;f,g)=1,\qquad
 \|(f,g)\|_2\le C_{\rm len}\sqrt q,
\]
and the single joint event
\begin{equation}
 \mathcal E_N=\left\{
 G_{\min}\ge q/N^\alpha,\quad
 G_{\max}\le C_{\rm sp}q\ln N,\quad
 \ln\Delta_{\rm sat}\le C_{\rm sat} N/D_N
 \right\}
 \label{uf:ntru-joint-event}
\end{equation}
has probability at least $1-\omega_N$, where $\omega_N=o(1)$ and
$\Delta_{\rm sat}$ is defined from the generated witness by
\eqref{uf:ntru-k} and the extended-index convention preceding
Proposition~\ref{uf:ntru-cyclic-geometry}.  The constants and the bound
$C_{\rm sat}$ are fixed before the generator is sampled.
\end{definition}

\begin{theorem}[Cyclic NTRU recovery]
\label{uf:ntru-generator-theorem}
Suppose
\[
 q=2^{b_q},\qquad 2N+1\le q\le N^{A_q},\qquad
 \ln(q/N)\le(\ln N)^{1-c},
\]
and let $\mathsf{CycNTRUGen}_{\mathsf P,N}$ satisfy
Definition~\ref{uf:ntru-regular-generator}.  For every fixed $K>0$ and
all sufficiently large $N$, there is one randomized algorithm,
depending only on $K$, the public profile, parameter sequence, and
specified constants, such that on public input $(\mathsf P,N,q,h)$ it
returns a valid short pair and the associated decryption tuple of
Lemma~\ref{uf:ntru-short-check}, with probability at least
\begin{equation}
 1-\omega_N-N^{-K}.
 \label{uf:ntru-generator-success}
\end{equation}
Its time and space are
\begin{equation}
 \exp\!\left(O\!\left(\frac{N}{\ln\ln N}\right)\right).
 \label{uf:ntru-generator-cost}
\end{equation}
The hidden constants depend only on $K$ and the specified constants.
More precisely, for every fixed generator output that satisfies
$\mathsf{Valid}_{\mathsf P}(h;f,g)=1$, the displayed norm bound, and
$\mathcal E_N$, the same algorithm has conditional success probability
at least $1-N^{-K}$ over independent internal coins.
\end{theorem}
\begin{proof}
On the joint event $\mathcal E_N$, Corollary~\ref{uf:ntru-mass-budget}
supplies the public rational bound
$A_{\rm mass}=\overline A_N$ in~\eqref{uf:ntru-public-mass-bound}, with
$\ln A_{\rm mass}\le2(C_{\rm sat}+2)N/D_N$ for large $N$.
Apply Theorem~\ref{uq:mass-search-framework} to the multiplication
matrix of $\widehat h$, with $m=N$, $c_{\rm search}=1$, and
$B_{\rm Q}=C_{\rm len}\sqrt q$.
Its width condition follows from $s_{\rm mass}^2/s_{\rm search}^2\to\infty$.
Lemma~\ref{uf:ntru-short-check} supplies the polynomial-time verifier
$\mathsf{Valid}_{\mathsf P}$ and the decryption inverses.
The given valid pair is an accepted witness, so the pointwise success
probability is at least $1-N^{-K}$, with the stated cost.
Averaging over the generator and charging
$\Pr[\mathcal E_N^c]\le\omega_N$ proves
\eqref{uf:ntru-generator-success}.
\end{proof}

\subsubsection{Arithmetic parameters and accepted sampling laws}
\label{uf:ntru-sampling}

For the distributional applications, take a supplied growing sequence
of primes satisfying
\begin{equation}
 N\ge5,\qquad \operatorname{ord}_N(2)=\operatorname{ord}_N(3)=N-1,
 \qquad q=2^{b_q}=\Theta(N),\qquad q\ge2N+1.
 \label{uf:ntru-arithmetic}
\end{equation}
These are the common arithmetic requirements of the two profiles
in~\cite[Sections~1.3.2--1.3.3]{NTRUSpec}. For HPS put
$w=q/8-2$, $m=w/2=q/16-1$, and require $w\le2N/3$.
For HRSS retain the correctness condition
$q>8\sqrt2(N-1)$; the specification's choice
$q=2^{\lceil7/2+\log_2N\rceil}$ satisfies it.
For each sufficiently large supplied HPS degree,
$q=2^{\lceil\log_2(2N+1)\rceil}$ meets its modulus and weight
conditions. The parameter sequence is part of the family definition.

Put $L=N-1$, $a=86/256$, and $b=85/256$. Let
$\mathbf X=(X_0,\ldots,X_{L-1})$ have independent coordinates with
\begin{equation}
 \Pr[X_j=0]=a,\qquad \Pr[X_j=1]=\Pr[X_j=-1]=b,
 \qquad P_{\mathbf X}(x)=\sum_{j=0}^{L-1}X_jx^j.
 \label{uf:ntru-raw-coefficients}
\end{equation}
This is the ternary law produced by independent uniform input bytes.
Write $U_{\mathbf X}=\operatorname{TernaryPlus}(\mathbf X)$: if the
adjacent correlation is negative, reverse the signs of the even-indexed
coefficients, and otherwise keep the coefficients. In particular,
\begin{equation}
 U_{\mathbf X}(z)=P_{\mathbf X}(z)
 \quad\text{or}\quad U_{\mathbf X}(z)=-P_{\mathbf X}(-z).
 \label{uf:ntru-plus-branches}
\end{equation}
Zero input blocks are included at this stage.

\begin{definition}[HPS and HRSS key distributions]
\label{uf:ntru-proposal-laws}
The two raw profiles use independent blocks and are summarized as
follows:
\[
\begin{array}{c|c|c|c}
\text{profile} & f & g & \text{profile condition}\\ \hline
\mathsf{HPS} & P_{\mathbf X} & \text{uniform fixed type}
  & w=q/8-2\le2N/3\\
\mathsf{HRSS} & U_{\mathbf X} & (x-1)U_{\mathbf Y}
  & q>8\sqrt2(N-1)
\end{array}
\]
For HPS, the fixed type has $m=w/2$ coefficients of each sign and
$L-2m$ zeros. The optional sorting law attaches independent uniform
$b_N$-bit labels before sorting by label and category.

Let $\pi_N:R_N\to S_N$ be the cyclotomic projection and define
\begin{equation}
 \mathcal A_N=\{\pi_N(f)\bmod3\in(S_N/3S_N)^\times,
 \quad \pi_N(f),\pi_N(g)\bmod2\in(S_N/2S_N)^\times\}.
 \label{uf:ntru-acceptance-event}
\end{equation}
For each raw law, the accepted law is $\mu=\nu\mid\mathcal A_N$.
Its public key is the unique representative modulo $x^N-1,q$ with
$h(1)=0$ and cyclotomic component $3\pi_N(g)\pi_N(f)^{-1}$.
\end{definition}

\begin{lemma}
\label{uf:ntru-algebraic-acceptance}
Under~\eqref{uf:ntru-arithmetic}, the acceptance probabilities are
\begin{equation}
 a_N^{\rm HPS}=1-a^{N-1},\qquad
 a_N^{\rm HRSS}=(1-a^{N-1})^2.
 \label{uf:ntru-acceptance-probabilities}
\end{equation}
The HPS formula holds for both the ideal law and every sorting label
length. Every accepted pair passes the corresponding profile verifier.
For every event $E$ defined on the raw polynomials,
\begin{equation}
 \Pr_{\mu}[E^c]\le\Pr_{\nu}[E^c]/a_N^{\mathsf P}.
 \label{uf:ntru-conditioning}
\end{equation}
\end{lemma}
\begin{proof}
For $p=2,3$, the order condition makes $\Phi_N$ irreducible of degree
$N-1$ over $\mathbb F_p$. A ternary polynomial of degree at most
$N-2$ is zero modulo $p$ exactly when all its coefficients are zero;
otherwise its cyclotomic image is a unit. TernaryPlus preserves the
zero-block event. Also $x-1$ is a unit modulo $\Phi_N,2$, since
$\Phi_N(1)=N$ is odd. Thus HPS acceptance is precisely
$\mathbf X\ne0$, since $g$ has positive fixed weight, and HRSS
acceptance is precisely $\mathbf X,\mathbf Y\ne0$. Independence gives
\eqref{uf:ntru-acceptance-probabilities}.

A cyclotomic element is a unit modulo $2^{b_q}$ exactly when it is a
unit modulo two. Since $N$ is odd, $x-1$ and $\Phi_N$ are coprime
modulo $q$, so the Chinese remainder theorem supplies the asserted
public representative. As $g(1)=0$, it satisfies $hf=3g$ in the full
cyclic ring. The sampler profiles and the accepted units are exactly
the checks in Lemma~\ref{uf:ntru-short-check}. Finally, conditional
probability gives~\eqref{uf:ntru-conditioning}.
\end{proof}

\begin{lemma}
\label{uf:ntru-fixed-type-tv}
For $L=N-1$ and $m=q/16-1$, the raw HPS laws satisfy
\begin{equation}
 \operatorname{TV}(\nu_{N,b_N}^{\rm HPS},\nu_N^{\rm HPS})
 \le\varepsilon_{\rm FT}(N):=
 \min\left\{1,\frac{m^2+2m(L-2m)}{2^{b_N}}\right\}.
 \label{uf:ntru-fixed-type-transfer}
\end{equation}
In particular, $b_N\ge(2+\delta)\log_2N$, with fixed $\delta>0$,
gives $\varepsilon_{\rm FT}=o(1)$, and
$b_N=\lceil(K_{\rm tag}+3)\log_2N\rceil$ gives
$\varepsilon_{\rm FT}=O(N^{-K_{\rm tag}-1})$.
\end{lemma}
\begin{proof}
Couple the specified ordering to one using an independent uniform
ordering within each equal-label group. Exchangeability makes the
latter ordering a uniform permutation, whose category word has the
ideal fixed-type law. The category words can differ only if two
entries of different categories share a label. There are $m^2$ pairs
between the nonzero categories and $2m(L-2m)$ pairs involving zero;
each collision has probability $2^{-b_N}$. The union bound proves
the estimate and its two asymptotic specializations.
\end{proof}

For the specification's fixed 30-bit labels, the same formula is a
finite distributional error bound. For the asymptotic sorting result
we use growing $b_N=O(\log N)$, so both the random-bit cost and sorting
cost remain polynomial.

\subsubsection{Logarithmic saturation from common cyclotomic ideals}
\label{uf:ntru-distributional-saturation}

The lemmas in this subsection hold for general odd prime $N$.
Put $\zeta_N=e^{2\pi i/N}$,
$E_N=\mathbb Q(\zeta_N)$, and $O_N=\mathbb Z[\zeta_N]\simeq S_N$;
write $\pi_N(f)=f(\zeta_N)$. This notation distinguishes the
cyclotomic projection from coefficient conjugation.
We use the standard facts that $O_N$ is the full ring of integers,
its nonzero ideals have unique prime ideal factorization, the only
ramified rational prime is $N$, and a prime above $p\ne N$ has
residue degree $\operatorname{ord}_N(p)$~\cite{Was97}.

\begin{lemma}
\label{uf:ntru-index-bridge}
Let $f,g\in R_N$, with $g(1)=0$, $f(1)\ne0$, and $\pi_N(g)\ne0$.
For any public relation defining $k=(\widehat hf-g)/q\in R_N$, put
$I=(\pi_N(f),\pi_N(g))\subseteq O_N$. Then
\begin{equation}
 \Delta_{\rm sat}=[R_N:(f,g,k)]\le N|f(1)|\operatorname{Norm}(I).
 \label{uf:ntru-index-bound}
\end{equation}
\end{lemma}
\begin{proof}
Embed $R_N$ in $A=\mathbb Z\oplus O_N$ by evaluation at $1,\zeta_N$.
The resultant of $x-1$ and $\Phi_N$ is $\Phi_N(1)=N$, so
$[A:R_N]=N$. For $J=fR_N+gR_N$ its extension is
$JA=f(1)\mathbb Z\oplus I$, and hence
$[A:JA]=|f(1)|\operatorname{Norm}(I)$.
The well-defined map of finite abelian groups
\[
 (A/R_N)^2\longrightarrow JA/J,\qquad
 ([a_1],[a_2])\longmapsto[fa_1+ga_2]
\]
is surjective. Thus $[JA:J]\le N^2$. The nonzero hypotheses make
$JA$, and therefore $J$, full rank. Multiplying indices gives
\[
 [R_N:J]=\frac{[A:JA][JA:J]}{[A:R_N]}
 \le N|f(1)|\operatorname{Norm}(I).
\]
Since $(f,g,k)\supseteq J$, its index can only be smaller.
\end{proof}

\begin{lemma}
\label{uf:ntru-modular-anticoncentration}
Let $X_1,\ldots,X_s$ be independent with
law~\eqref{uf:ntru-raw-coefficients}. For every prime $p$, all
$c_i\in\mathbb F_p^\times$, and every $t\in\mathbb F_p$,
\begin{equation}
 \Pr\left[\sum_{i=1}^s c_iX_i=t\right]
 \le\frac1p+\frac2{\sqrt s}.
 \label{uf:ntru-modular-small-ball}
\end{equation}
\end{lemma}
\begin{proof}
Let $\psi(u)=a+2b\cos(2\pi u)$. The independent-difference
characteristic function gives
\[
 1-|\psi(u)|^2\ge4ab(1-\cos(2\pi u))
 \ge32ab\|u\|_{\mathbb R/\mathbb Z}^2,
\]
and therefore $|\psi(u)|\le e^{-c\|u\|^2}$ with $c=16ab>1$.
Fourier inversion and H\"older's inequality for uniform measure on
$\mathbb F_p$ yield
\begin{align*}
 \Pr\left[\sum c_iX_i=t\right]
 &\le\frac1p\sum_v\prod_i|\psi(vc_i/p)|\\
 &\le\prod_i\left(\frac1p\sum_v|\psi(vc_i/p)|^s\right)^{1/s}
 =\frac1p\sum_v|\psi(v/p)|^s.
\end{align*}
Each nonzero $c_i$ permutes $\mathbb F_p$. Centering the last sum
and bounding its nonzero terms by a Gaussian integral gives
$p^{-1}+\sqrt{\pi/(cs)}\le p^{-1}+2/\sqrt s$.
\end{proof}

\begin{lemma}
\label{uf:ntru-prime-ideal-membership}
Let $F=P_{\mathbf X}$ or $F=U_{\mathbf X}$. For all sufficiently
large odd primes $N$ and every nonzero prime ideal
$\mathfrak p\subset O_N$,
\begin{equation}
 \Pr[\pi_N(F)\in\mathfrak p]
 \le\epsilon_N:=32\exp\!\left(-\tfrac14\sqrt{\ln N}\right).
 \label{uf:ntru-prime-ideal-bound}
\end{equation}
\end{lemma}
\begin{proof}
Suppose first that $\mathfrak p$ lies above $p\ne N$, has residue
degree $d$, and write $\xi=\zeta_N\bmod\mathfrak p$. Its order is
$N$, so $p^d\ge N+1$. The powers $1,\xi,\ldots,\xi^{d-1}$ form
an $\mathbb F_p$-basis. Condition on all other coefficients of
$P_{\mathbf X}$. The equation $P_{\mathbf X}(\xi)=0$ prescribes
the residues of these $d$ coefficients uniquely. Each residue has
probability at most $\eta=170/256$, including $p=2$, where the two
nonzero signs merge. Thus the probability is at most $\eta^d$.

For a second bound, fix any nonzero linear functional
$\lambda:\mathbb F_{p^d}\to\mathbb F_p$. Each block of $d$
consecutive powers of $\xi$ is a basis, so at least
$s=\lfloor(N-1)/d\rfloor$ of the weights $\lambda(\xi^j)$
are nonzero. Apply Lemma~\ref{uf:ntru-modular-anticoncentration}
to these coefficients, conditioning on the rest. Consequently
\[
 \Pr[P_{\mathbf X}(\xi)=0]
 \le\min\{\eta^d,\ p^{-1}+2/\sqrt{\lfloor(N-1)/d\rfloor}\}.
\]
For $d\ge\sqrt{\ln N}$ the first term is at most
$e^{-\sqrt{\ln N}/4}$, since $-\ln\eta>1/4$.
For $d<\sqrt{\ln N}$, the second is at most
$e^{-\sqrt{\ln N}}+4(\ln N)^{1/4}/\sqrt N$ for large $N$.
These explicit bounds give~\eqref{uf:ntru-prime-ideal-bound}.

The sole prime ideal above $N$ is $(1-\zeta_N)$, with residue
field $\mathbb F_N$. The membership condition becomes
$P_{\mathbf X}(1)=0\bmod N$, and the same scalar lemma gives
$N^{-1}+2/\sqrt{N-1}$.

For TernaryPlus, membership implies that either
$P_{\mathbf X}(\xi)=0$ or $P_{\mathbf X}(-\xi)=0$.
The element $-\xi$ generates the same residue field and its
consecutive powers are again bases. The preceding two bounds apply
unchanged; at the ramified prime use weights $1$ or alternating
signs. A union bound costs a factor of two, already allowed by
the stated constant. All comparisons above have effective thresholds.
For example, the displayed estimates apply once $\ln N\ge16$.
\end{proof}

\begin{lemma}
\label{uf:ntru-log-ideal-norm}
Let $f$ have either ternary law of the preceding lemma and let $g$
be independent of $f$. For each fixed $g$ with $\pi_N(g)\ne0$,
\begin{equation}
 \mathbb E_f[\ln\operatorname{Norm}(\pi_N(f),\pi_N(g))]
 \le\epsilon_N\ln|\operatorname{Norm}_{E_N/\mathbb Q}(\pi_N(g))|.
 \label{uf:ntru-ideal-expectation}
\end{equation}
If the cyclic coefficients of $g$ have absolute value at most two,
the right side is at most $\epsilon_N N\ln(3N)$.
\end{lemma}
\begin{proof}
For a prime ideal dividing $(\pi_N(g))$, put
$e_{\mathfrak p}=v_{\mathfrak p}(\pi_N(g))$. Even when $f=0$,
\[
 v_{\mathfrak p}(\pi_N(f),\pi_N(g))
 \le e_{\mathfrak p}\mathbf1_{\{\pi_N(f)\in\mathfrak p\}}.
\]
Take expectations in the finite sum over these prime ideals, weighted
by $\ln\operatorname{Norm}(\mathfrak p)$, and apply
Lemma~\ref{uf:ntru-prime-ideal-membership}. This proves
\eqref{uf:ntru-ideal-expectation}, including all prime ideal powers.
Finally, full cyclic Parseval and the arithmetic--geometric mean give
\[
 |\operatorname{Norm}_{E_N/\mathbb Q}(\pi_N(g))|
 =\prod_{j=1}^{N-1}|g(\zeta_N^j)|
 \le\left(\frac{N\|g\|_2^2}{N-1}\right)^{(N-1)/2}.
\]
Using $\|g\|_2^2\le4N$ bounds its logarithm by $N\ln(3N)$.
\end{proof}

Define the raw-polynomial proxy
\begin{equation}
 \Xi_N(f,g)=
 \begin{cases}
 N|f(1)|\operatorname{Norm}(\pi_N(f),\pi_N(g)),
       &f(1)\ne0,\ \pi_N(g)\ne0,\\
 \infty,&\text{otherwise}.
 \end{cases}
 \label{uf:ntru-saturation-proxy}
\end{equation}
It is defined before acceptance or public-key construction. On accepted
inputs with finite proxy, Lemma~\ref{uf:ntru-index-bridge} gives
$\Delta_{\rm sat}\le\Xi_N$.

\begin{proposition}[Saturation bound]
\label{uf:ntru-log-saturation}
Use any of the raw laws in Definition~\ref{uf:ntru-proposal-laws},
and the public dyadic $D_N\sim\ln\ln N$. For sufficiently large $N$,
\begin{align}
 \Pr[\ln\Xi_N>N/D_N]&\le r_{\rm sat}(N),\nonumber\\
 r_{\rm sat}(N)&:=64D_N\ln(3N)e^{-\sqrt{\ln N}/4}
                         +12N^{-1/2}=o(1).
 \label{uf:ntru-saturation-tail}
\end{align}
In particular the fixed saturation constant is $C_{\rm sat}=1$.
\end{proposition}
\begin{proof}
At $p=N$, Lemma~\ref{uf:ntru-modular-anticoncentration} bounds
$\Pr[P_{\mathbf X}(1)=0]$ by $N^{-1}+2/\sqrt{N-1}$.
For TernaryPlus use the two weights $1$ and alternating signs and
\eqref{uf:ntru-plus-branches}, giving at most twice this bound.
In HPS, $\pi_N(g)\ne0$ follows from its positive weight and degree
less than $N-1$. In HRSS it vanishes exactly when $\mathbf Y=0$,
an event of probability $a^{N-1}$. These exceptional probabilities
sum to at most $12/\sqrt N$ for large $N$.

Otherwise $|f(1)|\le N-1$, so
$\ln\Xi_N\le2\ln N+\ln\operatorname{Norm}(I)$.
For large $N$, $N/D_N-2\ln N\ge N/(2D_N)$.
For each nonexceptional $g$, Markov's inequality and
Lemma~\ref{uf:ntru-log-ideal-norm} bound the remaining probability
by $2\epsilon_ND_N\ln(3N)$. Average over $g$ and add the
exceptional probabilities. This gives the displayed formula.
The proof uses only independence, bounded coefficients, and
nonzero fixed weight for HPS $g$, so it also covers its sorting law.
\end{proof}

\subsubsection{Spectra of the raw samplers}

The following bounds hold before key generation rejection, uniformly
over the stated frequencies and small-ball centres.

\begin{lemma}
\label{uf:ntru-independent-small-ball}
Let $N\ge3$ be odd, $\zeta=e^{2\pi i/N}$, and
$P_{\mathbf X}(z)=\sum_{k=0}^{N-2}X_kz^k$, where
$\mathbf X=(X_0,\ldots,X_{N-2})$ has independent coordinates with
probabilities $86/256,85/256,85/256$ on $0,1,-1$.
There are absolute constants $C,c>0$ such that, for every $r\ge0$,
\begin{align}
 \sup_{w\in\mathbb C}
 \Pr[|P_{\mathbf X}(z)-w|\le r]
 &\le \frac{C(r+1)^2}{N},
 &&z\in\{\zeta^j,-\zeta^j:1\le j<N\},
 \label{uf:ntru-iid-complex-small-ball}\\
 \sup_{w\in\mathbb R}
 \Pr[|P_{\mathbf X}(z)-w|\le r]
 &\le \frac{C(r+1)}{\sqrt N},
 &&z\in\{1,-1\}.
 \label{uf:ntru-iid-real-small-ball}
\end{align}
Moreover, for every $|z|=1$ and $t\ge0$,
\[
 \Pr[|P_{\mathbf X}(z)|>t]\le4e^{-ct^2/N}.
\]
\end{lemma}
\begin{proof}
We first record the elementary Fourier smoothing inequality used here.
For an $s$-dimensional random vector $\mathbf Z$, $s\in\{1,2\}$, write
$\phi_{\mathbf Z}(\mathbf u)=\mathbb E[e^{i\langle\mathbf u,\mathbf Z\rangle}]$
and $R=r+1$.
For the fixed choice $c_0=1/16$,
\begin{equation}
 \Pr[\|\mathbf Z-\mathbf a\|_\infty\le r]
 \le C_sR^s\int_{[-c_0/R,c_0/R]^s}
             |\phi_{\mathbf Z}(\mathbf u)|\,d\mathbf u.
 \label{uf:ntru-fourier-smoothing}
\end{equation}
Indeed, the nonnegative function
$\prod_{\ell=1}^s\operatorname{sinc}^2(c_0(x_\ell-a_\ell)/(2R))$,
where $\operatorname{sinc}(v)=\sin(v)/v$ and
$\operatorname{sinc}(0)=1$, is bounded below by a positive constant
on the indicated cube.  Its Fourier transform is supported on
$[-c_0/R,c_0/R]^s$ and has absolute value at most $C_sR^s$.
Fourier inversion and expectation prove
\eqref{uf:ntru-fourier-smoothing}.

For $z$ in the first line of the lemma, put
$\mathbf v_k=(\operatorname{Re}z^k,\operatorname{Im}z^k)^T$.
Since $z^2$ is a nontrivial $N$th root,
\[
 \sum_{k=0}^{N-1}\mathbf v_k\mathbf v_k^T=\frac N2\mathbf I_2,
 \qquad
 \sum_{k=0}^{N-2}\mathbf v_k\mathbf v_k^T
 =\frac N2\mathbf I_2-\mathbf v_{N-1}\mathbf v_{N-1}^T
 \succeq\frac N6\mathbf I_2.
\]
This holds for both signs of $z$; it does not require $z^N=1$.
The common scalar coefficient characteristic function is
$\phi_{\rm coef}(u)=(86+170\cos u)/256$, and
$|\phi_{\rm coef}(u)|\le e^{-c u^2}$ for $|u|\le1$.
Consequently the characteristic function of
$\sum_{k=0}^{N-2}X_k\mathbf v_k$ has absolute value at most
$e^{-c'N\|\mathbf u\|_2^2}$ whenever $\|\mathbf u\|_2\le1$.
The integration cube in
\eqref{uf:ntru-fourier-smoothing} lies in this ball. The Gaussian
integral is $O(N^{-1})$, proving
\eqref{uf:ntru-iid-complex-small-ball}.
At $z=\pm1$ the same argument is one-dimensional with
$\sum_{k=0}^{N-2}(z^k)^2=N-1$, and the integral is $O(N^{-1/2})$.
Finally, the real and imaginary parts are sums of independent centred
bounded random variables; applying Hoeffding's inequality to both
parts proves the last assertion.
\end{proof}

\begin{lemma}
\label{uf:ntru-fixed-type-small-ball}
Fix $\eta>0$.  Let $N$ be an odd prime and let $m$ be an integer with
$\eta N\le m\le(N-1)/2$.  Let $Q$ be uniform among polynomials of
degree at most $N-2$ having $m$ coefficients of each sign and all other
coefficients zero.  There are constants $C,c>0$, depending only on
$\eta$, such that for every $1\le j<N$ and $r\ge0$,
\begin{equation}
 \sup_{w\in\mathbb C}\Pr[|Q(\zeta^j)-w|\le r]
 \le\frac{C(r+1)}{\sqrt N}+e^{-cN}.
 \label{uf:ntru-fixed-type-small-ball-bound}
\end{equation}
Also, $Q(1)=0$ and, uniformly for $|z|=1$ and $t\ge0$,
\[
 \Pr[|Q(z)|>t]\le4e^{-ct^2/N}.
\]
\end{lemma}
\begin{proof}
Put $d=N-1$.  Draw a uniform ordering $\pi(1),\ldots,\pi(d)$ of
$\{0,\ldots,N-2\}$, pair its first $2m$ entries, and assign one plus
and one minus to each pair using independent fair signs $\sigma_a$.
The resulting law is invariant under all permutations of the $d$
positions and has the required counts, so it is exactly uniform.
Conditional on $\pi$,
\[
 Q(z)=\sum_{a=1}^m\sigma_a
       (z^{\pi(2a-1)}-z^{\pi(2a)}).
\]
Fix $z=\zeta^j$.  Let $A$ and $B$ consist of positions whose powers
have real part at least $1/2$ and at most $-1/2$, respectively.
Because $N$ is prime, multiplication by $j$ permutes all $N$ residues.
After deleting one position, each set therefore has size at least
$N/3-2$, and hence at least $N/4$ for sufficiently large $N$.
If $L$ counts pairs joining $A$ and $B$, then
\[
 \mathbb E[L]=\frac{2m|A||B|}{d(d-1)}\ge c_1N,
\]
where $c_1>0$ depends only on $\eta$.
Transposing two entries of $\pi$ changes at most two pair indicators.
In the Doob martingale exposing the permutation one entry at a time,
two possible next entries can be coupled by transposing their labels
in the remaining completion.  Its increments thus have absolute value
at most $2$.  Azuma's inequality gives
\[
 \Pr[L<\mathbb E[L]/2]
 \le\exp\!\left(-\frac{(\mathbb E[L])^2}{32d}\right)
 \le e^{-c_2N}.
\]
On the complementary event, at least $c_1N/2$ of the real coefficients
$b_a=\operatorname{Re}(z^{\pi(2a-1)}-z^{\pi(2a)})$ have
$|b_a|\ge1$, while all satisfy $|b_a|\le2$.
For $|u|\le1/4$ the conditional characteristic function therefore obeys
\[
 \left|\mathbb E[e^{iu\operatorname{Re}Q(z)}\mid\pi]\right|
 =\prod_{a=1}^m|\cos(ub_a)|\le e^{-c_3Nu^2}.
\]
Apply the one-dimensional version of
\eqref{uf:ntru-fourier-smoothing} and use that a complex ball projects
into a real interval of the same radius.  Averaging over $\pi$ proves
\eqref{uf:ntru-fixed-type-small-ball-bound}; increasing $C$ covers the
finitely many smaller $N$.  For the upper tail, no good-pair event is
needed: conditional on every $\pi$, the summands are independent
centred signs with complex coefficients of modulus at most $2$.
Hoeffding's inequality for each component and $m\le N/2$ give the
stated uniform estimate.  The identity $Q(1)=0$ follows from balance.
\end{proof}

\begin{lemma}
\label{uf:ntru-ternary-plus-transfer}
For the input $\mathbf X$ of Lemma~\ref{uf:ntru-independent-small-ball}, put
$U=\operatorname{TernaryPlus}(\mathbf X)$. The same small-ball estimates hold
for $U$ with at most twice the constant, and
\[
 \Pr[|U(z)|>t]\le8e^{-ct^2/N}\qquad(|z|=1).
\]
Outputs obtained from independent input blocks remain independent.
\end{lemma}
\begin{proof}
The specification's sign adjustment gives, for every input,
$U(z)=P_{\mathbf X}(z)$ or $U(z)=-P_{\mathbf X}(-z)$.
Thus the event $|U(z)-w|\le r$ is contained in the union of
$|P_{\mathbf X}(z)-w|\le r$ and $|P_{\mathbf X}(-z)+w|\le r$.
Lemma~\ref{uf:ntru-independent-small-ball} covers both signs of every
nontrivial frequency, and also the two real frequencies $1,-1$.
The same union bound proves the upper tail.  This argument does not
condition on which sign adjustment was selected.  Each output is a
deterministic function of its own input block, which proves the final
assertion.
\end{proof}

\begin{proposition}[Spectral bound]
\label{uf:ntru-sampler-spectrum}
Suppose $q=\Theta(N)$, $q\ge2N+1$, and fix
$2/3<\alpha<1$ and a positive integer $K_{\rm sp}$.
For the raw HRSS and ideal HPS laws of
Definition~\ref{uf:ntru-proposal-laws},
\begin{equation}
 \Pr\left[
 G_{\min}<q/N^\alpha\ \text{or}\
 G_{\max}>C_{\rm sp}q\ln N
 \right]
 \le r_{\rm sp}(N):=
 C\left(N^{1-3\alpha/2}+N^{-\alpha/2}+N^{-K_{\rm sp}}\right)
 \label{uf:ntru-spectral-tail}
\end{equation}
where $C_{\rm sp}$ is a fixed public rational constant and $C$ is
effective, depending only on $K_{\rm sp}$ and the fixed modulus
window. One may take $C_{\rm sp}=32(K_{\rm sp}+3)$.
\end{proposition}
\begin{proof}
Put $T_0=q/N^\alpha$.  Since $q=\Theta(N)$ and $\alpha<1$, we may
assume $T_0\ge4$ throughout the asymptotic estimates.
For HPS, $m=q/16-1\ge N/16$ for sufficiently large $N$, while
$2m\le2N/3$ ensures the required support size.
Lemmas~\ref{uf:ntru-independent-small-ball}
and~\ref{uf:ntru-fixed-type-small-ball}, applied to the independent
polynomials $f,g$ with radius $\sqrt{T_0}$, give, for $1\le j<N$,
\[
 \Pr[|f(\zeta^j)|^2+|g(\zeta^j)|^2<T_0]
 \le C\frac{T_0}{N}
       \left(\sqrt{\frac{T_0}{N}}+e^{-cN}\right).
\]
Summing over the nontrivial frequencies gives
$O(T_0^{3/2}/\sqrt N+T_0e^{-cN})
 =O(N^{1-3\alpha/2})$.
At the trivial frequency, $g(1)=0$, and
\eqref{uf:ntru-iid-real-small-ball} bounds the lower-tail probability
by $O(\sqrt{T_0/N})=O(N^{-\alpha/2})$.

For HRSS, write $f=\operatorname{TernaryPlus}(\mathbf X)$ and
$g=(x-1)\operatorname{TernaryPlus}(\mathbf Y)$, with independent
coefficient vectors $\mathbf X,\mathbf Y$.
If $a_j=|\zeta^j-1|$, Lemma~\ref{uf:ntru-ternary-plus-transfer}
and $\sqrt{T_0}/a_j\ge1$ give
\[
 \Pr[|f(\zeta^j)|^2+|g(\zeta^j)|^2<T_0]
 \le C\frac{T_0}{N}
       \min\!\left\{1,\frac{T_0}{Na_j^2}\right\}.
\]
Writing $d_j=\min(j,N-j)$, we have $a_j\ge4d_j/N$ and hence
\[
 \sum_{j=1}^{N-1}\min\!\left\{1,\frac{T_0}{Na_j^2}\right\}
 \le 2\sum_{d=1}^{(N-1)/2}
       \min\!\left\{1,\frac{T_0N}{16d^2}\right\}
 =O(\sqrt{T_0N}).
\]
The last bound follows by splitting the sum at $\sqrt{T_0N}$ and
using $\sum_{d>L}d^{-2}\le1/L$.  The HRSS lower-tail sum is therefore
also $O(T_0^{3/2}/\sqrt N)=O(N^{1-3\alpha/2})$.
Here $g(1)=0$ again, and the real-frequency part of
Lemma~\ref{uf:ntru-ternary-plus-transfer} gives $O(N^{-\alpha/2})$.

For both profiles the upper-tail bounds in the lemmas, together with
$|\zeta^j-1|\le2$ for HRSS, imply
\[
 \Pr[G_{\max}>C_{\rm sp}q\ln N]
 \le16N\exp(-C_{\rm sp}q\ln N/(32N))\le N^{-K_{\rm sp}}
\]
for the stated $C_{\rm sp}$ and sufficiently large $N$.
Indeed the scalar Hoeffding estimates above allow the common upper
tail $8\exp(-t^2/(8N))$; for HRSS the independent-coefficient
TernaryPlus tail is $8\exp(-t^2/(4N))$, and the factor
$|\zeta^j-1|\le2$ gives the denominator $32N$ displayed here.
This proves
\eqref{uf:ntru-spectral-tail} without assuming independence between
different frequencies. All small-ball constants are obtained from
the fixed smoothing kernel, Gaussian integrals, and the displayed
permutation concentration bound, and thus have effective bounds.
\end{proof}

\subsubsection{Geometry of accepted keys}
\label{uf:ntru-accepted-bounds}

We now combine the spectral and saturation bounds for the raw
polynomials. Conditioning on the key generation acceptance event
gives the joint geometric event used in
Section~\ref{uq:ntru-distributional}. The raw failure bound is
$r_{\rm sp}(N)+r_{\rm sat}(N)$, and the acceptance probability
accounts for the denominator in the recovery guarantee.

For the HPS sorting sampler, the distribution comparison adds
$\varepsilon_{\rm FT}(N)$ to the same raw event before conditioning.

\begin{proposition}
\label{uf:ntru-accepted-geometry}
Fix rational $2/3<\alpha<\beta<1$ and use the finite schedule
$D_N\sim\ln\ln N$. Define on the raw polynomials
\begin{equation}
 \mathcal E_N^{\rm raw}=\left\{
 G_{min}\ge q/N^\alpha,\quad G_{max}\le C_{\rm sp}q\ln N,
 \quad\ln\Xi_N\le N/D_N\right\}.
 \label{uf:ntru-raw-good-event}
\end{equation}
For $\mathsf P\in\{\mathsf{HPS},\mathsf{HRSS}\}$ and its ideal
accepted law,
\begin{equation}
 \Pr_{\mu_N^{\mathsf P}}[(\mathcal E_N^{\rm raw})^c]
 \le\frac{r_{\rm sp}(N)+r_{\rm sat}(N)}{a_N^{\mathsf P}}=o(1).
 \label{uf:ntru-accepted-good-probability}
\end{equation}
For the accepted HPS sorting law, the same bound holds with
$r_{\rm sp}+r_{\rm sat}+\varepsilon_{\rm FT}$ in the numerator.
On $\mathcal E_N^{\rm raw}$, the actual saturation index satisfies
$\ln\Delta_{\rm sat}\le N/D_N$.
\end{proposition}
\begin{proof}
The raw failure probability is at most $r_{\rm sp}+r_{\rm sat}$ by
Propositions~\ref{uf:ntru-sampler-spectrum} and~\ref{uf:ntru-log-saturation}
and the union bound. Lemma~\ref{uf:ntru-algebraic-acceptance} transfers
this event to the accepted law. For sorting, first apply
Lemma~\ref{uf:ntru-fixed-type-tv} once to the complete raw event, then
condition using the same HPS acceptance probability. Finally,
Lemma~\ref{uf:ntru-index-bridge} bounds the actual index by $\Xi_N$.
\end{proof}

\section{Changing the Public Distribution}
\label{uk:section}

ML-KEM and ML-DSA generate their public matrices by expanding a seed.
In ML-KEM, key generation derives the public matrix from the seed
$\rho$~\cite[Section~5.1, Algorithm~13]{FIPS203}. ML-DSA constructs
its matrix $\mathbf A$ from $\rho$ through
\texttt{ExpandA}~\cite[Section~6.1, Algorithms~6 and~32]{FIPS204}.
Each scheme therefore specifies an induced distribution on its public
operator. Since the geometric bounds in this paper are stated for
specified operator distributions, we record conditions under which
a bound transfers to another public law.

\begin{proposition}
\label{ug:public-transfer}
\label{ua:public-transfer}
\label{un:geo:distribution}
Fix a public space $\Omega_n$, a schedule, and its geometric event
$\mathcal G_n$. If a reference law $\mu_n$ satisfies
$\mu_n(\mathcal G_n^c)\le\omega_n$, then any law $\nu_n$ obeys
\begin{align}
 \|\nu_n-\mu_n\|_{\rm TV}\le\varepsilon_n
 &\ \Longrightarrow\
 \nu_n(\mathcal G_n^c)\le\omega_n+\varepsilon_n,
 \label{ua:tv-transfer}\\
 \nu_n(x)\le K_n\mu_n(x)\quad(x\in\Omega_n)
 &\ \Longrightarrow\
 \nu_n(\mathcal G_n^c)\le K_n\omega_n.
 \label{ua:density-transfer}
\end{align}
If the new input law has energy failure at most $p_E$ and the
fixed-operator algorithm has internal failure at most $\xi_n$,
total recovery failure is at most
$p_E+\nu_n(\mathcal G_n^c)+\xi_n$, with unchanged resource bounds.
\end{proposition}
\begin{proof}
Apply the event bound for total variation or sum the density inequality
over $\mathcal G_n^c$, then add the energy and internal failures.
\end{proof}

\paragraph{Seed-expanded operators.}
\begin{remark}
\label{ug:seed-support}
\label{ua:seed-support}
If $\mu_n$ is uniform and $\nu_n$ expands an $L_n$-bit seed, then
\begin{equation}\label{ua:seed-tv}
 \|\nu_n-\mu_n\|_{\rm TV}\ge\max\{0,1-2^{L_n}/|\Omega_n|\},
 \qquad |\Omega_n|\max_x\nu_n(x)\ge|\Omega_n|/2^{L_n}.
\end{equation}
Here $|\Omega_n|$ is $q^{n^2}$ for scalar matrices, $q^{kn}$ for
rank-$k$ square module matrices, and $q^n$ for field multipliers.
A short seed alone gives neither transfer bound. One may instead
prove the geometry directly for the induced law, or use
pseudorandomness covering the recovery algorithm's running time and
success experiment.
\end{remark}

\section{Discussion of Fine-Grained Reductions}
\label{ux:section}

A recovery algorithm also constrains reductions into its input regime.
The question is quantitative: how much dimension growth is needed when
a reduction produces a low-energy target? We compare the recovery
exponent with the number of variables of a source $3$-SAT instance.
The comparison requires a matching recovery task and a bound on the
exceptional probability across all oracle queries.

\subsection{Recovery Cost and the Reduction Model}
\label{ux:regime}

Let $n$ count the secret coefficients of a linear recovery instance.
For fixed module rank, this is the rank times the ring degree.
Use the normalized candidate radius $U_n$ from~\eqref{ui:depth} and put
\begin{equation}\label{ux:depth}
 q=n^{\kappa+o(1)},\qquad \kappa>1/2,\qquad
 Q_n=\ln(q/\sqrt n),\qquad
 D_n=\ln\frac{Q_n}{1+\ln U_n}.
\end{equation}
In the recovery window $D_n\to\infty$, Theorem~\ref{ue:main} and the
finite implementation in Appendix~\ref{ue:finite} give the cost scale
$2^{O(n/D_n)}$. The timeout construction there supplies a computable
worst-case time bound for every public input, while retaining constant
success probability on each good instance. This suffices for composition.

Fix a target law $\mu_n$ specifying the public operator, labels,
observation count, energy bounds, and a uniquely specified correct
recovery output. Let $\mathcal G_n$ be the set on which both the
energy and geometric conditions hold, with
\[
 \mu_n(\mathcal G_n^c)\le\delta_n.
\]
On every instance in $\mathcal G_n$, the time-bounded recovery procedure
returns the specified output with probability at least $2/3$.
The integer secret task of the framework for linear relations supplies
this uniquely specified output.

Consider a classical randomized reduction from $3$-SAT on $N$ variables.
We use formulas without repeated clauses, whose length is polynomial
in $N$; arbitrary formulas can first be normalized in polynomial input
time. Assume the following conditions, uniformly over source formulas:
\begin{enumerate}
\item Apart from its oracle calls, the reduction takes polynomial time
and makes at most $a(N)=\operatorname{poly}(N)$ queries. Each query has
dimension $n(N)\to\infty$, with $n(N)\le\operatorname{poly}(N)$, and a
polynomial-length encoding.
\item With the specified correct oracle, the reduction decides the
source instance with probability at least $5/6$.
\item Conditional on any positive-probability query history in this
correct-oracle execution, the law of the next query has total variation
distance at most $\zeta_N$ from $\mu_{n(N)}$.
\item The accumulated exceptional probability satisfies
\begin{equation}\label{ux:ensemble-error}
 a(N)\bigl(\delta_{n(N)}+\zeta_N\bigr)=o(1).
\end{equation}
\end{enumerate}
The distribution-transfer principle in Appendix~\ref{uk:section}
then bounds the bad-instance probability of each query by
$\delta_{n(N)}+\zeta_N$. It is important to specify the full target
task: a short-error law alone does not fix the secret distribution,
number of observations, or geometric guarantee.

\subsection{Composition and Randomized ETH}
\label{ux:eth}

\begin{proposition}\label{ux:composition}
Under the reduction conditions above, $3$-SAT has a bounded-error
randomized algorithm with running time
\begin{equation}\label{ux:composed-time}
 2^{O(n(N)/D_{n(N)})}\operatorname{poly}(N).
\end{equation}
\end{proposition}
\begin{proof}
Answer a query by the most frequent output of
$c\ln((a(N)+2)(N+2))$ independent runs of the time-bounded recovery
procedure, with deterministic tie breaking. For a sufficiently large
constant $c$, a Chernoff bound makes its error on every good instance
at most
\[
 \eta_N=\frac1{100(a(N)+1)(N+2)^2}.
\]
The specified output has per-run probability at least $2/3$, so a
strict majority of correct outputs implies that this rule is correct.

Use the correct-oracle execution as a reference transcript, and attach
fresh independent recovery randomness to each query on that transcript.
Its conditional query-law assumption bounds the chance that a query
is bad by $\delta_{n(N)}+\zeta_N$. On a good query the internal error
is at most $\eta_N$. A union bound over the reference transcript gives
\[
 a(N)\bigl(\delta_{n(N)}+\zeta_N+\eta_N\bigr)=o(1).
\]
Couple the actual simulation to this transcript with the same reduction
randomness. The transcripts agree until the first incorrect simulated
answer, so the same bound controls the probability of any divergence.
Adding the correct-oracle reduction's error $1/6$ leaves total error
below $1/3$ for all sufficiently large $N$. The number of queries and
repetitions contributes only a polynomial factor to the time bound.
\end{proof}

We use randomized ETH, denoted $\mathsf{rETH}$, for the assertion that
$3$-SAT on $N$ variables has no bounded-error randomized algorithm in
time $2^{o(N)}\operatorname{poly}(|\phi|)$, where $|\phi|$ is the formula
length. This is the randomized counterpart of the exponential-time
hypothesis~\cite{IP01,IPZ01}. Randomization is essential to the formulation
because the composed recovery algorithm is randomized.

\begin{corollary}\label{ux:reth-exclusion}
Under $\mathsf{rETH}$, no reduction satisfying the conditions above has
\[
 n(N)/D_{n(N)}=o(N).
\]
\end{corollary}
\begin{proof}
Otherwise Proposition~\ref{ux:composition}, together with formula
normalization, would give the algorithm excluded by $\mathsf{rETH}$.
\end{proof}

Repeating recovery improves its internal success on a good instance.
It does not reduce the probability that a query itself violates the
energy or geometry assumptions. This is why the accumulated budget
\eqref{ux:ensemble-error} is separate from amplification.

\subsection{Dimension and Hidden Energy}
\label{ux:energy-tradeoff}

The simplest consequence concerns linear dimension growth. If
$n(N)=O(N)$ and $D_{n(N)}\to\infty$, then
\[
 \frac{n(N)}{D_{n(N)}}=o(N).
\]
Thus $\mathsf{rETH}$ excludes such a matching linear-dimension
reduction throughout the recovery window.

For polylogarithmic $U_n$ and polynomial modulus,
$D_n\sim\ln\ln n$. In this common window, the corollary also excludes
\begin{equation}\label{ux:polylog-exclusion}
 n(N)=o(N\ln\ln N).
\end{equation}
Indeed, if $n(N)\le\sqrt N$, then $n(N)/D_{n(N)}=o(N)$ directly.
For the remaining dimensions, $\sqrt N<n(N)\le\operatorname{poly}(N)$
gives $\ln\ln n(N)=\ln\ln N+O(1)$, and
\eqref{ux:polylog-exclusion} again implies the excluded exponent scale.
These conclusions exclude the stated asymptotic growth windows.

More generally, fixing the target dimension and modulus while increasing
$U_n$ decreases $D_n$ and raises the recovery cost. For example,
\[
 U_n=\exp(\ln n/L_n),\qquad 1\ll L_n\le\ln n
 \quad\Longrightarrow\quad
 D_n=\ln L_n+O(1).
\]
The composed exponent is then $O(n/\ln L_n)$. Lower target energy
therefore requires more dimension growth to avoid the excluded windows.
Dimension and energy must be compared together.

Known reductions provide background for this matching question, including
Regev's LWE reduction~\cite{Reg05}, small-parameter SIS and LWE
reductions~\cite{MP13}, and short-distribution module-LWE
reductions~\cite{BJRW23}. Applying the composition above to any such route
also requires a dimension-controlled reduction from $3$-SAT and matching
its actual target law, output task, observation count, and error budget.
A quantum reduction would require a separate quantum-oracle simulation.

\section{Finite-Parameter Substitutions}
\label{uc:parameters}

For each listed parameter set, we substitute its dimension and modulus,
together with the deterministic joint energy bound, into the depth
expression. The resulting formal leading term values are shown in
Table~\ref{uc:finite-parameter-table}. From~\eqref{ui:depth},
\[
 \begin{gathered}
 R^2=nV+3B^2,\qquad U=\sqrt{V+3B^2/n},\\
 D=\ln\frac{\ln(q/\sqrt n)}{1+\ln U},\qquad
 b_D=\frac{n}{2D},\qquad
 b_\infty=\frac{n}{2\ln\ln n}.
 \end{gathered}
\]
Here $nV$ bounds the joint secret and error energy, $B$ bounds each
secret coefficient, and all logarithms are natural. The quantity
$b_D$ is the finite formal substitution into the leading exponent
$n/(2D)$, while $b_\infty$ is the asymptotic comparison obtained by
replacing $D$ with $\ln\ln n$. The displayed nonintegral values of
$V$ are rounded only for presentation. All calculations use the exact
fractions in the diagnostic script.

For ML-KEM, public key generation uses $\mathrm{CBD}(\eta_1)$ for
both secret and error, giving $V=2\eta_1^2$ and $B=\eta_1$
\cite{KyberSpec21,FIPS203}. For FrodoKEM, the support bounds
$H=12,10,6$ give $V=2H^2$ and $B=H$. Here $n$ is the dimension of
one secret column, and all eight columns share the Gaussian list
\cite{FrodoSpec20}. For SABER, $B=\mu/2$ and the rounding step
$q/p=8$ gives $V=B^2+16$, including the rounding error
\cite{SaberSpec}. For ML-KEM and SABER, $n$ is the module rank times
the ring degree. For NTRU LPRime, $n=p$ and a ternary secret of exact
weight $w$ gives $V=1+w/n$ and $B=1$~\cite{NTRUPrimeSpec}. This
fraction is evaluated exactly before computing $D$.

For ML-DSA, $n=\ell m$, $B=\eta$, and
$V=\eta^2+(\eta+2^{d-1})^2$, as in
\eqref{umldsa:energy}. The public key rounding error has coefficient
magnitude at most $2^{d-1}$ and is included together with
$\mathbf s_2$ in the lifted error. The ML-DSA parameter sets use
$m=256$, $d=13$, and $\eta=2,4,2$, respectively~\cite{FIPS204}.

\begin{table}[!htbp]
\centering
\papertablesetup
\setlength{\tabcolsep}{5pt}
\renewcommand{\arraystretch}{1.12}
\caption{Parameter values, joint energy bounds, and formal leading term
values for five linear public families.}
\label{uc:finite-parameter-table}
\begin{tabular}{@{}lrrrrr@{}}
\toprule
Parameter set & $n$ & $q$ & $V$ & $b_D$ & $b_\infty$\\
\midrule
ML-KEM-512 & 512 & 3329 & 18 & 359.1 & 139.8 \\
ML-KEM-768 & 768 & 3329 & 8 & 450.2 & 202.8 \\
ML-KEM-1024 & 1024 & 3329 & 8 & 622.5 & 264.5 \\
\addlinespace
FrodoKEM-640 & 640 & 32768 & 288 & 511.3 & 171.5 \\
FrodoKEM-976 & 976 & 65536 & 200 & 659.6 & 253.0 \\
FrodoKEM-1344 & 1344 & 65536 & 72 & 772.8 & 340.3 \\
\addlinespace
LightSaber & 512 & 8192 & 41 & 354.0 & 139.8 \\
Saber & 768 & 8192 & 32 & 524.0 & 202.8 \\
FireSaber & 1024 & 8192 & 25 & 679.4 & 264.5 \\
\addlinespace
ntrulpr653 & 653 & 4621 & 1.3859 & 218.3 & 174.7 \\
ntrulpr761 & 761 & 4591 & 1.3285 & 254.0 & 201.1 \\
ntrulpr857 & 857 & 5167 & 1.3279 & 283.8 & 224.3 \\
ntrulpr953 & 953 & 6343 & 1.3620 & 311.9 & 247.4 \\
ntrulpr1013 & 1013 & 7177 & 1.3870 & 329.5 & 261.8 \\
ntrulpr1277 & 1277 & 7879 & 1.3359 & 412.0 & 324.5 \\
\addlinespace
ML-DSA-44 & 1024 & 8380417 & 16793608 & 1754.6 & 264.5 \\
ML-DSA-65 & 1280 & 8380417 & 16810016 & 2263.3 & 325.2 \\
ML-DSA-87 & 1792 & 8380417 & 16793608 & 3329.4 & 444.9 \\
\bottomrule
\end{tabular}
\end{table}

The modulus and energy jointly determine the available depth. The quantity
$b_D$ records this finite comparison, while $b_\infty$ records the common
asymptotic scale.

\begin{remark}
\label{uc:mldsa-depth}
The ML-DSA family in Definition~\ref{umldsa:family} extends the standard
compressed public key relation by increasing the ring degree and
modulus at fixed module ranks, with the secret coefficient bound
$\eta$ and rounding step $\Delta$ bounded or polylogarithmic.
In particular, $\eta$ and the number $d$ of discarded bits may stay
fixed. This choice separates dimension growth from the coefficient
scales contributing to the joint energy.

The depth can be written exactly as
\[
 D=\ln\ln n+
 \ln\!\left(\frac{\ln(q/\sqrt n)}{\ln n}\right)
 -\ln(1+\ln U).
\]
The modulus bounds $q\ge2m+1$ and $q\le n^{C_q}$, with
$n=\ell m$ and fixed $\ell$, give
$\ln(q/\sqrt n)=\Theta(\ln n)$. Hence the middle term is $O(1)$,
while $U=(\ln n)^{O(1)}$ makes the last term
$O(\ln\ln\ln n)$. Thus $D\sim\ln\ln n$.
At the listed ML-DSA dimensions, the rounding contribution
$2^{d-1}=4096$ dominates the deterministic energy bound.
For ML-DSA-44, $U\approx4098$ gives $D\approx0.292$, compared
with $\ln\ln n\approx1.936$. Consequently,
$b_D/b_\infty=\ln\ln n/D\approx6.63$.
The finite correction accounts for the separation between the
two formal exponents in the ML-DSA rows.
\end{remark}

\end{document}